\documentclass[review,onefignum,onetabnum]{article}

\usepackage[linesnumbered,ruled,vlined]{algorithm2e}
\usepackage{amsmath,amssymb,amsfonts,amsthm}
\usepackage{mathtools}
\usepackage{xcolor}
\usepackage{bm}
\usepackage{braket}
\usepackage{bbm}
\usepackage{dsfont}
\usepackage{mathrsfs}
\usepackage{booktabs}
\usepackage{hyperref}
\usepackage{geometry}
\hypersetup{colorlinks=true,linkcolor=blue,citecolor=blue,urlcolor=blue}

\newtheorem{definition}{Definition}
\newtheorem{proposition}{Proposition}
\newtheorem{theorem}{Theorem}
\newtheorem{lemma}{Lemma}
\newtheorem{corollary}{Corollary}
\newtheorem{assumption}{Assumption}

\definecolor{darkblue}{rgb}{0.0, 0.0, 0.55}
\newcommand{\R}{\mathbb{R}}
\newcommand{\E}{\mathbb{E}}
\newcommand{\Pee}{\mathbb{P}}
\newcommand{\F}{\mathcal{F}}

\newcommand{\cS}{\mathcal{S}}
\newcommand{\cN}{\mathcal{N}}

\newcommand{\Fcal}{\mathcal{F}}
\newcommand{\Xbf}{\mathbf{X}}
\newcommand{\Cbf}{\mathbf{C}}
\newcommand{\Zbf}{\mathbf{Z}}
\newcommand{\ubf}{\mathbf{u}}

\newcommand{\ebf}{\mathbf{e}}

\renewcommand{\tilde}{\widetilde}

\usepackage{tikz}
\usepackage{adjustbox} 
\usetikzlibrary{shapes.geometric, arrows.meta, positioning, calc, backgrounds, fit}

\title{\textbf{\color{darkblue}Mitigating Polycentric Conflict-Trap Risk in Mali via Intergenerational Volterra Mean-Field-Type Games}}

\author{Hamidou Tembine\thanks{Department of EECS, University of  Quebec  in  Trois-Rivieres, Canada} \thanks{Learning \& Game Theory Laboratory, TIMADIE } \thanks{ Contact: tembine@ieee.org} }

\begin{document}

\maketitle

\begin{abstract}

Persistent instability in Mali and neighboring countries is not a temporary security crisis but a self-reproducing conflict system sustained by intergenerational grievances, climatic shocks, predatory resource extraction, and decentralized war economies. We develop a novel intergenerational Volterra mean-field-type game framework in which historical memory, inherited distrust, revenge incentives, and evolving institutional responses jointly shape future conflict dynamics. The analysis reveals how war entrepreneurs recycle illicit resource revenues through global financial networks, while blunt regulatory interventions can inadvertently accelerate recruitment into armed groups. We design optimal  policies that alter the system's long-run geometry, severing the financial and psychological feedback loops of violence. The resulting framework provides a predictive and actionable architecture for durable peace, civilian protection, and multi-generational structural stabilization in the Sahel and beyond.

%%
%%Chronic security crises and systemic fragility rarely emerge from a single event. They are often the result of grievances, mistrust, exclusion, and cycles of retaliation that accumulate and persist across generations. This paper develops a new framework to understand how such dynamics can unfold in Mali and how they can be prevented before they become self-sustaining. Rather than focusing on isolated communities or average population behavior, the approach examines how individuals, groups, institutions, and their social, economic, political, and environmental surroundings evolve together over time. A key innovation is the incorporation of historical memory: past violence, inherited grievances, and intergenerational transmission of distrust are treated as active forces shaping future decisions. The framework identifies the conditions under which conflict can persist across generations and designs adaptive institutional incentives that weaken cycles of revenge. The results provide a  foundation for long-term peacebuilding strategies in Mali and offer insights relevant to other regions facing chronic risks of social fragmentation and polycentric conflict. We design optimal institutional policies that alter the spectral geometry of the system, severing the financial and psychological loops of the war economy to guarantee long-term structural stabilization.

\end{abstract}

{\bf Keywords: }
Intergenerational mean-field-type games, polycentric conflict-trap risk, joint probability measures, fractional Brownian motion, Rosenblatt process, Gauss-Volterra process, regime switching, expectiles, trans-generational boundary conditions, optimal institutional policy.

{\bf AMS: }
91A15, 91A16, 93E20, 49N80, 45D05, 60H30, 60G22.

\section{Introduction}
\label{sec:intro}
The most lethal form of organized violence in contemporary Mali is not a war in the classical sense. It is a persistent, structurally embedded system of conflict that operates beyond the boundaries of traditional warfare categories, a self-reproducing apparatus of violence, predation, and capital accumulation that has, for more than a decade, turned the villages, pastoral corridors, and artisanal mining zones of Mali and the wider Sahel into a laboratory of permanent instability. A {\it civil war} is classically defined as a sustained, organized armed confrontation between a central government and one or more non-state actors, fought within the territorial boundaries of a sovereign state. This definition, while analytically useful for historical cases such as the American Civil War or the Spanish Civil War, presumes that conflict is primarily a \emph{state-versus-opposition} phenomenon with identifiable temporal and spatial boundaries. Field observations from Mali and the broader Sahel region indicate that this assumption is no longer structurally valid. The central question this paper addresses is: \emph{Why does violence persist, mutate, and reproduce itself across generations long after its original political and territorial causes have dissipated?}

\medskip
\noindent\textbf{The Malian Theatre.}
Since the crisis onset in 2012, Mali has evolved into a multi-layered and spatially diffused conflict system. Field reports from central Mali (including the Niger Inland Delta, the Mopti region, and the Ségou-Niono-Koutiala axis), western corridors (Kayes-Diéma-Nioro), and transborder zones with Burkina Faso, Mauritania and Niger indicate that violence is no longer concentrated on a single front but distributed across heterogeneous localities. Qualitative field observations consistently report three structural features:
\begin{enumerate}
\item The fragmentation of armed authority into overlapping local, jihadist, bandit, and criminal governance structures.
\item The normalization of predatory economic extraction embedded within rural livelihoods.
\item The persistence of insecurity independent of formal ceasefires or peace agreements.
\end{enumerate}
These observations suggest that the relevant object of analysis is not a single war but a \emph{self-reproducing conflict system}.

\medskip
\noindent\textbf{From Civil War to Polycentric Intergenerational Conflict Systems.}
We introduce a refined analytical construct that captures the observed reality.
\begin{definition}[Polycentric Intergenerational Conflict System]
Let $\mathcal{P} = \{P_1,\dots,P_m\}$ denote a finite set of geographically and institutionally decentralized power centers. A \emph{polycentric intergenerational conflict system} is a dynamical socio-political system in which:
\begin{enumerate}
\item violence is generated by multiple non-coordinated centers $P_i$,
\item economic incentives are endogenously coupled to instability,
\item institutional capacity is fragmented across space,
\item environmental and climatic stressors act as exogenous amplifiers,
\item and grievances, assets, and organizational knowledge are transmitted across generations.
\end{enumerate}
The system evolves such that its state variable is not merely territorial control but the joint distribution of violence, incentives, and memory across space and time.
\end{definition}
Under this definition, the Malian conflict is not a discrete war but a \emph{dynamic equilibrium over a network of interacting violent and economic subsystems}. The initial phase of the 2012 crisis was characterized by a northern separatist insurgency. However, subsequent field evidence indicates a structural displacement of violence toward central and western Mali and adjacent Sahelian regions. Let $\Omega_t \subset \mathbb{R}^2$ denote the spatial support of active violence at time $t$. Empirical observation shows that $\Omega_t$ has expanded and fragmented over time, rather than contracting after military interventions or peace accords. This shift reflects a transition from 
\[
\text{Territorial insurgency} \quad \longrightarrow \quad \text{Distributed conflict ecology}
\]
in which local incentives, cross-border networks, and economic predation dominate strategic dynamics.

% =====================================================================
% STRUCTURAL COMPARISON: CIVIL WAR VS. POLYCENTRIC INTERGENERATIONAL CONFLICT SYSTEM
% =====================================================================

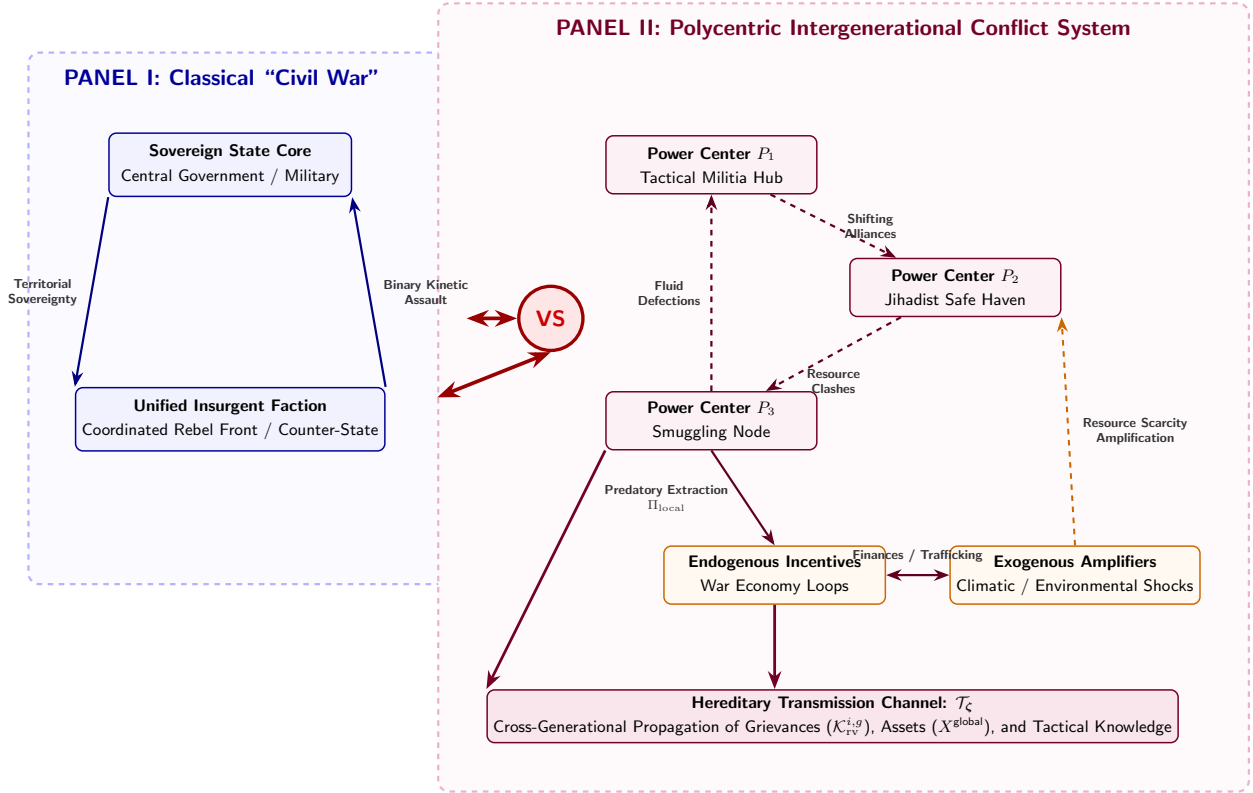
\begin{figure}[htbp]
\centering
\begin{adjustbox}{width=\textwidth, center}
\begin{tikzpicture}[
    font=\sffamily\small,
    >=Stealth,
    node distance=1.4cm and 1.3cm,
    % Box Styles (Professional, Modern Color Palette)
    statebox/.style={rectangle, draw=gray!70, fill=white, rounded corners=4pt, minimum width=4cm, minimum height=1.1cm, align=center, thick},
    classicbox/.style={rectangle, draw=blue!60!black, fill=blue!5, rounded corners=4pt, minimum width=4.6cm, minimum height=1.2cm, align=center, thick},
    polybox/.style={rectangle, draw=purple!60!black, fill=purple!5, rounded corners=4pt, minimum width=4.0cm, minimum height=1.1cm, align=center, thick},
    accentbox/.style={rectangle, draw=orange!80!black, fill=orange!5, rounded corners=4pt, minimum width=4.2cm, minimum height=1.1cm, align=center, thick},
    % Distinct Panel Background Design Styles
    panelclassic/.style={rectangle, draw=blue!30, fill=blue!2!white, dashed, rounded corners=6pt, inner sep=26pt, line width=1.2pt},
    panelpoly/.style={rectangle, draw=purple!30, fill=purple!2!white, dashed, rounded corners=6pt, inner sep=26pt, line width=1.2pt},
    % Edge Styles
    classicflow/.style={->, thick, color=blue!50!black, line width=1.2pt},
    polyflow/.style={->, thick, color=purple!50!black, line width=1.2pt},
    shockflow/.style={->, dashed, color=orange!80!black, line width=1.1pt},
    vsflow/.style={<->, line width=2pt, color=red!60!black},
    % Label Styles
    lbl/.style={align=center, font=\sffamily\scriptsize\bfseries, color=black!80}
]

% =====================================================================
% LEFT COLUMN: CLASSICAL "CIVIL WAR" PARADIGM
% =====================================================================
\node[classicbox] (sovereign) {\textbf{Sovereign State Core}\\[2pt] Central Government / Military};
\node[classicbox, below=3.6cm of sovereign] (insurgent) {\textbf{Unified Insurgent Faction}\\[2pt] Coordinated Rebel Front / Counter-State};

% Internal Centralized Dynamics
\draw[classicflow] (sovereign.south west) -- node[left, lbl, xshift=-4pt] {Territorial\\Sovereignty} (insurgent.north west);
\draw[classicflow] (insurgent.north east) -- node[right, lbl, xshift=4pt] {Binary Kinetic\\Assault} (sovereign.south east);

% Left Bounding Coordinates
\coordinate (left_pad_top) at ($(sovereign.north west)+(-0.6cm,0.6cm)$);
\coordinate (left_pad_bot) at ($(insurgent.south east)+(0.6cm,-1.6cm)$);

% Left Container
\begin{scope}[on background layer]
    \node[panelclassic, fit={(left_pad_top) (left_pad_bot)}, 
          label={[anchor=north, font=\sffamily\large\bfseries\color{blue!60!black}, yshift=-6pt]north:PANEL I: Classical ``Civil War''  \ \  \ \  \ \  }] (box_classic) {};
\end{scope}

% =====================================================================
% MIDDLE SEPARATION: CENTRAL CRITIQUE BIFURCATION
% =====================================================================
\path (box_classic.east) -- ++(1.6cm,0) node[coordinate] (vs_mid) {};
%\node[draw=red!60!black, fill=red!10, circle, inner sep=6pt, line width=1.8pt, font=\sffamily\large\bfseries\color_red!80!black] (vs_node) at (vs_mid) {VS};
\node[draw=red!60!black, fill=red!10, circle, inner sep=6pt, line width=1.8pt, font=\sffamily\large\bfseries\color{red!80!black}] (vs_node) at (vs_mid) {VS};
% =====================================================================
% RIGHT COLUMN: POLYCENTRIC INTERGENERATIONAL CONFLICT SYSTEM
% =====================================================================
% Decentralized Power Centers P_i (Using structured spacing keys instead of manual compound shifts)
\node[polybox, right=4.8cm of sovereign] (p1) {\textbf{Power Center $P_1$}\\[2pt] Tactical Militia Hub};
\node[polybox, below right=1.2cm and 0.6cm of p1] (p2) {\textbf{Power Center $P_2$}\\[2pt] Jihadist Safe Haven};
\node[polybox, below left=1.4cm and 0.6cm of p2] (p3) {\textbf{Power Center $P_3$}\\[2pt] Smuggling Node};

% Decentralized, Non-Coordinated Interactions
\draw[polyflow, dashed] (p1) -- node[right, lbl, pos=0.5, xshift=2pt] {Shifting\\Alliances} (p2);
\draw[polyflow, dashed] (p2) -- node[below, lbl, pos=0.5, yshift=-4pt] {Resource\\Clashes} (p3);
\draw[polyflow, dashed] (p3) -- node[left, lbl, pos=0.5, xshift=-2pt] {Fluid\\Defections} (p1);

% Economic and Environmental Coupling Strata
\node[accentbox, below=1.8cm of p3, xshift=1.2cm] (incentives) {\textbf{Endogenous Incentives}\\[2pt] War Economy Loops};
\node[accentbox, right=1.2cm of incentives] (stressors) {\textbf{Exogenous Amplifiers}\\[2pt] Climatic / Environmental Shocks};

% Hereditary Layer
\node[draw=purple!60!black, fill=purple!10, rectangle, rounded corners=4pt, minimum width=9.8cm, minimum height=1.0cm, below=1.6cm of incentives, xshift=1.1cm, align=center, thick] (transmission) {
    \textbf{Hereditary Transmission Channel: $\mathcal{T}_{\boldsymbol{\zeta}}$}\\[2pt]
    Cross-Generational Propagation of Grievances ($\mathcal{K}_{\mathrm{rv}}^{i,g}$), Assets ($X^{\text{global}}$), and Tactical Knowledge
};

% Connecting vectors inside the Polycentric Strata
\draw[polyflow] (p3.south) -- node[left, lbl, xshift=-4pt] {Predatory Extraction\\$\Pi_{\mathrm{local}}$} (incentives.north);
\draw[shockflow] (stressors.north) -- node[right, lbl, xshift=4pt] {Resource Scarcity\\Amplification} (p2.south east);

% Fixed the syntax parsing breakdown below
\draw[polyflow, <->] (incentives.east) -- node[above, lbl, yshift=3pt] {Finances / Trafficking} (stressors.west);

% Flowing down to Joint Distribution/Memory Space
\draw[polyflow, line width=1.5pt, color=purple!60!black] (incentives.south) -- (incentives.south |- transmission.north);
\draw[polyflow, line width=1.5pt, color=purple!60!black] (p3.south west) -- (transmission.north west);

% Right Bounding Coordinates
\coordinate (right_pad_top) at ($(p1.north |- box_classic.north)$);
\coordinate (right_pad_bot) at ($(transmission.south |- box_classic.south)$);

% Right Container
\begin{scope}[on background layer]
    \node[panelpoly, fit={(p1) (p2) (p3) (transmission) (stressors) (right_pad_top) (right_pad_bot)}, 
          label={[anchor=north, font=\sffamily\large\bfseries\color{purple!60!black}, yshift=-6pt]north:PANEL II: Polycentric Intergenerational Conflict System}] (box_poly) {};
\end{scope}

% Unified Horizontal Connector Paths
\draw[vsflow] (box_classic.east) -- (vs_node.west);
\draw[vsflow] (box_poly.west) -- (vs_node.south);

\end{tikzpicture}
\end{adjustbox}
\caption{The Paradigm Shift from Binary Conflict to Polycentric Systems. Panel I demonstrates the classical, symmetric ``civil war'' abstraction where conflict is constrained to a zero-sum, state-centric binary structure. Panel II maps the actual architecture defined in Definition 1, showing the non-coordinated interactions of decentralized power centers ($P_i$), tightly coupled with endogenous economic incentives, climatic shocks, and multi-generational hereditary transmission lines.}
\label{fig:paradigm_bifurcation}
\end{figure}

The theoretical limitations of analyzing modern Sahelian fragility through traditional lenses are formalized in the structural comparison illustrated in Figure~\ref{fig:paradigm_bifurcation}. Classical political science and game-theoretic literature remain wedded to the binary architecture mapped in Panel I. This paradigm represents conflict as a cohesive, macro-centric \emph{civil war}, where a centralized state core and a single, organized insurgent front engage in a symmetric, zero-sum struggle over territorial sovereignty. While  tractable under standard differential assumptions, this binary configuration assumes complete internal cohesion, institutional stability, and stationary preferences.
As contextualized in Panel II of Figure~\ref{fig:paradigm_bifurcation}, the real-world operational space of Mali is accurately captured only by the primitives of a \emph{Polycentric Intergenerational Conflict System}. Here, violence is explicitly de-centered, emerging from multiple non-coordinated power centers $\mathcal{P} = \{P_1, P_2, P_3\}$ that continually engage in fluid coalitional defection, localized tactical optimization, and resource competition. Crucially, the system moves beyond purely political payoffs; as traced by the network edges, kinetic activity is endogenously coupled to a highly informal war economy, where resource extraction profit functions are continuously amplified by exogenous environmental and climatic shocks (such as droughts and structural soil degradation).
The defining boundary between these two paradigms is the temporal transmission layer situated at the base of Figure~\ref{fig:paradigm_bifurcation}. In a standard civil war model, the termination of a single lifecycle or a signed peace accord resets the state space. In the polycentric system, the terminal law of the population-environment configuration undergoes a non-local push-forward across generational epochs via the operator $\mathcal{T}_{\boldsymbol{\zeta}}$. 
Trauma, offshore capital accumulation, and organizational tactical knowledge act as inherited state variables that harden the singular Volterra revenge operator $\mathcal{K}_{\mathrm{rv}}^{i,g}$ for the subsequent cohort. Strategic optimization is therefore never an isolated, stationary game. The system evolves such that the fundamental strategic variable is not territorial control lines on a map, but the joint distribution of violence, incentives, and memory across space and time on the infinite-dimensional Wasserstein manifold.

\medskip
\noindent\textbf{The Conflict Economy and the War Entrepreneur.}
We formalize the economic structure underlying persistence.
Let $\mathcal{P}_2(\mathbb{R}^d)$ denote the Wasserstein space of probability measures with finite second moment. For each agent $i \in \mathcal{I}$, define a distribution-dependent expected utility functional
\[
F_i(\mu) := e_\tau\!\big[U_i^\mu(t)\big], \qquad \mu \in \mathcal{P}_2(\mathbb{R}^d),
\]
where $U_i^\mu(t)$ is induced by the state distribution $\mu$ and $e_\tau[\cdot]$ denotes the $\tau$-expectile operator.

\medskip

\textbf{Variational structure on $\mathcal{P}_2$.}
Let $\mu \in \mathcal{P}_2(\mathbb{R}^d)$ and let $\eta$ be a signed finite measure satisfying
$
\eta(\mathbb{R}^d)=0,
\
\int_{\mathbb{R}^d}\|x\|^2\,|\eta|(dx)<\infty,
$
so that $\mu_\varepsilon := \mu + \varepsilon \eta \in \mathcal{P}_2(\mathbb{R}^d)$ for sufficiently small $\varepsilon$. The functional $F_i$ is said to be Gâteaux differentiable at $\mu$ if there exists a measurable representative $\frac{\delta F_i}{\delta \mu}(\mu)(x)$ such that
\[
\lim_{\varepsilon \to 0}
\frac{F_i(\mu+\varepsilon \eta)-F_i(\mu)}{\varepsilon}
=
\int_{\mathbb{R}^d}
\frac{\delta F_i}{\delta \mu}(\mu)(x)\,\eta(dx),
\]
for all admissible perturbations $\eta$. The derivative is defined up to additive constants, i.e.
$
\frac{\delta F_i}{\delta \mu} \sim \frac{\delta F_i}{\delta \mu} + c,\quad c\in\mathbb{R},
$
reflecting invariance under mass-preserving perturbations.

\medskip

\begin{definition}[Conflict economy]

\textbf{Conflict-inducing agents.}
A conflict economy is a stochastic system
$
\mathcal{E} := \big(X_t,\mu_t,\{U_i(t)\}_{i\in\mathcal{I}}\big)_{t\in[0,T]}
$
such that there exists a non-empty subset $\mathcal{I}_c \subseteq \mathcal{I}$ with the property that, for each $i \in \mathcal{I}_c$, the functional $F_i$ is Gâteaux differentiable at $\mu_t$ and satisfies the strict directional positivity condition
$
\int_{\mathbb{R}^d}
\frac{\delta F_i}{\delta \mu}(\mu_t)(x)\,\eta(dx) \;>\; 0,
$
for every non-zero admissible perturbation $\eta$ with $\eta(\mathbb{R}^d)=0$, for which $\mu_t+\varepsilon\eta \in \mathcal{P}_2(\mathbb{R}^d)$ for sufficiently small $\varepsilon>0$, and for almost every $t \in [0,T]$.
\end{definition}
\medskip

This condition defines a class of agents whose expected upper-tail utility is strictly increasing along every non-trivial mass-preserving deformation of the state distribution in a prescribed admissible cone. Equivalently, their marginal valuation of the distributional geometry of states is everywhere aligned with directions that preserve total mass but reallocate probability across the state space. Such agents therefore do not merely respond to realizations of the state but systematically benefit from endogenous reshaping of the law $\mu_t$ itself. Conflict is thus not introduced as an exogenous shock process, but emerges as an intrinsic property of the preference geometry over probability measures: a subset of agents exhibits strictly positive directional incentives toward distributional reconfiguration, making instability a structurally rewarded direction in Wasserstein space and thereby constituting the economic mechanism underlying the {\it war entrepreneur.}

\begin{lemma}[Gâteaux derivative of the expectile functional]
Let $\mu \in \mathcal{P}_2(\mathbb{R})$ and let $m := e_\tau(\mu)$ denote the $\tau$-expectile, defined as the unique solution to
\[
\int_{\mathbb{R}} w_\tau(x-m)(x-m)\,\mu(dx)=0,
\quad
w_\tau(z)=\tau \mathbf{1}_{\{z\ge 0\}}+(1-\tau)\mathbf{1}_{\{z<0\}}.
\]
Let $\eta$ be a signed finite measure such that
$
\eta(\mathbb{R})=0,
\ 
\int_{\mathbb{R}} x^2\,|\eta|(dx)<\infty,
$
and define $\mu_\varepsilon := \mu + \varepsilon \eta$ for sufficiently small $\varepsilon$.

Then the map $\mu \mapsto e_\tau(\mu)$ is Gâteaux differentiable at $\mu$ and satisfies
\[
\left.\frac{d}{d\varepsilon} e_\tau(\mu_\varepsilon)\right|_{\varepsilon=0}
=
\frac{
\int_{\mathbb{R}} w_\tau(x-m)(x-m)\,\eta(dx)
}{
\tau \mu([m,\infty)) + (1-\tau)\mu((-\infty,m))
}.
\]

Equivalently, the functional derivative kernel is given by
\[
\frac{\delta e_\tau}{\delta \mu}(x)
=
\frac{w_\tau(x-m)(x-m)}{\tau \mu([m,\infty)) + (1-\tau)\mu((-\infty,m))},
\quad m=e_\tau(\mu).
\]
\end{lemma}

\begin{proof}
The result follows by differentiating the first-order optimality condition defining the expectile under the perturbation $\mu_\varepsilon$, applying the implicit function theorem in Banach spaces, and exchanging differentiation and integration under square-integrability of $\eta$. The denominator arises from the derivative with respect to the endogenous threshold $m$.
\end{proof}

\begin{corollary}[Explicit conflict condition via expectile sensitivity]
Let $\mathcal{E} := (X_t,\mu_t,\{U_i(t)\}_{i\in\mathcal{I}})_{t\in[0,T]}$ be a stochastic system in which
$
F_i(\mu) := e_\tau\!\big[U_i^\mu(t)\big]
$
is Gâteaux differentiable for each $i$.
Then for every admissible perturbation $\eta$ with $\eta(\mathbb{R})=0$, the directional derivative admits the explicit representation
\[
\int_{\mathbb{R}}
\frac{\delta F_i}{\delta \mu_t}(x)\,\eta(dx)
=
\frac{
\int_{\mathbb{R}}
w_\tau\!\big(U_i^\mu(x)-m_i(\mu_t)\big)
\big(U_i^\mu(x)-m_i(\mu_t)\big)\,
\eta(dx)
}{
\tau \mu_t([m_i,\infty)) + (1-\tau)\mu_t((-\infty,m_i))
},
\]
where $m_i(\mu_t)=e_\tau[U_i^\mu(t)]$.
\end{corollary}

A conflict economy is a stochastic system
$
\mathcal{E} := (X_t,\mu_t,\{U_i(t)\}_{i\in\mathcal{I}})_{t\in[0,T]}
$
such that there exists a non-empty subset $\mathcal{I}_c \subseteq \mathcal{I}$ satisfying:
for each $i \in \mathcal{I}_c$ and for almost every $t \in [0,T]$,
\[
\frac{
\int_{\mathbb{R}}
w_\tau\!\big(U_i^\mu(x)-m_i(\mu_t)\big)
\big(U_i^\mu(x)-m_i(\mu_t)\big)\,
\eta(dx)
}{
\tau \mu_t([m_i,\infty)) + (1-\tau)\mu_t((-\infty,m_i))
}
> 0,
\]
for every non-zero admissible perturbation $\eta$ with $\eta(\mathbb{R})=0$ and $\mu_t+\varepsilon\eta \in \mathcal{P}_2(\mathbb{R})$ for sufficiently small $\varepsilon>0$.
Equivalently, agents in $\mathcal{I}_c$ strictly increase their expected upper-tail (expectile-based) utility under every admissible infinitesimal mass-preserving redistribution of the state distribution, making distributional deformation itself a strictly productive direction in the induced welfare geometry.

In Mali, field-level observations confirm that livestock markets, artisanal gold extraction corridors, protection rackets, and smuggling routes 
operate as coupled subsystems of this conflict economy. At the heart of this war economy stands the \emph{war entrepreneur}.

\begin{definition}[War Entrepreneur]
A \emph{war entrepreneur} is an economic actor who systematically exploits institutional collapse, territorial fragmentation, and insecurity to generate and accumulate wealth through the orchestration of violent and illicit economic activities across local and transnational networks.
\end{definition}

These individuals and the networks they command control the principal nodes of a transnational illicit economy: artisanal gold mines in the Kédougou-Kénieba corridor, smuggling routes across the porous borders with Mauritania, Algeria and Burkina Faso, protection rackets imposed on pastoralist communities, drug and arms trafficking networks stretching from the Gulf of Guinea to the Mediterranean, and the increasingly lucrative kidnapping‑for‑ransom industry. Their operations are financed by, and in turn sustain, a global financial infrastructure of offshore accounts, shell companies, and money‑laundering channels that convert the proceeds of chaos into clean, inheritable assets.

\medskip
\noindent\textbf{Micro‑Level Empirical Structure: Village Predation.}
Yet the war economy is not confined to transnational flows and high‑value commodities. It rests, at its base, on the systematic, intimate, and devastating predation of rural life. In the villages of central, western, southern and northern Mali (the plains of the Niger Inland Delta, the cliffs of Dogon country, the pastoral corridors of the Gourma, the areas of Niono, Nioro, Diema), armed groups descend upon farming communities with methodical regularity. Let $V_k$ denote a village‑level economic unit with production set $\mathcal{Y}_k = \{\text{grain}, \text{livestock}, \text{goods}, \text{infrastructure}\}$. Empirical reports indicate that armed incursions follow a structured extraction sequence:
\begin{enumerate}
\item Seizure of stored agricultural output (millet, sorghum, rice), the family's entire annual sustenance.
\item Appropriation of livestock assets (cattle, goats, sheep, camels, donkeys).
\item Looting of retail inventories and consumables (cooking oil, sugar, tea, batteries, soap, pharmaceuticals).
\item Destruction of productive infrastructure: granaries are burned, wells are poisoned or destroyed, fruit trees are cut down, agricultural tools and motorcycles are confiscated.
\item Forced displacement of labour and households, transforming a self‑sustaining production unit into a structurally non‑viable settlement.
\end{enumerate}

This process induces a transition $V_k^{\text{productive}} \longrightarrow V_k^{\text{depleted}} \longrightarrow V_k^{\text{displaced}}$. It is not random destruction but a structured mechanism of economic reconfiguration. The stolen grain re‑enters regional trade circuits; stolen livestock is rapidly absorbed into cross‑border pastoral and commercial flows; and the displaced populations become cheap labour for artisanal mines, recruits for armed groups, or human cargo for trafficking networks. For the war entrepreneur, a stable, prosperous village is unprofitable; a burned village yields immediate loot, drives down wages, and eliminates local resistance. Thus, local devastation is not a side effect of war:  it is a mechanism of economic reproduction within the conflict system.

\medskip
\noindent\textbf{Intergenerational Persistence Mechanism.}
Let $g \in \{0,\dots,G\}$ index generations. Define state variables $X^g = (\text{wealth}, \text{trauma}, \text{skills}, \text{network position})$. The system exhibits intergenerational transmission governed by capital inheritance (licit and illicit), transmission of grievance and memory structures, continuity of armed networks, and reproduction of survival strategies in fragile environments. Instability is not temporally localized but dynamically inherited. The central empirical paradox motivating this work is the following:

\begin{quote}
Although the Malian crisis originated as a geographically localized insurgency in the north, the majority of contemporary violence is now generated by a distributed system of armed, criminal, and communal actors operating across central and western Mali and neighboring 
Sahelian states, many of which did not participate in the original rebellion.
\end{quote}

This indicates that the system has exceeded its initial causal structure and entered a regime of self‑sustaining dynamics.

\medskip
\noindent\textbf{Limitations of Existing Models}
Standard conflict models rely on representative‑agent assumptions, conditional independence, mean‑field limits, or empirical decoupling. These simplifications eliminate precisely the features that dominate observed Sahelian dynamics: higher‑order dependence structures, endogenous coalition formation, long‑memory effects, asymmetric risk perception, institutional feedback loops, and intergenerational transmission mechanisms. The Malian theatre exemplifies these complexities. The affected population around the extended area, numbering approximately $N = 80\times 10^{6}$ when considering the cross‑border Sahelian zone, consists of heterogeneous actors who continuously adjust their radicalisation levels, asset portfolios, migration decisions, and coalitional affiliations. These choices are not independent. They are jointly shaped by gold prices, border porosity, rainfall patterns, the presence of peacekeeping forces, the effectiveness of anti‑money‑laundering regimes, and the inherited memory of past massacres. The relevant object is not an empirical distribution of independent trajectories, but the exact joint probability law of the entire population‑environment system.
To construct an analytically  framework capable of modeling risk mitigation and structural stabilization in the Malian theatre, we must deliberately break away from conventional game-theoretic paradigms. While classical Mean-Field Games (MFGs) and their various extensions have popularized the study of large-scale strategic interactions, they rely on asymptotic abstractions that {\it fundamentally misread} the operational, anthropological, and structural realities of the Sahelian crisis. 
Below, we provide critique demonstrating why standard mean-field and macro-game approximations are structurally inappropriate for the Malian context, and establish why an \emph{Intergenerational Volterra  MFTG} \cite{basar2026a,basar2026b,basar2026c,basar2026d,basar2026e,basar2026f} framework on the Wasserstein manifold is uniquely required.

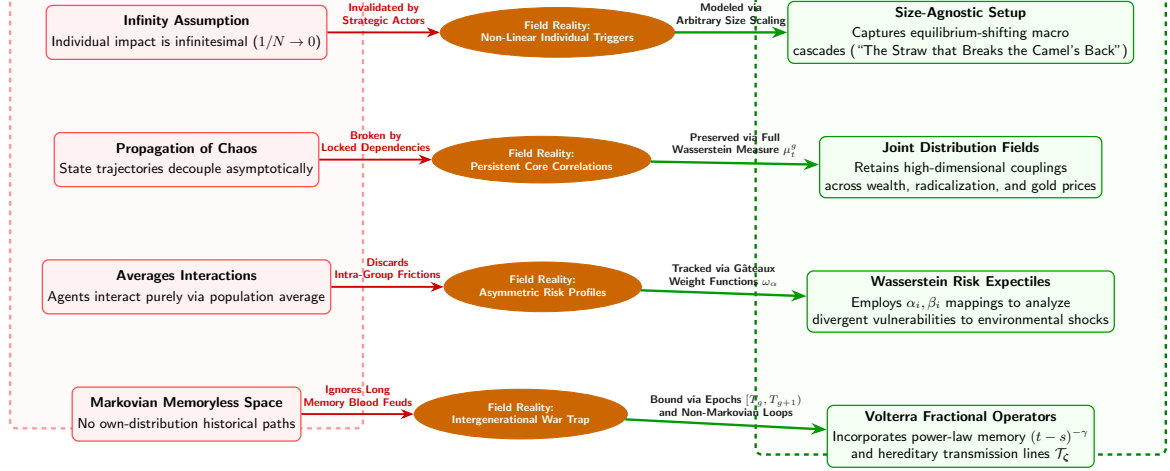
\begin{figure}[htbp]
\centering
\resizebox{\textwidth}{!}{%
\begin{tikzpicture}[
    font=\sffamily\small,
    >=Stealth,
    node distance=1.6cm and 1.4cm,
    % Core Box Styles
    failbox/.style={rectangle, draw=red!70, fill=red!5, rounded corners, minimum width=4.8cm, minimum height=1.2cm, align=center, thick},
    passbox/.style={rectangle, draw=green!60!black, fill=green!5, rounded corners, minimum width=4.8cm, minimum height=1.2cm, align=center, thick},
    mali_obs/.style={ellipse, draw=orange!80!black, fill=orange!80!black, text=white, font=\sffamily\scriptsize\bfseries, inner sep=4pt, align=center},
    % Container Panel Styles
    panel_mfg/.style={rectangle, draw=red!40, fill=red!2, dashed, rounded corners, inner sep=20pt, line width=1.5pt},
    panel_mftg/.style={rectangle, draw=green!50!black, fill=green!2, dashed, rounded corners, inner sep=20pt, line width=1.5pt},
    % Connector Styles
    critiqueflow/.style={->, thick, color=red!80!black, line width=1.2pt},
    solutionflow/.style={->, thick, color=green!60!black, line width=1.5pt},
    % Text/Label Styles
    lbl/.style={align=center, font=\sffamily\scriptsize\bfseries, color=black!90},
    critlbl/.style={align=center, font=\sffamily\scriptsize\bfseries, color=red!80!black}
]

% =====================================================================
% PANEL 1: WHY MEAN-FIELD GAMES (MFG) FAIL IN MALI
% =====================================================================
\node[failbox] (continuum) {\textbf{Infinity Assumption}\\[2pt] Individual impact is infinitesimal ($1/N \to 0$)};

\node[failbox, below=of continuum] (chaos) {\textbf{Propagation of Chaos}\\[2pt] State trajectories decouple asymptotically};

\node[failbox, below=of chaos] (average) {\textbf{Averages  Interactions}\\[2pt] Agents interact purely via population average};

\node[failbox, below=of average] (markov) {\textbf{Markovian Memoryless Space}\\[2pt] No own-distribution historical paths};

% Background Layer for MFG Panel
\begin{scope}[on background layer]
    \node[panel_mfg, yshift=1cm,fit={(continuum) (chaos) (average) (markov)}, 
          label={[anchor=north, font=\sffamily\Large\bfseries\color{red!80!black}, yshift=-4pt]north: Inadequacy of Classical Mean-Field Games (MFG)}] (left_panel) {};
\end{scope}

% =====================================================================
% CENTRAL AXIS: EMPIRICAL MALIAN FIELD OBSERVATIONS (CRITICAL ANCHORS)
% =====================================================================
% Symmetrically aligned middle column using coordinates relative to MFG nodes
\node[mali_obs, right=2.2cm of continuum, xshift=0.3cm] (o1) {Field Reality:\\Non-Linear Individual Triggers};
\node[mali_obs, right=2.2cm of chaos, xshift=0.3cm] (o2) {Field Reality:\\Persistent Core Correlations};
\node[mali_obs, right=2.2cm of average, xshift=0.3cm] (o3) {Field Reality:\\Asymmetric Risk Profiles};
\node[mali_obs, right=2.2cm of markov, xshift=0.3cm] (o4) {Field Reality:\\Intergenerational War Trap};

% Red critique flows highlighting structural invalidity
\draw[critiqueflow] (continuum) -- node[above, critlbl] {Invalidated by\\Strategic Actors} (o1);
\draw[critiqueflow] (chaos) -- node[above, critlbl] {Broken by\\Locked Dependencies} (o2);
\draw[critiqueflow] (average) -- node[above, critlbl] {Discards\\Intra-Group Frictions} (o3);
\draw[critiqueflow] (markov) -- node[above, critlbl] {Ignores Long\\Memory Blood Feuds} (o4);

% =====================================================================
% PANEL 2: WHY MEAN-FIELD-TYPE GAMES (MFTG) ARE APPROPRIATE
% =====================================================================
\node[passbox, right=2.2cm of o1, xshift=0.3cm] (mftg_size) {\textbf{Size-Agnostic Setup}\\[2pt] Captures equilibrium-shifting macro\\cascades (``The Straw that Breaks the Camel's Back'')};

\node[passbox, below=of mftg_size] (mftg_joint) {\textbf{Joint Distribution Fields}\\[2pt] Retains high-dimensional couplings\\across wealth, radicalization, and gold prices};

\node[passbox, below=of mftg_joint] (mftg_risk) {\textbf{Wasserstein Risk Expectiles}\\[2pt] Employs $\alpha_i, \beta_i$ mappings to analyze\\divergent vulnerabilities to environmental shocks};

\node[passbox, below=of mftg_risk] (mftg_volterra) {\textbf{Volterra Fractional Operators}\\[2pt] Incorporates power-law memory $(t-s)^{-\gamma}$ \\and hereditary transmission lines $\mathcal{T}_{\boldsymbol{\zeta}}$};

% Background Layer for MFTG Panel
\begin{scope}[on background layer]
    \node[panel_mftg, yshift=1cm,fit={(mftg_size) (mftg_joint) (mftg_risk) (mftg_volterra)}, 
          label={[anchor=north, font=\sffamily\Large\bfseries\color{green!50!black}, yshift=-4pt]north:  Mean-Field-Type Games (MFTG)}] (right_panel) {};
\end{scope}

% =====================================================================
% RESOLUTION FLOWS (CONNECTING OBSERVATIONS TO SOLUTION BLOCKS)
% =====================================================================
\draw[solutionflow] (o1) -- node[above, lbl] {Modeled via\\Arbitrary Size Scaling} (mftg_size);
\draw[solutionflow] (o2) -- node[above, lbl] {Preserved via Full\\Wasserstein Measure $\mu_t^g$} (mftg_joint);
\draw[solutionflow] (o3) -- node[above, lbl] {Tracked via G\^ateaux\\Weight Functions $\omega_\alpha$} (mftg_risk);
\draw[solutionflow] (o4) -- node[above, lbl] {Bound via Epochs $[T_g, T_{g+1})$\\and Non-Markovian Loops} (mftg_volterra);

\end{tikzpicture}%
}
\caption{Game-Theoretic Frameworks for the Malian Conflict Space. The left panel shows where classical Mean-Field Games (MFG) break down due to structural limitations like the continuum and independence assumptions. The right panel demonstrates how Mean-Field-Type Games (MFTG) accurately capture Sahelian field realities, modeling non-linear individual cascades (tiping thresholds), full distribution manifolds, asymmetric expectile risk fields, and historical Volterra-driven path dependencies.}
\label{fig:mfg_vs_mftg_bifurcation}
\end{figure}

\begin{enumerate}
    \item 
    Standard Mean-Field Games operate under the continuum assumption where the population is modeled as a continuum, meaning the actions of any single individual have an infinitesimal,  negligible impact on the aggregate state. In Mali, this atomistic assumption is structurally invalid. The security context is highly sensitive to non-linear cascades triggered by specific individuals. A single local actor, whether a tactical field commander in the Inner Niger Delta, an influential marabout in the central region, or a well-placed war entrepreneur manipulating smuggling nodes along the Mauritanian and Algerian border, can fundamentally alter the macro-security trajectory through specific, targeted actions (e.g., orchestrating a localized massacre, collapsing a fragile community ceasefire, or shifting a village's allegiance). Because individual actions regularly trigger macro-level system shifts, no single player can be modeled as strategically negligible.

    \item 
    Refining standard MFGs into multi-population frameworks (where players are grouped into a finite number of large teams) fails to resolve this limitation. In these formulations, individual agents within a specific team (such as a separatist alliance, a traditional  hunter guild, or a bandit or jihadist network) are treated as statistically symmetric and atomized within their population block. In reality, these groups are deeply asymmetric. Factions within the same broader coalition exhibit widely divergent risk sensitivities, distinct access to illicit capital loops, and entirely different individual payoff functionals. Reducing them to an average, indistinguishable team profile discards the intra-coalition frictions, localized leadership dynamics, and opportunistic defections that define the fluid context of Sahelian allegiances.

    \item 
    Coalitional MFG models assume that while players organize into clubs or coalitions, the internal structure of each coalition is governed by symmetric players who face identical incentives and possess identical risk profiles. In the Malian conflict space, intra-coalition symmetry is a structural fiction. For instance, within a community-based self-defence militia, a wealthy livestock owner and a destitute youth do not share symmetric payoffs or symmetric risk attitudes. Their individual vulnerabilities to shocks (such as droughts or state-level livestock seizures) are highly unequal. A coalitional MFG cannot capture these internal vulnerabilities, which are precisely the levers that war entrepreneurs exploit to manipulate group behavior.

    \item 
    A foundational  pillar of classical MFGs is the \emph{propagation of chaos}, which states that as the player count grows, the joint distribution of individual states decouples into independent, conditionally identically distributed trajectories given the mean field. In Mali, this decoupling is totally absent. The system is characterized by tight, persistent correlations across all dimensions. Individual radicalization levels, local asset holdings, transnational financial flows, cross-border border porosity, and changing gold prices are locked in a dense web of mutual dependency. Modeling individual trajectories as conditionally independent completely eliminates the tail-risk sensitivities and correlated systemic shocks that drive the conflict.

    \item 
    In standard MFGs, agents interact exclusively through a set of aggregate scalar statistics, typically the mean of the population state. However, in the Malian war economy, strategic choices are sensitive to the \emph{exact joint probability law} $\mu_t^g$ of the complete state-control configuration. For example, a war entrepreneur's extraction profit function is not driven by the average wealth of the population, but by the precise set of  the population that is highly radicalized versus the proportion that is totally destitute. By replacing the simple population average with the full joint measure on the Wasserstein space, we preserve these vital higher-order statistical dependencies.

    \item 
    A defining feature of our model is that an individual's own state appears in the system coefficients and payoff functionals by means of $\alpha_i$- and $\beta_i$-expectiles. As we will see below, the G\^ateaux derivative of an expectile risk maps onto a highly non-linear weight function $\omega_\alpha(\mathbf{z}', \mu)$ that depends on the global shape of the probability distribution. Because these risk measures are fundamentally non-linear with respect to the agent's own probability measure, the model falls squarely outside the scope of classical MFGs and must be formulated within the  class of \emph{Mean-Field-Type Games} (MFTGs).

    \item 
    Unlike classical MFGs, which are strictly asymptotic and valid only in the infinite limit, MFTG are  sound across an arbitrary population size. They apply perfectly to games played by two, a few, medium, or large cohorts. In the context of Malian stabilization, this size agnosticism is essential. It allows the  framework to seamlessly scale from analyzing micro-level strategic interactions between a few localized village factions up to analyzing macro-level regional dynamics involving millions of actors across the broader Sahelian zone, all without altering the underlying game-theoretic primitives.

    \item 
    Classical mean-field limits impose asymptotic indistinguishability, forcing the strategic profiles of the population to blend into a homogeneous mass. MFTGs, by contrast, preserve individual identities and structural asymmetries. In the Malian crisis, the retention of asymmetry is non-negotiable. It allows us to explicitly model highly influential, asymmetrical actors such as a specific transnational trafficking kingpin or a centralized institutional authority, operating alongside a large, heterogeneous population of civilians and combatants, capturing the exact top-down and bottom-up feedback loops that characterize the field.

    \item 
    Standard MFTG formulations are typically stationary or continuous over a single time horizon, which fails to capture the generational lifecycle of Sahelian instability. The conflict in Mali is intrinsically historic and intergenerational. The planning horizon must be partitioned into discrete generational epochs $[T_g, T_{g+1})$, where successive generations are coupled via non-local hereditary transmission operators \eqref{eq:pushforward}. Each new generation enters the game facing a distinct, inherited sequence of broken peace accords, accumulated financial capital, and shifting political regimes ($s_t^g$), meaning the equilibrium strategy of the current generation must actively anticipate its long-term impact on the initial conditions of the next.

    \item 
    Classical mean-field models  consider the pair of representative agent and the population mean-field object to characterize the system, assuming that the future evolution of the state depends exclusively on its current instantaneous value. In practice, the Malian conflict space is driven by deep-rooted, non-Markovian memory loops of own-paths and population paths. The memory of a land dispossession during the droughts of the 1970s or a village massacre from the 2010s does not decay exponentially; it acts as a persistent, power-law-weighted force that shapes modern insurgent recruitment. By embedding a singular \emph{Volterra fractional integral operator} $(t-s)^{-\gamma}$ into the MFTG framework \eqref{eq:revenge}, we can formally incorporate each agent's inherited historical memory. This binds personal historical trajectories directly into the continuous dynamics, capturing the long-memory blood feuds that create the generational war trap.
    \item The phenomenon of `` the straw that breaks the camel's back" corresponds, in mean-field terms, to a system in which a single state
    $ x_{i,t}$, a single action $(u_{i,t},v_{i,t})$, an individual state distribution $\mathcal{L}(x_{i,t})$ and an individual action distribution $\mathcal{L}(u_{i,t},v_{i,t})$  are going to be important when population
state distribution and action distribution are concentrated near critical thresholds
of tolerance or trauma. When a single individual agent regardless of its class or type, takes a
provocative or aggressive action, this initiates a shift of the distribution $\mu_t$ measurably, potentially
crossing a tipping point where the best responses $u_{i,t}^*(t, x_{i,t}^p, \mu_t^{g, state})$  for many agents switch
from peaceful to retaliatory. This results in a cascade, updating $\mu_t$  toward increased
aggression. The mean-field-type feedback loop amplifies the single action's impact, making it the effective trigger for mass unrest or escalation, demonstrating how, even in a large system, individual
actions can shift equilibria when the system is already near critical instability. The phenomenon
can be examined using singularity formation of the mean-field knowing everyone can be actor of
the singularity formation. The formal analogue of the proverb now becomes: {\it when the collective
frustration or instability is full, even a tiny drop (one action from a one single individual agent)
makes it overflow.} This is a phenomenon widely observed in the field in Mali since 2012.  Classical  population mean-field games cannot capture it. However, MFTG can capture very well such phenomena (see Figure \ref{fig:mfg_vs_mftg_bifurcation}).
\end{enumerate}

\medskip
\noindent\textbf{Literature review on Civil Wars and alternative system of profit and power.}
The works in \cite{hsiang2011, burke2009, hsiang2013, vonuexkull2016} establishe that large-scale armed instability is dynamically accelerated by exogenous environmental and climatic stressors  and catastrophic macroeconomic shocks \cite{miguel2004}, while simultaneously functioning as a severe, long-term public health crisis that fundamentally alters regional mortality trends and datasets across generational lines \cite{ghobarah2003, murray2002, coghlan2006, spagat2009, Pettersson}.
Within the contemporary literature on security and fragile states, a profound ontological critique highlights that the classical, binary framework of ``civil war" fails to capture the true structural mechanics of highly fragmented modern conflicts. Rather than describing an organized, state-centric, zero-sum standoff over administrative control, these seminal works uniformly demonstrate that modern crises behave as decentralized, adaptive ecosystems where violence functions through alternate, localized logic. \cite{kalyvas06} deconstructs the unitary macro-political narrative of civil war, showing that it is an artificial aggregation of hyper-local, privatized feuds that are better categorized as the \emph{joint production of violence} driven by endogenous cleavages.  \cite{kaldor13} explicitly rejects the state-centric, Clausewitzian paradigm of conventional domestic warfare, proposing instead the concept of \emph{new wars} or \emph{transnational network warfare} where decentralized non-state factions strategically weaponize identity politics and sustain their operational presence through global illicit trafficking networks rather than formal domestic taxation. This focus on alternate, self-sustaining systemic logic is further advanced by  \cite{keen08, keen05}, who reframes armed crises from a tragic breakdown of peace into an \emph{alternative system of profit and power}, demonstrating that for predatory elites and war entrepreneurs, the continuous preservation of institutional collapse and conflict serves as a highly optimized, functional war economy where the perpetuation of instability yields immediate financial and structural rents. These frameworks provide strong justification for discarding traditional state-versus-rebel assumptions in favor of a polycentric, multi-layered approach that views the crisis not as a bounded domestic war, but as a complex, self-reinforcing socio-economic system.

\medskip
\noindent\textbf{Content of this paper.}
In this paper we abandon asymptotic mean‑field approximations and adopt the finite‑player \emph{Mean‑Field‑Type Game (MFTG)} paradigm, in which all interactions are mediated by the true joint distribution of the complete state‑control configuration. For each agent $i\in \mathcal{N} = \{1,\ldots ,N\}$, let $X_{i,t}^{g}\in \mathbb{R}^{d}$, $C_{i,t}^{g}\in \{1,\ldots ,K\}$, denote respectively the continuous socio‑economic state (radicalisation, local wealth, human capital) and the discrete coalition affiliation at time $t$ within generation $g$. The macro‑environment (international gold prices, cross‑border permeability, rainfall, external intervention) co‑evolves with the population through coupled stochastic differential equations. To capture the multi‑scale nature of the conflict, we introduce a complete catalogue of stochastic drivers: Brownian motions for short‑term volatility, fractional Brownian motions for long‑range dependence, Gauss‑Volterra processes, Rosenblatt processes for heavy‑tailed asymmetric shocks, Poisson random measures for sudden structural breaks, and a continuous‑time Markov chain for regime shifts (ceasefire, asymmetric war, full‑scale war, foreign intervention).
Within each generation, individual agents minimise an intertemporal cost functional that incorporates asymmetric risk measures, specifically $\alpha_{i}$- and $\beta_{i}$-expectiles, which are coherent risk measures (if  $1>\alpha_i,\beta_i\geq \frac{1}{2}$) that weight gains and losses differently. This captures the pronounced asymmetry in how conflict actors value territory lost versus territory gained. These expectiles are computed directly from the joint probability measure $\mu_{i}^{g}$, making each agent's optimal behaviour sensitive to the tail of the population‑environment distribution.
A particularly salient feature of the Malian war economy  and one that our framework explicitly models, is the existence of a class of war entrepreneurs. These actors, a subset $\mathcal{N}_E\subset \mathcal{N}$, hold a dual portfolio: local (illicit) wealth, which evolves according to the same McKean‑Vlasov‑Volterra jump‑diffusion as the rest of the population, and global (laundered) wealth, which is sheltered in offshore financial centres and evolves according to \[
dX_{i,t}^{g,\mathrm{global}} = \left[r_{\mathrm{global}}X_{i,t}^{g,\mathrm{global}} + \eta_{\mathrm{wash}}\Pi_{\mathrm{local}} - u_{i,t}^{g,\mathrm{global}}\right]dt + \sigma_{\mathrm{global}}X_{i,t}^{g,\mathrm{global}}dW_{i,t}^{g,\mathrm{global}},
\]
where $\eta_{\mathrm{wash}}$ is the laundering efficiency. This equation embodies the self‑reinforcing war loop: instability generates illicit profits, which are laundered into safe global assets, compounded, and partially  repatriated to fuel further violence, ensuring that peace is economically unattractive for the entrepreneurial class.

\emph{The macro‑environment.}  
The environment in which the Malian conflict unfolds is far more than a passive backdrop; it is an active, co‑evolving stratum that both shapes and is shaped by the behaviour of every actor.  Its constituent dimensions include, among others, the international gold price (Mali being Africa's third‑largest producer, where artisanal mining sustains entire communities and finances armed groups), the porosity of borders with Algeria, Burkina Faso, and Niger that enables the unhindered flow of fighters, weapons, and illicit goods, the yearly rainfall and its effect on agricultural output and food security, the volume and composition of foreign military and humanitarian assistance, the reach and integrity of state institutions, and the intensity of transnational financial surveillance.  A surge in the gold price renders illegal extraction more lucrative, drawing civilians into conflict zones and fuelling competition among militias.  A drought simultaneously impoverishes pastoralists, erodes trust in the central state, and creates a reservoir of desperate recruits.  The withdrawal of a foreign intervention force abruptly alters the balance of power and opens space for extremist consolidation.  Each of these environmental variables is not merely an exogenous parameter but an endogenous product of the conflict itself: violence degrades institutional capacity, discourages investment, and re‑routes economic circuits through illicit channels, which in turn deepen instability.  The model therefore treats the environment as a dynamic, multi‑dimensional state vector that co‑evolves with the population, driven by a rich set of stochastic processes capable of reproducing the long‑memory, heavy‑tailed, and regime‑switching behaviour observed in empirical data.
The complete micro-macro configuration at time $t$ is therefore
\[
\mathbf{Z}_t^g = \bigl( s_t^g, \mathbf{X}_t^g,\;\mathbf{C}_t^g,\;E_t^g \bigr)
\in \cS\times \R^{d\times N}\times\{1,\dots,K\}^N\times\R^{d_E},
\]
and the fundamental state variable of the game is the exact probability measure
\begin{equation}
\mu_t^g = \operatorname{Law}\bigl(\mathbf{Z}_t^g\bigr)
\in \mathcal{P}\bigl(\cS\times \R^{d\times N}\times\{1,\dots,K\}^N\times\R^{d_E}\bigr).
\end{equation}
By working directly with $\mu_t^g$, we preserve all higher‑order dependencies among agents, coalitions, institutions, and environmental variables, thereby avoiding the information loss inherent in empirical particle approximations and asymptotic propagation‑of‑chaos arguments.

A second crucial empirical regularity of civil wars is their remarkable \emph{intergenerational persistence} \cite{tembine2026achieving}.  Violence often survives the disappearance of the original perpetrators through the transmission of collective memory, inherited grievances, revenge norms, institutional degradation, and accumulated war capital.  The children of today's combatants inherit not only their parents' trauma and dispossession but also, in the case of war entrepreneurs, offshore bank accounts and transnational laundering networks.  To capture these mechanisms, we partition the macroscopic planning horizon into a sequence of generational epochs
\[
[T_0,T_1) \to [T_1,T_2) \to \dots \to [T_G,T_{G+1}],
\]
and couple successive generations through \emph{hereditary transmission operators} acting directly on the probability measure.  At the interface $T_{g+1}$, the terminal law $\mu_{T_{g+1}-}^g$ is mapped to the initial law of the next generation via a non‑local push‑forward
\[
\mu_{T_{g+1}}^{g+1} = \mathcal{T}_{\boldsymbol{\zeta}}\bigl[\mu_{T_{g+1}-}^g\bigr],
\]
where the trauma transmission coefficients $\zeta_i\in[0,1]$, coalition inheritance kernels, and environmental persistence matrices encode the complex interplay of cultural and economic reproduction.  The resulting game is a genuinely \emph{intergenerational} stochastic game on the space of probability measures.

A further distinctive feature of the model is its comprehensive treatment of uncertainty.  Sahelian conflict data exhibit a combination of short‑term volatility, long‑range dependence, heavy‑tailed shocks, and abrupt geopolitical transitions that cannot be captured by standard Brownian diffusions alone.  We therefore incorporate a heterogeneous collection of stochastic drivers: standard Brownian motions for idiosyncratic continuous fluctuations; two independent fractional Brownian motions with Hurst parameters $H_1,H_2\in(1/2,1)$ to model long‑memory environmental processes such as rainfall and commodity prices; a Gauss–Volterra process for flexible long‑memory modelling; a Rosenblatt process, a non‑Gaussian self‑similar process from the second Wiener chaos, to represent heavy‑tailed asymmetric shocks in illicit financial flows; compensated Poisson random measures for sudden catastrophic events (coups, sanctions, droughts); and a finite‑state Markov chain capturing macro‑political regime shifts between peacetime, low‑intensity conflict, full‑scale war, and external intervention.
Within each generation, individual agents minimise an intertemporal cost functional that incorporates \emph{asymmetric risk measures}.  We employ $\alpha_i$- and $\beta_i$-expectiles, which are coherent risk measures that weight gains and losses differently, thereby capturing the pronounced asymmetry in how conflict actors value territory lost versus territory gained.  These expectiles are computed directly from the joint probability measure $\mu_t^g$, making each agent's optimal behaviour sensitive to the tail of the population-environment distribution. The strategic interactions within and across generations give rise to a coupled forward‑backward structure on the infinite‑dimensional Wasserstein manifold.  In the forward direction, the joint measure evolves according to controlled McKean-Vlasov-Volterra jump‑diffusions with endogenous coalition switching.  In the backward direction, adjoint processes propagate future strategic sensitivities across generational boundaries through a sequence of path‑dependent \emph{adjoint equations}.  Because the fractional, Gauss-Volterra, and Rosenblatt noises enter the dynamics with coefficients independent of the agents' controls, they act as exogenous long‑memory perturbations and do not require a maximum principle beyond the standard one for jump‑diffusions with regime switching, a fact that we exploit to keep the adjoint analysis relevant. The resulting framework provides a  foundation for studying how civil wars emerge, persist, and can potentially be ended when strategic incentives, institutional structures, environmental conditions, and inherited historical memory jointly interact across multiple generations.  By working with the true joint law rather than mean‑field approximations, and by incorporating the long‑memory and non‑Gaussian features of the data, the model offers a principled tool for quantitative policy analysis in fragile states.

\medskip
\noindent\textbf{Main Contributions.}
The main contributions of this paper are the following:

\begin{enumerate}
\item  We formulate a finite‑population intergenerational Mean‑Field‑Type Game in which the state variable is the exact joint probability law of all individuals, coalition memberships, and environmental variables. The model captures the full spectrum of Sahelian conflict dynamics: continuous state adaptation, discrete coalition switching, environmental feedback, long‑memory revenge cycles, intergenerational trauma and wealth transmission, and asymmetric risk attitudes via expectiles.

\item We incorporate a complete catalogue of stochastic drivers (Brownian, fractional Brownian, Gauss‑Volterra, Rosenblatt, Poisson, and regime‑switching) enabling the model to reproduce simultaneously short‑term volatility, long‑range dependence, heavy‑tailed shocks, and abrupt structural shifts.

\item  Working directly on the measure space without invoking dynamic programming, we derive the system of path‑dependent adjoint equations that characterise the mean-field-type  equilibrium. Exact trans‑generational boundary conditions link the adjoint processes of successive generations, providing a foundation for the backward induction.

\item  We prove existence and uniqueness of the equilibrium under an explicit local spectral condition. Most importantly, we identify a sharp generational ergodic breakdown threshold: if the combined spectral radius of the grievance memory matrix and the offshore laundering loop exceeds unity, the peaceful steady state is unstable and any small perturbation triggers permanent civil war.

\item  We design an optimal institution‑level transfer policy (combining an offshore transaction tax with endogenous mediation subsidies) that forces the system's spectral radius below one, thereby breaking the war trap and guaranteeing convergence to a stable peaceful equilibrium. The design acknowledges that traditional conflict‑resolution mechanisms such as the joking kinship system, while historically a source of social cohesion, have in contemporary Sahelian practice often been instrumentalised as tools of domination and social control. Our proposed mediation subsidies are therefore structured to reinforce the genuinely conciliatory dimensions of such customs while neutralising their co‑optation by predatory elites.
\end{enumerate}

These contributions provide the first  framework for understanding and breaking the self‑reproducing conflict system in Mali. By moving from descriptive accounts of civil war toward a precise, measure‑theoretic theory of polycentric intergenerational conflict, this paper offers both a diagnostic tool for identifying when a society has entered a war trap and a prescriptive policy design for escaping it. The results are directly applicable to the Malian‑Sahelian theatre and, with appropriate calibration, to other regions afflicted by chronic, multi‑generational civil violence.

%%%
%

\medskip
\noindent\textbf{Paper Outline.}
The remainder of this paper is organized as follows.  The next section lays out the necessary preliminaries on fractional Brownian motion, Gauss‑Volterra processes, Rosenblatt processes, Poisson random measures, and regime‑switching. We then construct the full intergenerational model, including the four‑layer individual and environmental dynamics, controlled coalition switching, intergenerational transmission, the dual asset structure of war entrepreneurs, expectile risk measures, and the individual payoff functionals. After that we characterize the mean-field-type  equilibrium via trans‑generational  adjoint  equations. We then prove existence and uniqueness of the equilibrium and establishes the generational ergodic breakdown threshold (the war trap). Then we  design the optimal institutional policy that stabilizes the system. Finally we conclude with a discussion of policy implications and directions for empirical calibration.

We provide a glossary of notation used throughout.
\subsection*{Notations}
\begin{itemize}
\item $N$: total population ($80\times10^6$).
\item $G$: number of generational transitions.
\item $T_g$: start of generation $g$.
\item $d$, $d_E$: dimensions of individual continuous state and environment.
\item $K$: maximum number of coalitions.
\item $\cS = \{1,\dots,M\}$: macro‑political regime states.
\item $X_{i,t}^g$: individual $i$ continuous state (vector in $\R^d$).
\item $C_{i,t}^g$: coalition of agent $i$, values in $\{1,\dots,K\}$.
\item $E_t^g$: environment state (vector in $\R^{d_E}$).
\item $s_t^g$: regime state (Markov chain on $\cS$).
\item $\mu_t^g$: joint probability law of $(s_t^g,\mathbf{X}_t^g,\mathbf{C}_t^g,E_t^g)$.
\item $u_{i,t}^g$, $v_{i,t}^g$: continuous and discrete switching controls.
\item $\mathcal{K}_{\text{rv}}^{i,g}(t)$: Volterra revenge operator.
\item $\mathbb{M}_i(s,s)$: historical grievance operator, dependent on regime.
\item $\gamma\in(0,1)$: memory exponent.
\item $\zeta_i$: intergenerational trauma transmission coefficient.
\item $\eta_{\text{wash}}$: offshore laundering efficiency.
\item $e_{\alpha_i}(\mu)$, $e_{\beta_i}(\mu)$: $\alpha_i$- and $\beta_i$-expectiles of the distribution $\mu$.
\item $B^{H_1}, B^{\frac{1+H_2}{2}}$: independent fractional Brownian motions ($1>H_1,H_2>1/2$).
\item $G_t$: Gauss-Volterra process.
\item $R^{H_2}$: Rosenblatt process associated with $B^{\frac{1+H_2}{2}}$.
\item $\tilde N, \tilde N_E$: compensated Poisson random measures.
\item $Q^g$: generator of the Markov chain $s_t^g$.
\item $R^{i,g}, S^{i,g}, Q^{i,g}$: cost weighting matrices (may depend on $s_t^g$).
\item $\Phi_\beta^{i,g}$: asymmetric expectile-based running penalty.
\item $G^{i,g}$: terminal cost function.
\item $Y_{i,t}^g, Y_{E,t}^{i,g}$: adjoint processes.
\item $\mathbb{A}_E$: environmental persistence matrix.
\item $\mathbf{A}_{\text{policy}}$: effective system matrix after policy intervention.
\end{itemize}

\section{Preliminaries}
\label{sec:prelim}
Let $(\Omega, \F, \Pee)$ be a complete probability space equipped with a right-continuous filtration $\mathbb{F} = (\F_t)_{t \ge 0}$, where $\F_0$ contains all $\Pee$-null sets of $\F$, thereby satisfying the usual conditions. The macroscopic planning horizon across all generations is denoted by $\mathcal{T} = [T_0, T_{G+1}] \subset \R_+$. Unless stated otherwise, all stochastic processes introduced herein are defined on $(\Omega, \F, \mathbb{F}, \Pee)$, are  $(\F_t)$-progressively measurable, and possess trajectories that are $\Pee$-almost surely right-continuous with left limits (càdlàg) in their respective Polish state spaces.
 
We assume the existence of the following mutually independent random objects.

\subsection{Fractional Brownian motion}
A fractional Brownian motion $B^H = (B^H(t))_{t\ge 0}$ with Hurst parameter $H\in(0,1)$ is a centred Gaussian process with covariance
\begin{equation}
\E\bigl[B^H(s)B^H(t)\bigr] = \frac12\bigl(s^{2H} + t^{2H} - |t-s|^{2H}\bigr), \qquad s,t\ge 0.
\end{equation}
For $H = 1/2$, $B^{1/2}$ is a standard Brownian motion.  
For $H > 1/2$, the increments are positively correlated and exhibit long‑range dependence.  
In this paper we use two independent fractional Brownian motions $B^{H_1}$ and $B^{\frac{1+H_2}{2}}$ with Hurst indices $H_1, H_2 \in (1/2,1)$.  
The process $B^{\frac{1+H_2}{2}}$ is also the fBm canonically associated with the Rosenblatt process defined below \cite{Coupep2022}.

\subsection{Gauss–Volterra process}
A Gauss–Volterra process $G = (G_t)_{t\ge 0}$ is defined by the Wiener integral
\begin{equation}
G_t = \int_0^t K(t,s)\,dW_s,
\end{equation}
where $W$ is a standard Brownian motion and $K$ is a deterministic square‑integrable kernel.  
The quadratic variation of $G$ is $\langle G\rangle_t = \int_0^t K(t,s)^2\,ds$.

\subsection{Rosenblatt process}
Let $H \in (1/2,1)$.  
A Rosenblatt process $R^H = (R^H(t))_{t\ge 0}$ is a non‑Gaussian, $H$‑self‑similar process with stationary increments living in the second Wiener chaos.  
It admits the double Wiener-Itô representation
\begin{equation}
R^H(t) = C^H_R \int_{\R^2 \backslash \{Diagonal \}} \Bigl( \int_0^t h_2^H(u,y_1,y_2)\,du \Bigr) dW(y_1)dW(y_2),
\label{eq:rosenblatt}
\end{equation}
where $h_2^H=\left( \int_0^t (u-y_1)_+^{\frac{H}{2}-1} (u-y_2)_+^{\frac{H}{2}-1} du \right)$ is a  square‑integrable kernel and $C^H_R$ is a normalising constant ensuring $\E[R^H(1)^2] = 1,$  Beta: $
B(x,y) = \int_0^1 t^{x-1}(1-t)^{y-1} dt, 
$
               $C_R^H := \sqrt{\frac{2H(2H-1)}{2B(1-H, \frac{H}{2})}} $ ensuring that $\mathbb{E}[(R^H(t))^2]=1.$
                $\frac{1}{2} <H <1.$
The fractional Brownian motion $B^H$ associated with $R^H$ is obtained from the first‑order kernel $h_1^H$.  
The Rosenblatt process exhibits long memory and non‑zero skewness and kurtosis, making it an ideal model for asymmetric, heavy‑tailed fluctuations in illicit financial flows.  
For the Itô formula of Rosenblatt processes we refer to \cite{Coupep2022}.

\subsection{Poisson random measures}
Let $N(dt,d\theta)$ be a Poisson random measure on $\R_+ \times \Theta$ with intensity $\nu(d\theta)dt$, where $\Theta \subseteq \R^m$ and $\nu$ is a $\sigma$-finite Lévy measure.  
The compensated measure is $\tilde N(dt,d\theta) = N(dt,d\theta) - \nu(d\theta)dt$.  
We employ several independent copies of such measures to capture sudden macro‑economic and security shocks (sanctions, coups, resource price collapses).

\subsection{Regime switching}
The macro‑political environment is subject to discrete shifts.  
We model these by a continuous‑time Markov chain $s = (s_t)_{t\ge 0}$ on the finite state space $\cS = \{1,\dots,M\}$ with generator $Q = (q_{ij})$.  
For $i \neq j$, $q_{ij} \ge 0$ and $\sum_{j \neq i} q_{ij} = -q_{ii}$.  
The jump measure of the chain is $N^s(dt,i) = \sum_{u:\Delta s_u \neq 0} \mathbbm{1}_{\{\Delta s_u = i\}} \delta_{(u,i)}$, with compensator $q_{s_{t-},i}\,dt$.  
The chain is independent of all other driving noises.

\section{The Intergenerational Model of Polycentric Conflict-Trap Risk  in Mali}
\label{sec:model}

We now construct a  stochastic framework that captures the essential features of the Malian conflict.  
Every variable is chosen to reflect a tangible quantity observed in the field, and every equation is justified by the underlying socio‑economic logic of polycentric conflict.

\subsection{Population, generations, and state variables}
\label{sec:population}

The population affected by the conflict is estimated at $N = 80 \times 10^6$ individuals (the combined population of Mali and the cross‑border Sahelian regions involved in the crisis).  
We index them by $\cN = \{1,\dots,N\}$.

The macroscopic planning horizon is partitioned into $G+1$ discrete \emph{generational epochs}
\begin{equation}\label{eq:horizon}
[T_0,T_1) \to [T_1,T_2) \to \dots \to [T_G,T_{G+1}],
\end{equation}
where $[T_g,T_{g+1}]$ is the active lifetime of generation $g$ (could be 3 years, 5 years, or even 20 years).  
Within generation $g$, at any instant $t \in [T_g,T_{g+1}]$, each agent $i$ is described by:

\begin{itemize}
\item \textbf{True continuous state} $X_{i,t}^g \in \R^d$, a vector that aggregates:
\begin{itemize}
    \item $X_{i,t}^{g,1}$: radicalisation index, ranging from $0$ (fully integrated civilian) to $1$ (active combatant);
    \item $X_{i,t}^{g,2}$: local asset holdings (livestock, agricultural land, gold stocks, weapons, local currency);
    \item $X_{i,t}^{g,3}$: human capital and health status (proxied by food security, access to education, etc.).
\end{itemize}
The dimension $d$ is flexible; for concreteness we often take $d = 3$.
\item \textbf{Discrete coalition affiliation} $C_{i,t}^g \in \{1,\dots,K\}$, where $K\geq 1$ is the maximum number of distinct coalitions that can coexist.  
In the Malian context, a typical value is $K=50$ with the interpretation:
\begin{enumerate}
    \item State core forces and loyalist militias;
    \item Separatist coalitions;
    \item Insurgent networks ;
    \item Community‑based self‑defence militias;
    \item Transnational criminal/entrepreneurial networks (war entrepreneurs), etc.
\end{enumerate}
Coalitions can become endogenously extinct if no agent belongs to them, in which case the effective number of active coalitions drops below $K$.
\end{itemize}

The \textbf{macro‑environment} is described by a vector $E_t^g \in \R^{d_E}, d_E\geq 1$, whose components may include:
\begin{itemize}
    \item $E_t^{g,1}$: international gold price (Mali was Africa's third largest gold producer);
    \item $E_t^{g,2}$: cross‑border permeability index (porosity with  Mauritania, Algeria, Burkina Faso, Niger, etc.);
    \item $E_t^{g,3}$: aggregate rainfall and agricultural productivity (droughts are a known conflict multiplier);
    \item $E_t^{g,4}$: volume of transnational aid and peacekeeping resources deployed, etc.
\end{itemize}

The \textbf{macro‑political regime} $s_t^g \in \cS = \{1,\dots,M\}, M\geq 1$ captures the qualitative state of the peace process.  
Typical regimes might be:
\begin{enumerate}
    \item Active peace negotiations / ceasefire;
    \item Low‑intensity asymmetric warfare;
    \item Full‑scale polycentric conflict with high civilian displacement;
    \item Foreign military intervention and private military company services.
\end{enumerate}
The regime process $s_t^g$ is a continuous‑time Markov chain independent of all other noises, with generator $Q^g$ (which may depend on the generation).  
It is assumed to be observable by all agents.

The \textbf{complete micro‑macro configuration} at time $t$ is the joint vector
\begin{equation}\label{eq:Z}
\mathbf{Z}_t^g \;=\; \bigl(s_t^g, \mathbf{X}_t^g,\;\mathbf{C}_t^g,\;E_t^g\bigr) \;\in\; \cS\times \R^{d\times N} \times \{1,\dots,K\}^N \times \R^{d_E},
\end{equation}
where $\mathbf{X}_t^g = (X_{1,t}^g,\dots,X_{N,t}^g)$ and $\mathbf{C}_t^g = (C_{1,t}^g,\dots,C_{N,t}^g)$.  
The \textbf{statistical state} of the system is the true joint probability measure
\begin{equation}\label{eq:mu}
\mu_t^{g,state} \;=\; \operatorname{Law}(\mathbf{Z}_t^g) \;\in\; \mathcal{P}\bigl(\cS\times \R^{d\times N}\times\{1,\dots,K\}^N\times\R^{d_E}\bigr).
\end{equation}
All strategic interactions are channelled through $\mu_t^{g,state}$.

\subsection{Individual continuous dynamics: the four‑layer structure}
\label{sec:fourlayer}

In the Malian conflict, no actor possesses perfect information. Combatants, civilians, and entrepreneurs operate under a thick fog of war: they receive fragmented reports, observe signals corrupted by rumour, and must form beliefs about their own situation (radicalisation level, asset holdings, health) before choosing a course of action. To capture this reality, the continuous state of each agent $i\in\cN$ within generation $g$ is decomposed into four interacting layers: a \emph{true} state, a noisy \emph{measurement}, an \emph{observation} that filters the measurement, and finally a \emph{perception} (belief) upon which decisions are based. All coefficients are $\mathbb{F}$-progressively measurable, satisfy standard Lipschitz and linear growth conditions, and depend on the macro‑political regime $s_t^g$ in a deterministic fashion.

\medskip
\noindent\textbf{True state.}
The true continuous state $X_{i,t}^g\in\mathbb R^d$ evolves according to the McKean-Vlasov-Volterra jump‑diffusion
\begin{equation}
\begin{aligned}
dX_{i,t}^g &=
\Bigl[\,b_i\bigl(t,s_t^g,\mathbf X_t^g,\mathbf C_t^g,E_t^g,\mathbf u_t,\mu_t^g\bigr)
      + H^{i,g}\bigl(t,s_t^g,E_t^g,e_{\alpha_i}(\mu_t^g)\bigr)
      + \mathcal{K}_{\mathrm{rv}}^{i,g}(t) \Bigr]\,dt \\[2pt]
&\quad + \sigma_i\bigl(t,s_t^g,\mathbf X_t^g,\mathbf C_t^g,E_t^g,\mu_t^g\bigr)\,dW_{i,t}
   + \int_{\mathcal Z}\sigma^n_i\bigl(t,s_{t-}^g,\mathbf X_{t-}^g,\mathbf C_{t-}^g,E_{t-}^g,\mu_{t-}^g,z\bigr)\,\tilde N_i(dt,dz) \\[2pt]
&\quad + \sigma_i^{H_1}\bigl(t,s_t^g,\mathbf X_t^g,\mathbf C_t^g,E_t^g,\mu_t^g\bigr)\,dB^{H_1}_{i}(t)
   +  2\tilde{c} \sigma_i^{H_2}\bigl(t,s_t^g,\mathbf X_t^g,\mathbf C_t^g,E_t^g,\mu_t^g\bigr)\,dB^{\frac{1+H_2}{2}}_i(t) \\[2pt]
&\quad + \sigma_i^{gv}\bigl(t,s_t^g,\mathbf X_t^g,\mathbf C_t^g,E_t^g,\mu_t^g\bigr)\,dG_i(t)
   + \sigma_i^{R}\bigl(t,s_t^g,\mathbf X_t^g,\mathbf C_t^g,E_t^g,\mu_t^g\bigr)\,dR^{H_2}_i(t) \\[2pt]
&\quad + \sigma_{0,i}\bigl(t,s_t^g,\mathbf X_t^g,\mathbf C_t^g,E_t^g,\mu_t^g\bigr)\,dW_t^{0}
   + \int_{\mathcal Z_0}\sigma^n_{0,i}\bigl(t,s_{t-}^g,\mathbf X_{t-}^g,\mathbf C_{t-}^g,E_{t-}^g,\mu_{t-}^g,z_0\bigr)\,\tilde N^0(dt,dz_0) \\[2pt]
&\quad + \sigma_{0,i}^{H_1}\bigl(t,s_t^g,\mathbf X_t^g,\mathbf C_t^g,E_t^g,\mu_t^g\bigr)\,dB^{H_1,0}(t)
   +  2\tilde{c}\sigma_{0,i}^{H_2}\bigl(t,s_t^g,\mathbf X_t^g,\mathbf C_t^g,E_t^g,\mu_t^g\bigr)\,dB^{\frac{1+H_2}{2},0}(t) \\[2pt]
&\quad + \sigma_{0,i}^{gv}\bigl(t,s_t^g,\mathbf X_t^g,\mathbf C_t^g,E_t^g,\mu_t^g\bigr)\,dG^0(t)
   + \sigma_{0,i}^{R}\bigl(t,s_t^g,\mathbf X_t^g,\mathbf C_t^g,E_t^g,\mu_t^g\bigr)\,dR^{H_2,0}(t).
\end{aligned}
\label{eq:true_dyn}
\end{equation} 

The initial condition for the individual state vector $X_{i,t}^g$ of agent $i$ at the temporal initiation of generation $g$, denoted by the 
epoch  switch boundary $t = T_g$, is defined by the random variable:
\begin{equation}
X_{i,T_g}^g = \xi_i^g \in \mathbb{R}^d, \quad \text{where} \quad \operatorname{Law}\bigl(\boldsymbol{\xi}^g\bigr) = \nu_{T_g}^g,
\label{eq:initial_condition}
\end{equation}
and the joint population-environment initialization measure $\mu_{T_g}^{g,state}$ is  given by the non-local hereditary push-forward operator acting on the left-hand terminal limit of the previous generation:
\begin{equation}
\mu_{T_g}^g = \mathcal{T}_{\boldsymbol{\zeta}} \Bigl( \mu_{T_g-}^{g-1} \Bigr) \in \mathcal{P}_2(\mathcal{S}\times \mathcal{Z}).
\label{eq:push_forward_init}
\end{equation}
Here, $\xi_i^g$ is an $\mathcal{F}_{T_g}$-measurable, square-integrable random variable that is independent of all future increments of the standard Brownian motions ($W_{i,t}, W_t^0$), fractional Brownian motions ($B_i^{H_1}, B_i^{\frac{1+H_2}{2}}, B^{H_1,0}, B^{\frac{1+H_2}{2},0}$), Poisson jump measures ($\tilde{N}_i, \tilde{N}^0$), and rough/sub-fractional drivers ($G_i, R_i^{H_2}, G^0, R^{H_2,0}$) for $t \ge T_g$.
$\mathbf{u}_t=(u_{1,t},\dots,u_{N,t})$ denotes the profile of control actions, and $\mu_t^g$ is the joint law of $(s_t^g,\mathbf X_t^g,\mathbf C_t^g,E_t^g,\mathbf u_t)$.    The control action profile $\mathbf{u}_t = (u_{1,t}, \dots, u_{N,t})$ represents the highly dynamic, risk-sensitive allocation of localized strategic capital across the fragmented situations of central, western and northern Mali. In the context of our polycentric framework, these controls transcend abstract  optimizations; they materialize in the field as real-time tactical adjustments made by decentralized actors facing severe informational opacity and volatile environmental shocks. For a local village council or traditional  hunting guild in the Mopti region, an upward shift in their specific control parameter $u_{i,t}$ reflects an active, defensive expenditure of finite community resources such as purchasing black-market ammunition, subsidizing local intelligence monitors, or constructing fortified earthen berms to protect livestock from hit-and-run raids. For a tactical field commander within the Jama'at Nusrat al-Islam wal-Muslimin  networks or a regional war entrepreneur manipulating the artisanal gold fields of Kidal, the control vector $\mathbf{u}_t$ dictates the rate of illicit capital reinvestment into asymmetric logistical networks. This includes procuring high-displacement motorbikes via informal Hawala financial pipelines, distributing protection payouts to fragile pastoralist groups to anchor local allegiances, or paying off corrupt local border nodes to secure weapon transit corridors from Niger  or Algeria. Because these controls are endogenously coupled to the full population distribution and asymmetric tail-risk expectiles rather than simple state averages, the collective profile $\mathbf{u}_t$ maps a highly sensitive, non-linear survival topology, where a top-down regulatory distortion by the state  such as a vehicle ban or an aggressive livestock tax forces an instantaneous, protective pivot in the control trajectory of civilians, paradoxically driving their strategic investments straight into the informal extraction loops managed by predatory insurgent networks.
The drift $b_i$ captures the autonomous dynamics of radicalisation, wealth, and human capital; the coupling term $H^{i,g}$ models how the true environment and the $\alpha_i$-expectile of the joint distribution affect individual trajectories (e.g.\ a spike in the gold price draws civilians into mining areas controlled by armed groups); the operator $\mathcal{K}_{\mathrm{rv}}^{i,g}$ encodes historical vengeance.  
The first block of noise terms (with $W_{i,t}$, $\tilde N_i$, $B^{H_1}_i$, $B^{\frac{1+H_2}{2}}_i$, $G_i$, $R^{H_2}_i$) are \emph{idiosyncratic}, independent across agents. The second block (superscript $0$) contains \emph{common} noises that hit the entire population simultaneously, for instance, a regional drought or a sudden collapse of the world gold price. All Brownian motions, fractional Brownian motions, Gauss–Volterra processes, Rosenblatt processes, and compensated Poisson random measures are mutually independent. The integrals with respect to fractional and related processes are defined pathwise as Riemann-Stieltjes integrals under suitable Hölder conditions.
$H^{i,g}(t,s,E_t^g,e_{\alpha_i}(\mu_t^g))$ : \textbf{environment-population coupling}.  This term captures how the macro‑environment and the aggregate distribution of all agents influence individual trajectories.  The expectile $e_{\alpha_i}(\mu_t^g)$ (defined in \S\ref{sec:expectile}) is a tail risk measure of the joint distribution; a high ${\alpha_i}$-expectile signals that extreme events (massacres, famine) are likely.  The function $H^{i,g}$ models phenomena such as: a rise in gold price ($E_t^{g,1}$) makes artisanal mining more profitable, drawing civilians towards conflict zones; a high expectile indicates severe instability, which may deter individuals from cooperating with the state.

$\mathcal{K}_{\text{rv}}^{i,g}(t)$ : \textbf{Volterra revenge operator} (see Definition~\ref{def:volterra} below).  This non‑Markovian drift encodes the long memory of historical grievances.  Past generations' traumas (massacres, land dispossessions, broken peace agreements) affect the current generation through a singular fractional kernel $(t-s)^{-\gamma}$ with $\gamma\in(0,1)$.  The smaller $\gamma$, the slower the memory fades, reflecting the deep‑rooted nature of Sahelian blood feuds. 

\begin{definition}[Volterra revenge operator] \label{def:volterra}
For each agent $i\in\cN$ and generation $g$, the historical revenge drift is given by the singular fractional integral
\begin{align}
\mathcal{K}_{\mathrm{rv}}^{i,g}(t) &=
\sum_{j=0}^{g-1} \int_{T_j}^{T_{j+1}} (t-s)^{-\gamma}\,
\mathbb{M}_i(s,s_s^j)\Bigl(\int \mathbf{z}\,d\mu_s^j(\mathbf{z})\Bigr)\,ds
\;+\; \int_{T_g}^{t} (t-s)^{-\gamma}\,
\mathbb{M}_i(s,s_s^g)\Bigl(\int \mathbf{z}\,d\mu_s^g(\mathbf{z})\Bigr)\,ds,
\label{eq:revenge}
\end{align}
where $\gamma\in(0,1)$ is the memory exponent and $\mathbb{M}_i(s,\mathfrak{s})$ is a bounded linear operator depending on the regime $\mathfrak{s}$ at time $s$.  
The inner integral $\int \mathbf{z}\,d\mu_s^\star(\mathbf{z})$ is the expectation of the joint state vector under the law of generation $\star$ at time $s$.  
Thus $\mathcal{K}_{\mathrm{rv}}^{i,g}$ aggregates the mean state of all past generations, weighted by the fractional kernel and the regime‑dependent grievance matrix.
\end{definition}
The operator $\mathcal{K}_{\mathrm{rv}}^{i,g}(t)$ is the embodiment of the \emph{ancestral blood feud}.  
In Mali, the memory of a massacre committed in a village  in 2012-2026, or the dispossession of  pastoralists during the droughts of the 1970s, does not fade exponentially with time.  
Instead, it decays slowly, like a power law $(t-s)^{-\gamma}$, remaining a potent mobilising force decades later.  
The kernel integrates over \emph{all} preceding generations, from the pre‑colonial empires through  colonisation, the rebellions of the 1960s and 1990s, up to the present day.  
The grievance matrix $\mathbb{M}_i(s,\mathfrak{s})$ specifies \emph{how} a particular historical trauma affects agent $i$.  
For a  farmer, $\mathbb{M}_i$ may amplify the memory of raiders burning granaries; for a  separatist, it may magnify the recollection of broken peace agreements.  
The regime $\mathfrak{s}$ modulates this effect: during active peace negotiations ($\mathfrak{s}=1$), the matrix may dampen the transmitted trauma as reconciliation efforts take hold; during full‑scale war ($\mathfrak{s}=3$), it may amplify grievances as each side recalls past betrayals.  
The term $\int \mathbf{z}\,d\mu_s^\star(\mathbf{z})$ is the mean state of the entire population-environment system at time $s$, meaning that the revenge impulse depends not only on individual memory but on the \emph{collective} trauma of the whole society.  
Thus, a community that was collectively traumatised will exert a stronger intergenerational pull on its members than one in which the trauma was isolated.  
The smaller the memory exponent $\gamma$, the slower the healing and the deeper the persistence of the conflict trap.

\medskip
\noindent\textbf{Measurement.}
Agent $i$ does not know $X_{i,t}^g$ exactly; she only receives a noisy measurement $X_{i,t}^m\in\mathbb R^d$, corrupted by the chaos of the field:
\begin{equation}
\begin{aligned}
dX_{i,t}^m &=
b^{m,p}\bigl(t,s_t^g,X_{i,t}^g,\mu_t^{m}\bigr)\,dt \\
&\quad +  \sigma^{w,m}\bigl(t,s_t^g,X_{i,t}^g,\mu_t^{m}\bigr)\,dW_{i,t}^m
   + \int_{\mathcal Z_m}\sigma^{n,m}\bigl(t,s_{t-}^g,X_{i,t-}^g,E_{t-}^g,\mu_{t-}^{m},z_m\bigr)\,\tilde N_i^m(dt,dz_m) \\
&\quad +  \sigma^{fbm,H_1,m}\bigl(t,s_t^g,X_{i,t}^g,\mu_t^{m}\bigr)\,dB^{H_1,m}_i(t)
   + 2\tilde{c}\sigma^{fbm,H_2,m}\bigl(t,s_t^g,X_{i,t}^g,\mu_t^{m}\bigr)\,dB^{\frac{1+H_2}{2},m}_i(t) \\
&\quad + \sigma^{gv,m}\bigl(t,s_t^g,X_{i,t}^g,\mu_t^{m}\bigr)\,dG_i^m(t)
   + \sigma^{R,m}\bigl(t,s_t^g,X_{i,t}^g,\mu_t^{m}\bigr)\,dR^{H_2,m}_i(t),
\end{aligned}
\label{eq:measure_dyn}
\end{equation}
where $\mu_t^{m}=\operatorname{Law}(X_{i,t}^m,s_t^g,\mathbf X_t^g,\mathbf C_t^g,E_t^g,\mathbf u_t)$.  
In Mali, this measurement corresponds to the fragmentary intelligence received by a fighter, a rumour of an approaching patrol, a report of a looted village, or an estimate of the gold price heard on a smuggled radio. The noises are agent‑specific and independent of the true‑state noises.

\medskip
\noindent\textbf{Observation.}
The agent processes the raw measurement into an observation $X_{i,t}^o$, filtering out some noise and combining it with prior knowledge:
\begin{equation}
\begin{aligned}
dX_{i,t}^o &=
 b^{o,p}\bigl(t,s_t^g,X_{i,t}^m,\mu_t^{o}\bigr)\,dt \\
&\quad +  \sigma^{w,p}\bigl(t,s_t^g,X_{i,t}^m,\mu_t^{o}\bigr)\,dW_{i,t}^o
   + \int_{\mathcal Z_o} \sigma^{n,o}\bigl(t,s_{t-}^g,X_{i,t-}^m,\mu_{t-}^{o},z_o\bigr)\,\tilde N_i^o(dt,dz_o) \\
&\quad + \sigma^{fbm,H_1,o}\bigl(t,s_t^g,X_{i,t}^m,\mu_t^{o}\bigr)\,dB^{H_1,o}_i(t)
   +  2\tilde{c}\sigma^{fbm,H_2,o}\bigl(t,s_t^g,X_{i,t}^m,\mu_t^{o}\bigr)\,dB^{\frac{1+H_2}{2},o}_i(t) \\
&\quad +\sigma^{gv,o}\bigl(t,s_t^g,X_{i,t}^m,\mu_t^{o}\bigr)\,dG_i^o(t)
   + \sigma^{R,o}\bigl(t,s_t^g,X_{i,t}^m,\mu_t^{o}\bigr)\,dR^{H_2,o}_i(t),
\end{aligned}
\label{eq:obs_dyn}
\end{equation}
with $\mu_t^{o}=\operatorname{Law}(X_{i,t}^o,s_t^g,\mathbf X_t^g,\mathbf C_t^g,E_t^g,\mathbf u_t)$.  
This step mirrors how a local commander sifts through contradictory reports to form a coherent picture of the battlefield.

\medskip
\noindent\textbf{Perception (belief).}
The agent builds a perception $X_{i,t}^p\in\mathbb R^d$ of her own state, a subjective belief that directly guides her choices:
\begin{equation}
\begin{aligned}
dX_{i,t}^p &=
b^{p}\bigl(t,s_t^g,X_{i,t}^o,u_{i,t},\mu_t^{p}\bigr)\,dt \\
&\quad + \sigma^{w,p}\bigl(t,s_t^g,X_{i,t}^o,\mu_t^{p}\bigr)\,dW_{i,t}^p
   + \int_{\mathcal Z_p}\sigma^{n,p}\bigl(t,s_{t-}^g,X_{i,t-}^o,\mu_{t-}^{p},z_p\bigr)\,\tilde N_i^p(dt,dz_p) \\
&\quad + \sigma^{fbm,H_1,p}\bigl(t,s_t^g,X_{i,t}^o,\mu_t^{p}\bigr)\,dB^{H_1,p}_i(t)
   +   2\tilde{c}\sigma^{fbm,H_2,p}\bigl(t,s_t^g,X_{i,t}^o,\mu_t^{p}\bigr)\,dB^{\frac{1+H_2}{2},p}_i(t) \\
&\quad + \sigma^{gv,p}\bigl(t,s_t^g,X_{i,t}^o,\mu_t^{p}\bigr)\,dG_i^p(t)
   + \sigma^{R,p}\bigl(t,s_t^g,X_{i,t}^o,\mu_t^{p}\bigr)\,dR^{H_2,p}_i(t),
\end{aligned}
\label{eq:perc_dyn}
\end{equation}
with $\mu_t^{p}=\operatorname{Law}(X_{i,t}^p,s_t^g,\mathbf X_t^g,\mathbf C_t^g,E_t^g,\mathbf u_t)$.  
The control  action $u_{i,t}$ enters the perception drift: an agent may invest resources to gather better intelligence (paying informants, sending scouts) thereby altering the evolution of her own belief. In the Malian theatre, this is the difference between a faction that relies on rumour and one that operates a network of trusted informers.

The four layers are coupled through the laws $\mu_t^g,\mu_t^{m},\mu_t^{o},\mu_t^{p}$, each of which is a probability measure on an extended state space. Because the noises are independent across layers, the agent's best estimate of her true state is the solution of a nonlinear filtering problem whose structure is fully determined by the coefficients in \eqref{eq:true_dyn}-\eqref{eq:perc_dyn}. In the MFTG, the agent optimises her cost functional on the basis of her perception $X_{i,t}^p$, while the true state $X_{i,t}^g$ determines the actual costs incurred, so that, exactly as in the field, a decision may be rational given the agent's beliefs yet prove disastrous in reality.

\subsection{Environmental dynamics}
\label{sec:env_dyn_four}

In the Malian conflict, no actor observes the macro‑environment perfectly. International gold prices, border porosity, rainfall patterns, and the level of foreign intervention are perceived only through incomplete and often distorted signals. To mirror the information structure of individual agents, the environment is endowed with an analogous four‑layer architecture. There is a single \emph{true} environment process that affects all agents, while each agent $i\in\cN$ forms her own \emph{measurement}, \emph{observation}, and \emph{perception} of this common process. All coefficients are $\mathbb{F}$-progressively measurable, satisfy Lipschitz and linear growth conditions, and depend on the macro‑political regime $s_t^g$.

\medskip
\noindent\textbf{True environment.}
The true environment $E_t^g\in\mathbb R^{d_E}$ evolves according to the McKean–Vlasov–Volterra jump‑diffusion
\begin{equation}
\begin{aligned}
dE_t^g &=
b_E\bigl(t,s_t^g,\mathbf X_t^g,\mathbf C_t^g,E_t^g,\mathbf u_t,\mu_t^g\bigr)\,dt \\[2pt]
&\quad + \sigma_E^{w}\bigl(t,s_t^g,\mathbf X_t^g,\mathbf C_t^g,E_t^g,\mu_t^g\bigr)\,dW_{E,t}
   + \int_{\mathcal Z_E}\gamma_E^{n}\bigl(t,s_{t-}^g,\mathbf X_{t-}^g,\mathbf C_{t-}^g,E_{t-}^g,\mu_{t-}^g,z\bigr)\,\tilde N_E(dt,dz) \\[2pt]
&\quad + \sigma_E^{fbm,H_1}\bigl(t,s_t^g,\mathbf X_t^g,\mathbf C_t^g,E_t^g,\mu_t^g\bigr)\,dB^{H_1}_t
   +  2\tilde{c} \sigma_E^{fbm,H_2}\bigl(t,s_t^g,\mathbf X_t^g,\mathbf C_t^g,E_t^g,\mu_t^g\bigr)\,dB^{\frac{1+H_2}{2}}_t \\[2pt]
&\quad + \sigma_E^{gv}\bigl(t,s_t^g,\mathbf X_t^g,\mathbf C_t^g,E_t^g,\mu_t^g\bigr)\,dG_t
   + \sigma_E^{R}\bigl(t,s_t^g,\mathbf X_t^g,\mathbf C_t^g,E_t^g,\mu_t^g\bigr)\,dR^{H_2}_t.
\end{aligned}
\label{eq:E_true_dyn}
\end{equation}
Here $\mu_t^g$ is the joint law of $(s_t^g,\mathbf X_t^g,\mathbf C_t^g,E_t^g,\mathbf u_t)$.  
The drift $b_E$ encapsulates mean‑reversion, linear feedback from individual states ($\sum_j M_j^g X_{j,t}^g$), and the non‑linear aggregate effect ($\int F^g d\mu_t^g$).  
All driving noises are \emph{common} to the entire population; for example, $W_{E,t}$ models continuous short‑term fluctuations in global commodity markets, $\tilde N_E$ captures sudden extreme events (sanctions, discovery of new gold deposits), the fractional Brownian motions $B^{H_1}$ and $B^{\frac{1+H_2}{2}}$ represent long‑memory components such as climatic cycles and slowly evolving financial conditions, the Gauss–Volterra process $G$ captures more flexible long‑memory phenomena, and the Rosenblatt process $R^{H_2}$ introduces heavy‑tailed asymmetric shocks typical of illicit financial flows.  
The integrals with respect to fractional and related processes are defined pathwise as Riemann–Stieltjes integrals under suitable Hölder conditions.

\medskip
\noindent\textbf{Environment measurement.}
Agent $i$ does not observe $E_t^g$ directly; she receives a noisy measurement $E_{i,t}^{m}\in\mathbb R^{d_E}$, corrupted by agent‑specific disturbances:
\begin{equation}
\begin{aligned}
dE_{i,t}^{m} &=
b_E^{m}\bigl(t,s_t^g,E_t^g,\mu_t^{E,m}\bigr)\,dt \\[2pt]
&\quad + \sigma_E^{w,m}\bigl(t,s_t^g,E_t^g,\mu_t^{E,m}\bigr)\,dW_{i,t}^{E,m}
   + \int_{\mathcal Z_{E,m}}\sigma_E^{n,m}\bigl(t,s_{t-}^g,E_{t-}^g,\mu_{t-}^{E,m},z_m\bigr)\,\tilde N_i^{E,m}(dt,dz_m) \\[2pt]
&\quad + \sigma_E^{fbm,H_1,m}\bigl(t,s_t^g,E_t^g,\mu_t^{E,m}\bigr)\,dB^{H_1,E,m}_{i}(t)
   + 2\tilde{c} \sigma_E^{fbm,H_2,m}\bigl(t,s_t^g,E_t^g,\mu_t^{E,m}\bigr)\,dB^{\frac{1+H_2}{2},E,m}_{i}(t) \\[2pt]
&\quad + \sigma_E^{gv,m}\bigl(t,s_t^g,E_t^g,\mu_t^{E,m}\bigr)\,dG_i^{E,m}(t)
   + \sigma_E^{R,m}\bigl(t,s_t^g,E_t^g,\mu_t^{E,m}\bigr)\,dR^{H_2,E,m}_{i}(t),
\end{aligned}
\label{eq:E_measure_dyn}
\end{equation}
where $\mu_t^{E,m}=\operatorname{Law}\bigl(E_{i,t}^{m},\,s_t^g,\,\mathbf X_t^g,\,\mathbf C_t^g,\,E_t^g,\,\mathbf u_t\bigr)$.  
In the field, this measurement corresponds to the fragmentary environmental information available to a rural community: a trader's report of the gold price in Bamako, a rumour of a closed border, or an estimate of the coming harvest based on local rainfall. The noises are independent across agents and independent of the true environment noises.

\medskip
\noindent\textbf{Environment observation.}
The agent processes the raw environmental measurement into an observation $E_{i,t}^{o}\in\mathbb R^{d_E}$, combining it with prior knowledge or local expertise:
\begin{equation}
\begin{aligned}
dE_{i,t}^{o} &=
b_E^{o}\bigl(t,s_t^g,E_{i,t}^{m},\mu_t^{E,o}\bigr)\,dt \\[2pt]
&\quad + \sigma_E^{w,o}\bigl(t,s_t^g,E_{i,t}^{m},\mu_t^{E,o}\bigr)\,dW_{i,t}^{E,o}
   + \int_{\mathcal Z_{E,o}}\sigma_E^{n,o}\bigl(t,s_{t-}^g,E_{i,t-}^{m},\mu_{t-}^{E,o},z_o\bigr)\,\tilde N_i^{E,o}(dt,dz_o) \\[2pt]
&\quad + \sigma_E^{fbm,H_1,o}\bigl(t,s_t^g,E_{i,t}^{m},\mu_t^{E,o}\bigr)\,dB^{H_1,E,o}_{i}(t)
   +  2\tilde{c}\sigma_E^{fbm,H_2,o}\bigl(t,s_t^g,E_{i,t}^{m},\mu_t^{E,o}\bigr)\,dB^{\frac{1+H_2}{2},E,o}_{i}(t) \\[2pt]
&\quad + \sigma_E^{gv,o}\bigl(t,s_t^g,E_{i,t}^{m},\mu_t^{E,o}\bigr)\,dG_i^{E,o}(t)
   + \sigma_E^{R,o}\bigl(t,s_t^g,E_{i,t}^{m},\mu_t^{E,o}\bigr)\,dR^{H_2,E,o}_{i}(t),
\end{aligned}
\label{eq:E_obs_dyn}
\end{equation}
with $\mu_t^{E,o}=\operatorname{Law}\bigl(E_{i,t}^{o},\,s_t^g,\,\mathbf X_t^g,\,\mathbf C_t^g,\,E_t^g,\,\mathbf u_t\bigr)$.  
This step mirrors how an experienced village elder interprets scattered reports to form a coherent picture of the macro‑situation, distinguishing a true shift in the gold price from transient noise, for instance.

\medskip
\noindent\textbf{Environment perception.}
The agent forms a perception $E_{i,t}^{p}\in\mathbb R^{d_E}$ of the environment, a subjective belief that directly guides her strategic choices:
\begin{equation}
\begin{aligned}
dE_{i,t}^{p} &=
b_E^{p}\bigl(t,s_t^g,E_{i,t}^{o},u_{i,t},\mu_t^{E,p}\bigr)\,dt \\[2pt]
&\quad + \sigma_E^{w,p}\bigl(t,s_t^g,E_{i,t}^{o},\mu_t^{E,p}\bigr)\,dW_{i,t}^{E,p}
   + \int_{\mathcal Z_{E,p}}\sigma_E^{n,p}\bigl(t,s_{t-}^g,E_{i,t-}^{o},\mu_{t-}^{E,p},z_p\bigr)\,\tilde N_i^{E,p}(dt,dz_p) \\[2pt]
&\quad + \sigma_E^{fbm,H_1,p}\bigl(t,s_t^g,E_{i,t}^{o},\mu_t^{E,p}\bigr)\,dB^{H_1,E,p}_{i}(t)
   + 2\tilde{c} \sigma_E^{fbm,H_2,p}\bigl(t,s_t^g,E_{i,t}^{o},\mu_t^{E,p}\bigr)\,dB^{\frac{1+H_2}{2},E,p}_{i}(t) \\[2pt]
&\quad + \sigma_E^{gv,p}\bigl(t,s_t^g,E_{i,t}^{o},\mu_t^{E,p}\bigr)\,dG_i^{E,p}(t)
   + \sigma_E^{R,p}\bigl(t,s_t^g,E_{i,t}^{o},\mu_t^{E,p}\bigr)\,dR^{H_2,E,p}_{i}(t),
\end{aligned}
\label{eq:E_perc_dyn}
\end{equation}
with $\mu_t^{E,p}=\operatorname{Law}\bigl(E_{i,t}^{p},\,s_t^g,\,\mathbf X_t^g,\,\mathbf C_t^g,\,E_t^g,\,\mathbf u_t\bigr)$.  
Notice that the continuous control $u_{i,t}$ appears in the perception drift: an agent may devote resources to improving her environmental intelligence, deploying scouts to monitor borders, paying for satellite imagery, or cultivating informants in neighbouring countries. In the Malian theatre, this is the difference between a faction that acts on reliable intelligence and one that relies on rumour and propaganda.
The four environmental layers are coupled through the laws $\mu_t^g$, $\mu_t^{E,m}$, $\mu_t^{E,o}$, $\mu_t^{E,p}$, each of which is a probability measure on an extended state space. Because the noises are independent across layers, each agent solves a nonlinear filtering problem to estimate the true environment from her perception. The agents' decisions are based on their idiosyncratic environmental perceptions $E_{i,t}^{p}$, while the actual payoffs depend on the true environment $E_t^g$, exactly as in the field, where a strategy that appears sound given local information may prove disastrous when the true state of the world is revealed.

\subsection{Controlled coalition switching}
\label{sec:coalition}

In the Malian conflict, actors continuously re‑evaluate their allegiances. Fighters defect from one armed group to another, community militias dissolve and reform under different leadership, and opportunistic entrepreneurs pivot between criminal networks and political factions. To capture this fluidity, we model the discrete coalition affiliation $C_{i,t}^g \in \{1,\dots,K\}$ of each agent $i\in\cN$ within generation $g$ as a controlled continuous‑time Markov jump process. The control variable is a vector $v_{i,t}^g \in \mathcal{V}^{i,g} \subseteq \R^{q_v}$ that agent $i$ chooses; it may represent effort spent on building new alliances, the payment of bribes, or investment in loyalty‑shifting operations.

Let the \emph{joint coalition vector} at time $t$ be $\mathbf{C}_t^g = (C_{1,t}^g,\dots,C_{N,t}^g)$. Conditional on the regime $s_t^g$, the full continuous state vector $\mathbf{X}_t^g$, the environment $E_t^g$, and the joint law $\mu_t^g$, the transitions of different agents are independent. For an agent $i$ currently in coalition $m$, the instantaneous probability of switching to coalition $m' \neq m$ during the interval $[t,t+dt]$ is
\begin{equation}
\Pee\bigl( C_{i,t+dt}^g = m' \;\big|\; C_{i,t}^g = m,\; \mathcal{H}_t \bigr)
= \lambda_{m\to m'}^i\bigl(t,s_t^g,\mathbf{X}_t^g,E_t^g,\mu_t^g,v_{i,t}^g\bigr)\,dt + o(dt),
\label{eq:switch_prob}
\end{equation}
where $\mathcal{H}_t$ denotes the history of the system up to time $t$. The \emph{transition intensity} is specified as
\begin{equation}
\lambda_{m\to m'}^i\bigl(t,s,\mathbf{x},e,\mu,v\bigr)
= \bar{\lambda}_{m\to m'}^{\,i}(t,s)\,
\exp\!\Bigl( v^\top D_{mm'}^i(t,s)\,\mathbf{x}_i \;+\; \Psi_{mm'}\bigl(e,\, e_{\beta_i}(\mu), s\bigr) \Bigr),
\label{eq:lambda_def}
\end{equation}
with the convention $\lambda_{m\to m}^i = -\sum_{m'\neq m} \lambda_{m\to m'}^i$. The baseline intensity $\bar{\lambda}_{m\to m'}^{\,i}(t,s) \ge 0$ captures the structural propensity of switching from $m$ to $m'$ under regime $s$ in the absence of active control. The matrix $D_{mm'}^i(t,s) \in \R^{q_v \times d}$ determines how the agent's own continuous state (radicalisation, wealth) modulates the attractiveness of the switch; for example, a highly radicalised individual may find it easier to join a jihadist network. The function $\Psi_{mm'}$ couples the environment and the $\beta_i$-expectile $e_{\beta_i}(\mu_t^g)$ of the joint distribution to the switching intensity: a high expectile signals an extremely unstable situation, which typically lowers the threshold for changing sides.

The control $v_{i,t}^g$ enters multiplicatively through the exponential factor. This functional form ensures that the intensity remains strictly positive and that the marginal cost of control, which will appear in the agent's running cost, is convex. The admissible control set $\mathcal{V}^{i,g}$ is a closed, convex subset of $\R^{q_v}$, and the agent selects a $\mathbb{F}$-progressively measurable process $v_{i,\cdot}^g$ with values in this set.
The control action profile $\mathbf{v}_t^g = (v_{1,t}^g, \dots, v_{N,t}^g)$ represents the highly fluid, strategic manipulation of alliance networks and political allegiances across the deeply fragmented Malian theater. Far from being a static or ideological commitment, allegiance in the Sahel is a highly volatile currency, and the control $v_{i,t}^g$ materializes in the field as the real-time financial, social, and tactical investments made by individual actors to navigate or actively induce defections within the continuous-time Markov jump process. For an opportunistic combatant or local commander in the Sikasso, Menaka or Gao regions, a targeted shift in $v_{i,t}^g$ represents concrete field actions: the clandestine distribution of cash bribes via Hawala informal networks, the renegotiation of tribal security guarantees over tea in remote desert encampments, or the tactical reinvestment of artisanal gold rents to secure protection pacts with whichever coalition holds localized dominance. This control profile explicitly captures the mechanics behind the dizzying structural fluidity observed on the ground such as fighters defecting overnight between secular  nationalist factions  and deep-seated jihadist coalitions, or community-level  militias splintering and reforming under entirely new leadership structures depending on localized survival payoffs. Because these loyalty-shifting operators act as the direct velocity variable altering the instantaneous transition intensity $\lambda_{m\to m'}^i$, the collective profile $\mathbf{v}_t^g$ functions as the primary behavioral steering mechanism of the game. It  formalizes how the top-down introduction of blunt regulatory shocks or heavy-handed military campaigns by the state alters the underlying joint law $\mu_t^g$, driving individual agents to dynamically adjust their allegiance controls $v_{i,t}^g$ to shift their entire coalitional footprint away from predatory elements or, conversely, deeper into hidden criminal and insurgent networks to hedge against systemic risk. The probabilistic description above is equivalent to saying that each agent's coalition process is a controlled, non‑homogeneous Markov chain with generator matrix $Q^i(t) = \bigl(q^i_{mm'}(t)\bigr)$ given by $q^i_{mm'}(t) = \lambda_{m\to m'}^i(\dots)$. The joint process $\mathbf{C}_t^g$ is a $K^N$‑state Markov chain, and its compensator measures can be expressed in terms of the individual intensities. In particular, for each agent $i$ and each ordered pair $(m,m')$ with $m \neq m'$, we introduce the point process $N_{i,t}^{C, m\to m'}$ counting the transitions from $m$ to $m'$ up to time $t$. Its $\mathbb{F}$‑compensator is $\int_0^t \lambda_{m\to m'}^i(\dots)\,ds$, and the compensated process $\tilde N_{i,t}^{C, m\to m'} = N_{i,t}^{C, m\to m'} - \int_0^t \lambda_{m\to m'}^i(\dots)\,ds$ is an $\mathbb{F}$‑martingale. An important structural feature of the Malian conflict is the endogenous emergence and collapse of coalitions. If at some instant $\mu_t^g(C_{i,t}^g = m) = 0$ for all $i\in\cN$, i.e., no agent belongs to coalition $m$, then coalition $m$ is declared \emph{extinct}. All intensities targeting that coalition are set to zero, and the state space is reduced accordingly. This allows the effective number of active coalitions, denoted $K_t^g$, to fluctuate between $1$ (a unified state) and the maximum $K$. New coalitions cannot appear ex nihilo; they must be seeded by an existing agent switching into a previously empty label, which can occur only if the intensity $\lambda_{m\to m'}^i$ becomes positive for some active $m$ and empty $m'$, reflecting a deliberate act of creation (e.g., the formation of a new splinter faction).

Thus, the controlled coalition dynamics capture the full spectrum of allegiance shifts observed in the Sahel: opportunistic defections, coerced recruitment, strategic realignments driven by environmental changes, and the birth and death of armed organisations as the conflict evolves.

\subsection{Intergenerational Transmission of Trauma and Wealth}
\label{sec:transmission}

One of the most persistent features of the Malian conflict is its intergenerational continuity. The children of today's combatants inherit not only the psychological scars and social grievances of their parents but also, in the case of war entrepreneurs, the financial capital accumulated through decades of violence. To capture these hereditary mechanisms, we specify the transition rules that map the terminal state of generation $g$ onto the initial state of generation $g+1$ at each generational interface $T_{g+1}$.

Let $\mu_{T_{g+1}-}^g$ denote the joint probability law of the system immediately before the generational boundary. The initial configuration of generation $g+1$ is determined by a set of non‑local jump operators that act on the state variables and on the measure itself. All random innovations are independent of the past and have zero mean, ensuring that the mapping preserves the martingale structure of the forward dynamics.

\medskip
\noindent\textbf{Continuous state transmission.}
For each agent $i\in\cN$, the initial continuous state of the new generation is given by
\begin{equation}
X_{i,T_{g+1}}^{g+1} \;=\;
\zeta_i\, X_{i,T_{g+1}-}^g
\;+\; (1-\zeta_i)\,\int \mathbf{x}_i\; d\mu_{T_{g+1}-}^g(\mathbf{x},\mathbf{c},e)
\;+\; \xi_{i}^{g+1},
\label{eq:X_jump}
\end{equation}
where:
\begin{itemize}
\item $\zeta_i \in [0,1]$ is the \emph{intergenerational trauma transmission coefficient}. For $\zeta_i = 1$, the child perfectly inherits the parent's radicalisation and asset holdings; for $\zeta_i = 0$, the child starts from the population average, free of parental influence. In practice, $\zeta_i$ is high for households that have experienced severe trauma and for war entrepreneurs who bequeath offshore wealth, while it is lower for communities with strong social mixing.
\item The integral $\int \mathbf{x}_i\, d\mu_{T_{g+1}-}^g(\mathbf{x},\mathbf{c},e)$ is the expectation of agent $i$'s state under the terminal law of generation $g$, i.e., the population mean of that generation. This term represents the ``melting pot'' effect: absent direct inheritance, the child tends towards the societal norm.
\item $\xi_i^{g+1} \in \R^d$ is a zero‑mean random innovation, capturing idiosyncratic shocks at birth (e.g., an unusually gifted leader emerges, or a child is born into extraordinary hardship). These shocks are independent across agents and independent of all previous history.
\end{itemize}

\medskip
\noindent\textbf{Coalition inheritance.}
The coalition affiliation of a child is drawn from a discrete distribution that depends on the parent's coalition and the overall measure of the preceding generation:
\begin{equation}
\Pee\bigl( C_{i,T_{g+1}}^{g+1} = m' \;\big|\; \mathbf{C}_{T_{g+1}-}^g = \mathbf{c},\;
\mu_{T_{g+1}-}^g \bigr) \;=\; \Pi_{m'}^{i}\bigl(\mathbf{c},\, \mu_{T_{g+1}-}^g\bigr),
\label{eq:C_jump}
\end{equation}
where $\Pi_{m'}^{i} \ge 0$ and $\sum_{m'=1}^K \Pi_{m'}^{i} = 1$. The inheritance kernel $\Pi^{i}$ can encode a variety of transmission patterns: a strong direct transmission ($\Pi_{m'}^{i}$ close to $1$ when $m'$ matches the parent's coalition), a tendency towards radicalisation when the population expectile is high, or a probabilistic fusion of allegiances in mixed communities.

\medskip
\noindent\textbf{Environmental persistence.}
The macro‑environment carries over across generations with a persistence matrix $\mathbb{A}_E \in \R^{d_E \times d_E}$:
\begin{equation}
E_{T_{g+1}}^{g+1} \;=\; \mathbb{A}_E\, E_{T_{g+1}-}^g \;+\; \xi_E^{g+1},
\label{eq:E_jump}
\end{equation}
where $\xi_E^{g+1} \in \R^{d_E}$ is a zero‑mean random innovation independent of the past. The matrix $\mathbb{A}_E$ can be diagonal, with entries close to $1$ for slowly evolving environmental variables (e.g., climate, infrastructure) and entries near $0$ for highly volatile ones (e.g., aid flows, temporary policies).

\medskip
\noindent\textbf{Regime continuity.}
The macro‑political regime $s_t^g$ is assumed to be a continuous process across generational boundaries, reflecting the fact that the broader political environment (peace negotiations, military intervention, etc.) does not reset simply because a new cohort reaches adulthood:
\begin{equation}
s_{T_{g+1}}^{g+1} \;=\; s_{T_{g+1}-}^g.
\label{eq:regime_jump}
\end{equation}
A probabilistic reset could be easily incorporated by replacing the identity map with a stochastic kernel; however, for analytical clarity we maintain continuity.

\medskip
\noindent\textbf{Push‑forward of the probability measure.}
Taking the joint law of the random variables defined in \eqref{eq:X_jump}-\eqref{eq:E_jump} yields a deterministic operator $\mathcal{T}_{\boldsymbol{\zeta}}$ that maps the terminal measure of generation $g$ to the initial measure of generation $g+1$:
\begin{equation}
\mu_{T_{g+1}}^{g+1} \;=\; \mathcal{T}_{\boldsymbol{\zeta}}\bigl[\,\mu_{T_{g+1}-}^g\,\bigr],
\label{eq:pushforward}
\end{equation}
where the subscript $\boldsymbol{\zeta} = (\zeta_1,\dots,\zeta_N)$ highlights the dependence on the vector of trauma transmission coefficients. The operator $\mathcal{T}_{\boldsymbol{\zeta}}$ acts on the infinite‑dimensional Wasserstein space of probability measures and encapsulates the full intergenerational transition. It is this push‑forward that couples successive generations and makes the game genuinely intergenerational: the equilibrium strategy in generation $g$ must anticipate how the terminal distribution will be transformed and how it will influence the initial conditions and strategic incentives of generation $g+1$.

\subsection{War Entrepreneurs and the Dual Local‑Global Asset Structure}
\label{sec:war_entrepreneurs}

A defining feature of the Malian war economy is the existence of a class of actors who derive systematic profit from instability and who use the global financial system to preserve and multiply those profits across generations.  We designate this subset as $\cN_E \subset \cN$.  Unlike ordinary civilians or state forces, a war entrepreneur holds a dual portfolio: \emph{local (illicit) wealth}, which is directly exposed to the conflict theatre, and \emph{global (laundered) wealth}, which is sheltered in offshore financial centres and can be bequeathed to the next generation with minimal loss.

\subsubsection{Local wealth dynamics}
The local wealth of entrepreneur $i \in \cN_E$ is a sub‑vector $X_{i,t}^{g,\text{local}} \in \R^{d_{\text{local}}}$ of her full continuous state $X_{i,t}^g$.  It encompasses tangible assets such as gold stocks, weapons caches, local currency, livestock, and territory under control.  By construction, this sub‑vector obeys exactly the same four‑layer McKean-Vlasov-Volterra dynamics described in Section~\ref{sec:fourlayer}.  In particular, its true component evolves according to \eqref{eq:true_dyn} with drift $b_i$, environmental coupling $H^{i,g}$, the revenge operator $\mathcal{K}_{\mathrm{rv}}^{i,g}$, and all idiosyncratic and common noise terms.  The coefficients pertaining to the local wealth dimensions may be subscripted by ``$\text{local}$'' to emphasise their economic content, but they remain structurally identical to those of any other agent.  Thus, the radicalisation, asset accumulation, and human capital of a war entrepreneur are driven by the same forces that shape the entire population: the macro‑environment $E_t^g$, the expectile $e_{\alpha_i}(\mu_t^g)$, the historical grievance memory, and the strategic choices of all actors.

\subsubsection{Global wealth and the laundering pipeline}
The global wealth $X_{i,t}^{g,\text{global}} \in \R$ is held in offshore bank accounts, foreign real estate, or international bonds, and is largely insulated from local conflict risks.  Unlike the local components, global wealth is \emph{exactly observed} by the entrepreneur (there is no measurement noise obscuring an offshore balance) and its dynamics are driven by a single Brownian motion independent of the local noises.  The process satisfies the stochastic differential equation
\begin{equation}
\begin{aligned}
dX_{i,t}^{g,\text{global}} &=
\Bigl[\, r_{\text{global}}\, X_{i,t}^{g,\text{global}}
      + \eta_{\text{wash}}\;\Pi_{\text{local}}\bigl(X_{i,t}^{g,\text{local}},\, E_t^g,\, e_{\beta_i}(\mu_t^g),\, s_t^g\bigr)
      - u_{i,t}^{g,\text{global}} \Bigr] dt \\[4pt]
&\quad + \sigma_{\text{global}}\, X_{i,t}^{g,\text{global}} \; dW_{i,t}^{g,\text{global}},
\end{aligned}
\label{eq:global_wealth}
\end{equation}
where
\begin{itemize}
    \item $r_{\text{global}} \ge 0$ is the risk‑free or market rate of return earned on global assets;
    \item $\Pi_{\text{local}}(\cdot) \ge 0$ is the \emph{local extraction profit function}, which quantifies the instantaneous monetary profit, expressed in internationally convertible form, that the entrepreneur can extract from local chaos.  It increases with local wealth and with environmental factors favourable to trafficking (high gold price, porous borders), but is penalised by a high $\beta_i$-expectile $e_{\beta_i}(\mu_t^g)$, which signals that the overall situation is so extreme that even the entrepreneur's operations are threatened;
    \item $\eta_{\text{wash}} \in [0,1]$ is the \emph{laundering efficiency}: the fraction of local profits that can be successfully transferred into the global financial system without detection or seizure;
    \item $u_{i,t}^{g,\text{global}} \ge 0$ is a \emph{consumption/repatriation control}, representing the amount of global wealth the entrepreneur chooses to withdraw for personal consumption or to reinject into the local conflict (e.g.\ to purchase weapons, bribe officials, or finance recruitment).  This control is a component of the full continuous control vector $u_{i,t}^g$ and is optimised alongside the violence‑related actions;
    \item $\sigma_{\text{global}} > 0$ is the volatility of global financial markets, and $W_{i,t}^{g,\text{global}}$ is an idiosyncratic Brownian motion independent of all other noises in the model.
\end{itemize}

\subsubsection{Intergenerational transmission of global wealth}
At the generational boundary $T_{g+1}$, the global wealth is transmitted to the offspring according to the same hereditary rule \eqref{eq:X_jump} that governs all continuous state variables.  For a war entrepreneur, the trauma transmission coefficient $\zeta_i$ corresponding to the global wealth component is typically close to unity, reflecting the fact that offshore accounts and property titles can be bequeathed with minimal dissipation.  Consequently, the next generation of war entrepreneurs enters the conflict with a substantial financial head start, perpetuating the structural advantage of the entrepreneurial class.

\subsubsection{The self‑reinforcing war economy loop}
The coupling between \eqref{eq:true_dyn} (for the local wealth components) and \eqref{eq:global_wealth} creates a vicious cycle that lies at the heart of the Malian conflict trap:
\begin{enumerate}
    \item Local instability, fuelled by the revenge operator $\mathcal{K}_{\mathrm{rv}}^{i,g}$ and by strategic violence, raises the extraction profit $\Pi_{\text{local}}$, generating large illicit revenues.
    \item A fraction $\eta_{\text{wash}}$ of these revenues is laundered into the safe global wealth process, where it compounds at the rate $r_{\text{global}}$ beyond the reach of local authorities.
    \item At the generational boundary, the accumulated global wealth is transferred to the next generation with high fidelity ($\zeta_i \approx 1$).
    \item A portion of this inherited capital is repatriated through $u_{i,t}^{g,\text{global}}$ to finance further violence, bribe officials, and maintain the environment $E_t^g$ in a state of high instability.
    \item The cycle repeats, making peace economically unattractive for the entrepreneurial class and structurally locking the system into the war trap whose threshold condition is identified in Theorem~\ref{thm:breakdown}.
\end{enumerate}
Thus, even if local grievances could be temporarily calmed, the recycling of global capital ensures that the peaceful steady state remains unstable.  The optimal institutional policy designed in Section~\ref{sec:policy} directly attacks this loop by imposing an offshore transaction tax $\theta_{\text{tax}}$ that reduces the effective laundering efficiency to $\eta_{\text{wash}} - \theta_{\text{tax}}$, and by investing in endogenous mediation subsidies that erode the grievance kernel $\mathbb{M}(t)$.  When $\theta_{\text{tax}} \ge \eta_{\text{wash}}$, the financial artery of the war economy is severed, and the peaceful equilibrium becomes globally stable.

\subsection{Expectile Risk Measures}
\label{sec:expectile}

In a polycentric conflict, actors do not treat gains and losses symmetrically. The loss of a village held for generations is felt far more intensely than the satisfaction derived from capturing an equivalent territory; the risk of a catastrophic defeat weighs more heavily than the possibility of a modest victory. Standard risk measures such as variance or Value‑at‑Risk fail to capture this asymmetry. We therefore employ \emph{expectiles}, a class of asymmetric least‑squares risk functionals that are coherent for $\alpha \ge 1/2$ and possess a simple, interpretable structure.

For a scalar random variable $Y$ with law $\mu$, the $\alpha$-expectile $e_\alpha(\mu)$, $\alpha \in (0,1)$, is defined as the unique solution $m$ of the equation
\begin{equation}
\alpha\,\E[(Y-m)_+] \;=\; (1-\alpha)\,\E[(m-Y)_+].
\label{eq:expectile_def}
\end{equation}
Equivalently, $e_\alpha(\mu)$ is the minimiser of the asymmetric quadratic loss $\E\bigl[|\alpha - \mathbf{1}_{\{Y \le m\}}|\,(Y-m)^2\bigr]$. For $\alpha = 1/2$, the expectile reduces to the expectation (mean). For $\alpha > 1/2$, it overweights the upper tail (gains/positive outcomes), while for $\alpha < 1/2$, it overweights the lower tail (losses). Expectiles depend continuously on $\alpha$ and are strictly increasing with $\alpha$; they are law‑invariant and respect first‑order stochastic dominance.

\begin{definition}[$\tau$-Expectile of a Random Vector]
Let
$
Z=(X,C,E)\in\mathbb{R}^{n},
\ 
n=d_X+d_C+d_E,
$
be a square-integrable random vector with law
$
\mu:=\mathcal{L}(Z)
\in \mathcal{P}_2(\mathbb{R}^{n}),
$
and let
$
\tau=(\tau_1,\ldots,\tau_n)\in(0,1)^n.
$

The $\tau$-expectile of $Z$, denoted by
$
e_\tau(Z)\in\mathbb{R}^{n},
$
is defined as the unique minimizer
\[
e_\tau(Z)
=
\arg\min_{m\in\mathbb{R}^{n}}
\;
\mathbb{E}
\!\left[
\sum_{k=1}^{n}
\ell_{\tau_k}(Z_k-m_k)
\right],
\]
where
\[
\ell_{\tau_k}(u)
=(
\tau_k\,\mathbf 1_{\{u\ge0\}}
+
(1-\tau_k)\,\mathbf 1_{\{u<0\}}) u^2.
\]
\end{definition}

In the present model, each agent $i$ is characterised by two expectile levels, $\alpha_i, \beta_i \in (0,1)$, which may differ across agents to reflect heterogeneous risk attitudes. Their roles are as follows:
\begin{itemize}
    \item \textbf{$\alpha_i$-expectile $e_{\alpha_i}(\mu_t^g)$:} This appears in the environment–population coupling term $H^{i,g}\bigl(t,s_t^g,E_t^g,e_{\alpha_i}(\mu_t^g)\bigr)$ in the drift of the true individual state \eqref{eq:true_dyn} and in the running cost functional \eqref{eq:running}. It captures the agent's perception of the ``central tendency'' of the population distribution, tilted towards gains or losses according to $\alpha_i$. For an agent with $\alpha_i > 1/2$, the expectile lies above the mean, meaning she bases her decisions on an optimistic view of the distribution, whereas $\alpha_i < 1/2$ corresponds to a pessimistic, loss‑sensitive attitude. In the Malian context, a government commander ($\alpha_i \approx 0.4$) may plan for worst‑case scenarios, while a risk‑seeking war entrepreneur ($\alpha_i \approx 0.7$) may focus on the upside potential of chaos.
    \item \textbf{$\beta_i$-expectile $e_{\beta_i}(\mu_t^g)$:} This appears in the coalition switching intensities \eqref{eq:lambda_def} and in the asymmetric penalty term $\Phi_{\beta_i}^{i,g}$ of the running cost \eqref{eq:running}. It measures tail risk and is typically chosen with $\beta_i$ close to 0 or 1 to focus on extreme events. A high value of $\beta_i$ (e.g., $0.9$) makes the expectile sensitive to the upper tail: if the population distribution has a heavy right tail (many highly radicalised individuals or very wealthy entrepreneurs), the expectile spikes, and this sharply increases the attractiveness of extremist coalitions (via $\Psi_{mm'}$) and the cost penalty $\Phi_{\beta_i}^{i,g}$. Conversely, a low $\beta_i$ (e.g., $0.1$) emphasises catastrophic lower‑tail events such as widespread famine or state collapse. In the Sahel, this mechanism captures the empirical regularity that extremist recruitment surges precisely when the population experiences an extreme shock, whether positive (a gold rush) or negative (a drought).
\end{itemize}
Both expectiles are computed on the true joint law $\mu_t^g$ of the system, which evolves according to the forward dynamics. In the mean‑field‑type game, each agent's control influences $\mu_t^g$, creating a feedback loop: a commander who invests heavily in violence shifts the distribution rightward, raising the $\beta$-expectile, which in turn makes defection to insurgent groups more attractive for other agents.

\subsection{Individual Payoff Functionals}
\label{sec:payoff}

Each agent $i \in \cN$ within generation $g$ is a rational, forward‑looking decision maker who chooses her continuous control $u_i^g = (u_{i,t}^g)_{t\in[T_g,T_{g+1}]}$ and her discrete coalition switching control $v_i^g = (v_{i,t}^g)_{t\in[T_g,T_{g+1}]}$ to minimise an expected total cost over her active lifetime. The controls are required to be $\mathbb{F}^i$‑adapted, where $\mathbb{F}^i$ is the filtration generated by the agent's perceptions $\{X_{i,t}^p, E_{i,t}^p, s_t^g\}$, that is, she can only condition her actions on her subjective beliefs, not on the unobserved true states. However, the costs she actually incurs depend on the true state variables $X_{i,t}^g$ and $E_t^g$, which she does not observe perfectly.

The cost functional for agent $i$ in generation $g$ is
\begin{equation}
J^{i,g}\bigl(u_i^g, v_i^g;\, \mathbf{u}_{-i}^g, \mathbf{v}_{-i}^g\bigr)
\;=\; \E\Bigl[ \int_{T_g}^{T_{g+1}} L^{i,g}\bigl(t,s_t^g,\mathbf{X}_t^g,\mathbf{C}_t^g,E_t^g,u_{i,t}^g,v_{i,t}^g,\mu_t^g\bigr)\,dt
\;+\; G^{i,g}\bigl(\mathbf{Z}_{T_{g+1}-}^g,s_{T_{g+1}-}^g,\mu_{T_{g+1}-}^g\bigr) \Bigr],
\label{eq:cost}
\end{equation}
where the expectation is taken under the true probability measure, given the initial conditions of generation $g$. The notation $\mathbf{u}_{-i}^g$ and $\mathbf{v}_{-i}^g$ denotes the control processes of all agents except $i$, whose choices affect the joint measure $\mu_t^g$ and thus the environment faced by agent $i$.

The \textbf{running cost} $L^{i,g}$ is specified as
\begin{align}
L^{i,g}(t,s,\mathbf{x},\mathbf{c},e,u,v,\mu) &=
\frac12\, u^\top R^{i,g}(t,s)\,u
\;+\; \frac12\, v^\top S^{i,g}(t,s)\,v \nonumber\\[4pt]
&\quad + \frac12\Bigl(\mathbf{x}_i - M^{i,g}(t,s)\, e_{\alpha_i}(\mu)\Bigr)^\top
   Q^{i,g}(t,s)\Bigl(\mathbf{x}_i - M^{i,g}(t,s)\, e_{\alpha_i}(\mu)\Bigr) \nonumber\\[4pt]
&\quad + \Phi_{\beta_i}^{i,g}\bigl(\mathbf{x}_i,\, e,\, e_{\beta_i}(\mu),\, s\bigr).
\label{eq:running}
\end{align}
Each term has a concrete behavioural interpretation grounded in the Malian conflict:

\begin{itemize}
    \item \textbf{Quadratic control costs:} $\frac12 u^\top R^{i,g} u$ and $\frac12 v^\top S^{i,g} v$. The matrices $R^{i,g}(t,s) \succ 0$ and $S^{i,g}(t,s) \succ 0$ are symmetric positive definite and regime‑dependent. The term $u^\top R u$ penalises continuous actions such as violence investment, capital flight, or intelligence gathering. The term $v^\top S v$ penalises effort spent on switching coalitions, paying bribes, building new alliances, or breaking old ones. During a ceasefire ($s=1$), these matrices may have larger eigenvalues, reflecting the higher political and logistical cost of militarisation; during full‑scale war ($s=3$), controls become cheaper as violence is normalised.
    
    \item \textbf{Conformity to the population norm:} The term with $Q^{i,g}(t,s) \succeq 0$ penalises deviations of the agent's own true state $\mathbf{x}_i$ from a target $M^{i,g}(t,s)\, e_{\alpha_i}(\mu)$. The matrix $M^{i,g}$ maps the $\alpha_i$-expectile of the joint distribution into a reference point in the agent's state space. An agent with $\alpha_i$ near $1/2$ tries to stay close to the population mean; an agent with $\alpha_i > 1/2$ aims to be above the mean (aspirational), while $\alpha_i < 1/2$ encourages being below the mean (risk‑avoidance). In the Malian context, this captures phenomena such as: a herder who sees most of his community joining a self‑defence militia feels pressure to do the same; a gold miner whose output falls below the local expectile may be stigmatised or exploited.
    
    \item \textbf{Asymmetric tail penalty:} $\Phi_{\beta_i}^{i,g}\bigl(\mathbf{x}_i,\, e,\, e_{\beta_i}(\mu),\, s\bigr)$ is a convex, non‑negative function that penalises the agent when her own state or the environment is in the extreme tail, as measured by the $\beta_i$-expectile. For example, if $\beta_i$ is close to $1$, the function becomes large when the agent is highly radicalised \emph{and} the population's upper expectile is also high, reflecting the increased risk of being targeted by counter‑insurgency operations when one stands out as an extremist. If $\beta_i$ is close to $0$, the penalty triggers when the agent's health or wealth drops below a catastrophic threshold. This term introduces a non‑linear, asymmetric risk sensitivity that cannot be captured by the quadratic deviation alone.
\end{itemize}
The \textbf{terminal cost} $G^{i,g}$ is evaluated at the end of the generation and reflects the bequest motive: agents care about the state in which they leave the world to their children. Its general form is analogous to the running cost, but without the control terms:
\begin{equation}
G^{i,g}\bigl(\mathbf{z},s,\mu\bigr)
\;=\; \frac12\Bigl(\mathbf{x}_i - M^{i,g}_{\text{term}}(s)\, e_{\alpha_i}(\mu)\Bigr)^\top
   Q^{i,g}_{\text{term}}(s)\Bigl(\mathbf{x}_i - M^{i,g}_{\text{term}}(s)\, e_{\alpha_i}(\mu)\Bigr)
\;+\; \Phi_{\beta_i,\text{term}}^{i,g}\bigl(\mathbf{x}_i,\, e,\, e_{\beta_i}(\mu),\, s\bigr).
\label{eq:terminal}
\end{equation}
The terminal matrices $M^{i,g}_{\text{term}}$ and $Q^{i,g}_{\text{term}}$ allow the target and penalty to differ from the running cost, capturing the idea that the end‑of‑life evaluation may be more forward‑looking (e.g., parents care especially about the radicalisation level they pass on). For a war entrepreneur, the terminal cost on global wealth is small ($\zeta_i \approx 1$ ensures almost perfect transmission), whereas for a civilian, a high terminal cost on health and assets may incentivise peace.
The combination of the running and terminal costs, with their dependence on the $\alpha_i$‑ and $\beta_i$‑expectiles of the joint distribution, creates a rich strategic context in which each agent's optimal behaviour is shaped by her individual risk preferences, by the observed and perceived states, and by the aggregate behaviour of the entire population.

\subsection{The  Intergenerational  Volterra MFTG}
\label{sec:complete_game}

Assembling the components laid out in the preceding sections, we obtain a closed, self‑consistent description of the Malian conflict as an intergenerational MFTG  with imperfect information, heterogeneous risk attitudes, and a comprehensive noise structure. The  object under study is the piecewise sequential game defined by the following primitives.

\begin{itemize}
    \item \textbf{Generational partition.} The macroscopic horizon is divided into $G+1$ epochs $[T_0,T_1) \to [T_1,T_2) \to \cdots \to [T_G,T_{G+1}]$ \eqref{eq:horizon}, each corresponding to an active generation of several years.
    
    \item \textbf{Four‑layer individual dynamics.} Every agent $i\in\cN$ possesses a true continuous state $X_{i,t}^g$ \eqref{eq:true_dyn} that is unobserved. She receives a noisy measurement $X_{i,t}^m$ \eqref{eq:measure_dyn}, processes it into an observation $X_{i,t}^o$ \eqref{eq:obs_dyn}, and finally forms a subjective perception $X_{i,t}^p$ \eqref{eq:perc_dyn} that serves as the basis for her decisions. The true dynamics incorporate the Volterra revenge operator $\mathcal{K}_{\mathrm{rv}}^{i,g}(t)$ \eqref{eq:revenge}, which encodes the long memory of historical grievances across generations, and the $\alpha_i$-expectile coupling $H^{i,g}$ \eqref{eq:true_dyn} linking individual trajectories to the macro‑environment and the joint distribution.
    
    \item \textbf{Four‑layer environmental dynamics.} A single true environment $E_t^g$ \eqref{eq:E_true_dyn} evolves under the influence of the entire population and a rich catalogue of stochastic drivers (Brownian, fractional, Gauss–Volterra, Rosenblatt, and Poisson). Each agent $i$ perceives this environment through her own measurement $E_{i,t}^{m}$ \eqref{eq:E_measure_dyn}, observation $E_{i,t}^{o}$ \eqref{eq:E_obs_dyn}, and perception $E_{i,t}^{p}$ \eqref{eq:E_perc_dyn}, creating a heterogeneous belief system.
    
    \item \textbf{Controlled coalition switching.} The discrete affiliation $C_{i,t}^g \in \{1,\dots,K\}$ of each agent follows a continuous‑time Markov chain whose transition intensities $\lambda_{m\to m'}^i$ \eqref{eq:lambda_def} depend on the agent's control $v_{i,t}^g$, her own state, the environment, and the $\beta_i$-expectile of the joint law. Coalitions can become endogenously extinct when they lose all members.
    
    \item \textbf{Intergenerational transmission.} At each boundary $T_{g+1}$, the terminal state of generation $g$ is mapped to the initial state of generation $g+1$ through the hereditary rules \eqref{eq:X_jump} (continuous state), \eqref{eq:C_jump} (coalition inheritance), \eqref{eq:E_jump} (environmental persistence), and \eqref{eq:regime_jump} (regime continuity). These rules induce a push‑forward operator $\mathcal{T}_{\boldsymbol{\zeta}}$ \eqref{eq:pushforward} on the space of probability measures, which is the central mechanism coupling successive generations.
    
    \item \textbf{War entrepreneurs' dual asset structure.} A subset $\cN_E$ of agents are war entrepreneurs whose local wealth is a sub‑vector of $X_{i,t}^g$ and whose global (laundered) wealth follows \eqref{eq:global_wealth}. The laundering efficiency $\eta_{\text{wash}}$ and the repatriation control $u_{i,t}^{g,\text{global}}$ encode the self‑reinforcing war economy loop that makes peace structurally unattractive for this class.
    
    \item \textbf{Expectile risk measures.} Each agent $i$ is characterised by two expectile levels $\alpha_i,\beta_i\in(0,1)$ \eqref{eq:expectile_def}, which govern her risk attitudes in the running and terminal costs, in the environment‑population coupling, and in the coalition switching intensities.
    
    \item \textbf{Payoff functionals.} Agent $i$ in generation $g$ minimises the expected total cost $J^{i,g}$ \eqref{eq:cost}, composed of a running cost $L^{i,g}$ \eqref{eq:running} and a terminal bequest cost $G^{i,g}$ \eqref{eq:terminal}. Both costs are functions of the true (unobserved) states, yet the agent must choose her controls based solely on her perceptions $X_{i,t}^p$ and $E_{i,t}^p$, together with the observed regime $s_t^g$.
\end{itemize}

The interaction of these components defines a coupled forward‑backward system on the infinite‑dimensional Wasserstein manifold of joint probability measures. Because agents base decisions on heterogeneous perceptions while costs are evaluated on true states, the model captures the essential opacity of the battlefield and the informational frictions that drive the tactical and strategic choices of the actors in the Sahel. The piecewise sequential equilibrium of this game is characterised, and its properties analysed, in the subsequent sections.

%%%%%%%%%%%%%%%%%%%%%%%%%%%%%%%%%%%%%%%%%%%%%

% =====================================================================
% STOCHASTIC STATE SPACE INFRASTRUCTURE AND COUPLING TOPOLOGY
% =====================================================================

\begin{figure}[htbp]
\centering

\resizebox{\textwidth}{!}{
\begin{tikzpicture}[
    font=\sffamily\small,
    >=Stealth,
    node distance=1.3cm and 1.1cm,
    % Component Block Styles
    microblock/.style={rectangle, draw=blue!60, fill=blue!5, rounded corners, minimum width=4.8cm, minimum height=1.1cm, align=center, thick},
    macroblock/.style={rectangle, draw=orange!70, fill=orange!5, rounded corners, minimum width=4.8cm, minimum height=1.1cm, align=center, thick},
    infblock/.style={rectangle, draw=purple!70, fill=purple!5, rounded corners, minimum width=5.5cm, minimum height=1.2cm, align=center, thick},
    container/.style={rectangle, draw=gray!30, fill=gray!2, dashed, rounded corners, inner sep=14pt},
    % Edge Styles
    stateflow/.style={->, thick, color=black!70, line width=1.1pt},
    feedbackflow/.style={->, thick, color=purple!80, line width=1.3pt},
    % Label Styles
    lbl/.style={align=center, font=\sffamily\scriptsize\bfseries, color=black!80}
]

% =====================================================================
% LEFT COLUMN: THE MICRO-LEVEL AGENT STATE HIERARCHY (4-LAYER FILTER)
% =====================================================================
\node[microblock] (x_true) {\textbf{True Individual State $X_{i,t}^g \in \mathbb{R}^d$}\\[2pt] 1. Radicalization \quad 2. Wealth \quad 3. Health};
\node[microblock, below=of x_true] (x_meas) {\textbf{Noisy Measurement $X_{i,t}^m \in \mathbb{R}^d$}\\[2pt] Rumors, Fragmented Intelligence};
\node[microblock, below=of x_meas] (x_obs) {\textbf{Filtered Observation $X_{i,t}^o \in \mathbb{R}^d$}\\[2pt] Tactical Filtering by Local Commanders};
\node[microblock, below=of x_obs] (x_perc) {\textbf{Subjective Perception $X_{i,t}^p \in \mathbb{R}^d$}\\[2pt] Cognitive Belief Set $\to$ Drives Control $u_{i,t}^g$};

% Micro Layer Background Container
\begin{scope}[on background layer]
    \node[container, fit={(x_true) (x_meas) (x_obs) (x_perc)}, 
          label={[anchor=north west, font=\sffamily\bfseries\color{blue!80!black}, xshift=5pt]north west:Agent Micro-State Filter}] (box_micro) {};
\end{scope}

% =====================================================================
% RIGHT COLUMN: THE MACRO-ENVIRONMENTAL STRATA
% =====================================================================
\node[macroblock, right=3.5cm of x_true] (e_true) {\textbf{True Environment $E_t^g \in \mathbb{R}^{d_E}$}\\[2pt] Gold Prices, Rainfall, Border Porosity};
\node[macroblock, below=of e_true] (e_perc) {\textbf{Perceived Environment $E_{i,t}^p \in \mathbb{R}^{d_E}$}\\[2pt] Subsidized Intelligence Gathering};
\node[macroblock, below=of e_perc] (s_regime) {\textbf{Discrete Regime State $s_t^g \in \mathcal{S}$}\\[2pt] Markov Chain: Peacetime $\leftrightarrow$ Full War};

% Macro Layer Background Container
\begin{scope}[on background layer]
    \node[container, fit={(e_true) (e_perc) (s_regime)}, 
          label={[anchor=north west, font=\sffamily\bfseries\color{orange!80!black}, xshift=5pt]north west:Macro-Environmental Stratum}] (box_macro) {};
\end{scope}

% =====================================================================
% BOTTOM ZONE: THE INFINITE-DIMENSIONAL SYSTEM MEASURE
% =====================================================================
\node[infblock, below=1.8cm of x_perc, xshift=4.1cm] (mu_measure) {\textbf{Joint Probability Distribution $\mu_t^g = \operatorname{Law}\left(s_t^g, \mathbf{X}_t^g, \mathbf{C}_t^g, E_t^g\right) \in \mathcal{P}_2(\mathcal{S}\times \mathcal{Z})$}\\[3pt] Fundamental Variable on the Wasserstein Manifold};

% =====================================================================
% STRUCTURAL CONNECTIONS AND INTERACTION VECTORS
% =====================================================================

% Micro-filtration connections
\draw[stateflow] (x_true) -- node[right, lbl] {Corrupted by\\Field Noise} (x_meas);
\draw[stateflow] (x_meas) -- node[right, lbl] {Non-Linear Filtered\\Information} (x_obs);
\draw[stateflow] (x_obs) -- node[right, lbl] {Bayesian Belief\\Update Loop} (x_perc);

% Macro connections
\draw[stateflow] (e_true) -- node[right, lbl] {Agent Signals\\\& Surveys} (e_perc);

% Cross-layer coupling
\draw[stateflow, dashed, color=blue!50!orange, line width=1pt] (e_true.west) -- node[above, lbl, pos=0.4] {Exogenous Drifts} (x_true.east);
\draw[stateflow, dashed, color=orange!50!blue, line width=1pt] (x_perc.east) -- node[above, lbl, pos=0.3, sloped] {Intelligence Cost Investment ($u_{i,t}$)} (e_perc.west);
\draw[stateflow] (s_regime.west) to[out=160, in=-20] node[below, lbl, pos=0.5, sloped] {Baseline Drifts ($b_i, b_E$)} (x_true.east);

% Push forward to Wasserstein Space
\draw[feedbackflow] (x_true.west) to[out=180, in=180] node[left, lbl, color=purple!90] {Aggregation of\\True Profiles} (mu_measure.west);
\draw[feedbackflow] (e_true.east) to[out=0, in=0] node[right, lbl, color=purple!90] {Systemic Law\\Integration} (mu_measure.east);

% Non-linear Expectile Feedback Loop (The Core of MFTG)
\draw[feedbackflow] (mu_measure.north) -- node[right, lbl, color=purple!90, yshift=0.3cm] {Asymmetric Tail Risk Assessment\\Non-Linear Expectiles $e_{\alpha_i}(\mu_t^g), e_{\beta_i}(\mu_t^g)$} (x_perc.south);

% Intergenerational Epoch Transition Overlay
% I%
\node[draw=purple!40, fill=purple!2, rectangle, rounded corners, minimum width=3.2cm, align=center, right=1.2cm of mu_measure, yshift=-1.3cm] (intergen) {
    \textbf{Hereditary Operator $\mathcal{T}_{\boldsymbol{\zeta}}$}\\[2pt] 
    Maps $\mu_{T_{g+1}-}^g \to \mu_{T_{g+1}}^{g+1}$
};
%\node[draw=purple!40, fill=purple!2, rectangle, rounded corners, minimum width=3.2cm, right=1.2cm of mu_measure, yshift=-1.3cm] (intergen) {\textbf{Hereditary Operator $\mathcal{T}_{\boldsymbol{\zeta}}$}\\[2pt] Maps $\mu_{T_{g+1}-}^g \to \mu_{T_{g+1}}^{g+1}$};

\draw[stateflow, color=purple!90, line width=1.2pt] (mu_measure.south) to[out=-25, in=180] node[below, lbl, color=purple!90, xshift=-0.4cm] {Generational Reset Boundary} (intergen.west);

\end{tikzpicture}
}
\caption{The Structural State-Space Hierarchy and Functional Topology of the Intergenerational Volterra MFTG Framework. The left stratum decomposes the micro-level agent state into a filter, mapping the divergence between the latent physical states and subjective operational perceptions. The right stratum captures the dynamic environmental co-evolution and qualitative regimes. These variables are aggregated into the Wasserstein manifold at the base, which feeds back asymmetric tail-risk expectiles directly into individual belief systems to govern cross-generational optimization.}
\label{fig:mftg_state_space}
\end{figure}
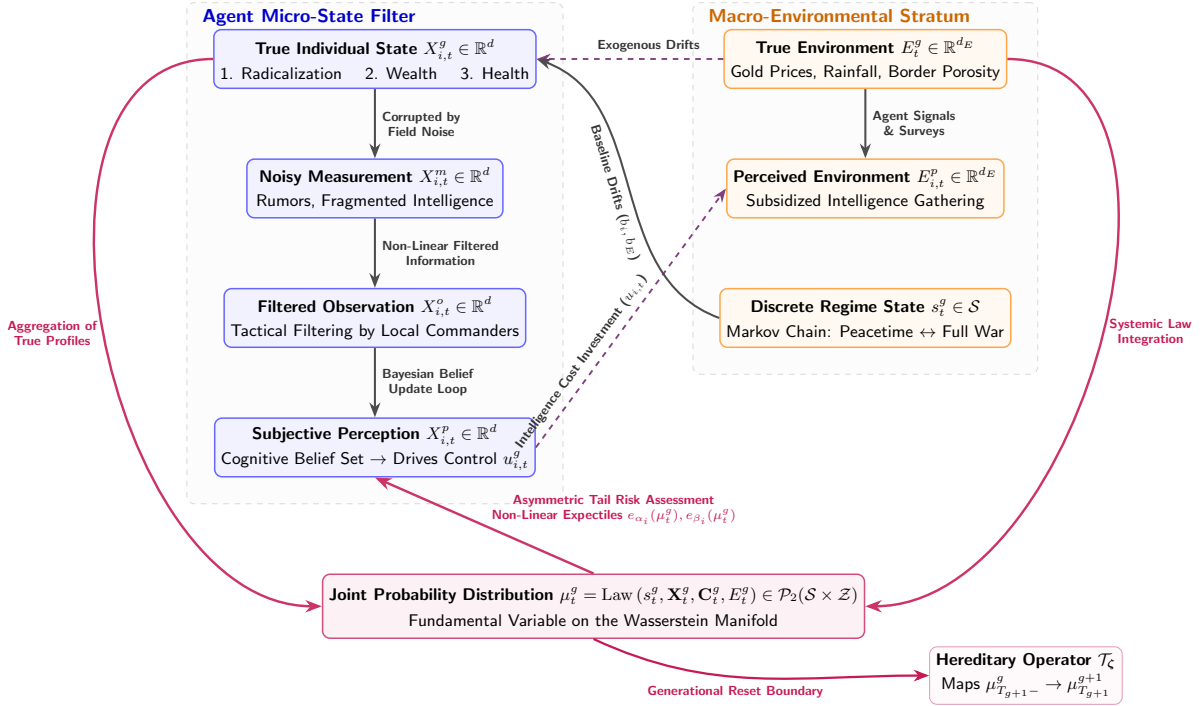

\subsection{Information Frictions and the Topology of MFTG State Strata}
\label{subsec:state_space_topology}

The  coupling between individual decision-making,  uncertainty, and aggregate macro-dynamics is formalized within our structural framework and visually contextualized in the state-space hierarchy of Figure~\ref{fig:mftg_state_space}. Traditional game-theoretic approaches typically enforce complete information transparency or basic additive noise over state trajectories, assuming that players maintain immediate access to the true states of the world. In the highly opaque, rumor-driven context of Sahelian conflict, this assumption introduces severe modeling misspecifications.  As mapped out in the left column of Figure~\ref{fig:mftg_state_space}, our model replaces these simplified assumptions with a  \textbf{Filter}. The unobserved true state vector $X_{i,t}^g \in \mathbb{R}^d$ represents the physical baseline of radicalization, liquid assets, and human capital. This latent variable passes through successive layers of measurement corruption and tactical filtering by local commanders, finally yielding the subjective perception vector $X_{i,t}^p$. This perception serves as the exclusive, restricted informational base upon which agents optimize their continuous controls $u_{i,t}^g$.
This micro-level  filter operates in constant conversation with the multi-dimensional \textbf{Macro-Environmental Stratum}, shown in the right column of Figure~\ref{fig:mftg_state_space}. The environment co-evolves with the population, linking global commodity variations (the true environment $E_t^g$) with the agents' localized, subsidized intelligence gatherers ($E_{i,t}^p$). These cross-layer interactions are further modulated by the discrete macro-political regime chain $s_t^g$, which acts as an overarching common driver altering the structural drift coefficients ($b_i, b_E$) across all equations. The core game-theoretic closure of the model occurs at the base of Figure~\ref{fig:mftg_state_space}. The continuous individual state profiles and discrete coalitional configurations are  integrated to define the true joint probability law $\mu_t^g$ on the  Wasserstein space $\mathcal{P}_2(\mathcal{Z})$. By routing interactions through this full distribution rather than simple population averages, our framework avoids the information loss that cripples classical mean-field approximations. 
As illustrated by the upward purple feedback vector, the global shape of this joint distribution yields the non-linear tail-risk parameters, the $\alpha_i$- and $\beta_i$-expectiles, which are fed back directly into the agent's subjective perception layer. Decisions are thus endogenously linked to tail sensitivity. At the epoch interface, the joint measure maps onto the generational boundary, where the non-local hereditary push-forward operator $\mathcal{T}_{\boldsymbol{\zeta}}$ transforms the terminal measure into the initial conditions of the next cohort. This closed-loop configuration provides the  foundation necessary to  evaluate the structural conditions of the generational war trap.

%%%%%%%%%%%%%%%%%%%%%%%%%%%%%%%%%%%%%%%%%%%%%%%%
\section{Equilibrium Characterisation }
\label{sec:equilibrium}

In this section, we provide a  characterisation of the Nash equilibrium for the intergenerational MFTG. Because the strategic interactions are mediated directly through the exact joint probability law $\mu_t^g$, the state space is inherently infinite-dimensional. Rather than pursuing a dynamic programming approach on the Wasserstein manifold, which is severely complicated by the non-Markovian Volterra revenge operator $\mathcal{K}_{\mathrm{rv}}^{i,g}(t)$ and the multi-layered information structures, we employ variational analysis. 

We derive a system of forward-backward stochastic Volterra differential equations with structural regime-switching and establish the exact trans-generational boundary conditions that link successive generations.

An essential feature of our model is the comprehensive noise structure detailed in Section~\ref{sec:prelim}, which includes fractional Brownian motions $B^{H_1}, B^{\frac{1+H_2}{2}}$, Gauss-Volterra processes $G$, and Rosenblatt processes $R^{H_2}$. Because these processes exhibit long memory and non-Gaussian characteristics, standard stochastic calculus does not apply to them directly. 

As specified in equations \eqref{eq:true_dyn} and \eqref{eq:E_true_dyn}, the diffusion coefficients $\sigma_i^{H_1}, \sigma_i^{H_2}, \sigma_i^{gv}, \sigma_i^{R}$ and their environmental counterparts are independent of the agents' controls $(u_{i,t}^g, v_{i,t}^g)$. As a consequence, these long-memory noises enter the system as exogenous stochastic perturbations. Under the variational approach, when we take the directional derivative of the state trajectory with respect to a control perturbation, the terms involving these non-standard noises vanish in the linearized dynamics. This decoupling allows us to formulate the adjoint equations using standard Itô and jump-diffusion adjoint structures relative to the Brownian motions, Poisson measures, and regime-switching chains, while treating the long-memory trajectories as path-dependent exogenous drifts.

Because the costs and switching intensities depend on the $\alpha_i$- and $\beta_i$-expectiles of the joint distribution $\mu_t^g$, we must compute the sensitivity of $e_\alpha(\mu)$ with respect to variations in the measure. We utilize the notion of the G\^ateaux derivative on the Wasserstein space $\mathcal{P}_2$.
From the implicit definition of the expectile in \eqref{eq:expectile_def}, let $F(m, \mu) = \alpha \int (z-m)_+ \, d\mu(z) - (1-\alpha) \int (m-z)_+ \, d\mu(z) = 0$. By the implicit function theorem, the derivative of the expectile map $\mu \mapsto e_\alpha(\mu)$ at a point $\mathbf{z}'$ in the extended state space is given by:
\begin{equation}
\partial_{\mu} e_\alpha(\mu)(\mathbf{z}') = \frac{\alpha (\mathbf{z}' - e_\alpha(\mu))_+ - (1-\alpha) (e_\alpha(\mu) - \mathbf{z}')_+}{\alpha \Pee_\mu(Y > e_\alpha(\mu)) + (1-\alpha) \Pee_\mu(Y \le e_\alpha(\mu))}.
\label{eq:lions_expectile}
\end{equation}
To simplify notation in the subsequent variational equations, we define the normalized asymmetric weight function $\omega_\alpha(\mathbf{z}', \mu)$ as:
\begin{equation}
\omega_\alpha(\mathbf{z}', \mu) \coloneqq \frac{\alpha \mathbbm{1}_{\{\mathbf{z}' > e_\alpha(\mu)\}} + (1-\alpha) \mathbbm{1}_{\{\mathbf{z}' \le e_\alpha(\mu)\}}}{\alpha \Pee_\mu(Y > e_\alpha(\mu)) + (1-\alpha) \Pee_\mu(Y \le e_\alpha(\mu))}.
\label{eq:omega_weight}
\end{equation}
Thus, $\partial_{\mu} e_\alpha(\mu)(\mathbf{z}') = \omega_\alpha(\mathbf{z}', \mu) (\mathbf{z}' - e_\alpha(\mu))$.

\subsection{Systems driven by Rosenblatt noise  }

Before stating the core change-of-variable formulas, we formalize the foundational path operators and analytical regularity criteria governing the non-Markovian memory loops, Rosenblatt processes, and fractional fields.

\begin{definition}[Liouville Fractional Integral and Derivation Operators]
Let $\alpha \in (0,1)$. The left-sided Riemann--Liouville fractional integral operator $I^{\alpha}_{+}$ acting on a function $f: \R \to \R$ is defined pathwise by:
\begin{equation}
(I^{\alpha}_{+}f)(x) := \frac{1}{\Gamma(\alpha)}\int_{-\infty}^{x} f(u) (x-u)^{\alpha-1}\,du,
\end{equation}
where $\Gamma(z) = \int_0^\infty t^{z-1} e^{-t}\,dt$ is the standard Euler Gamma function. Correspondingly, for the mixed-space two-dimensional tensor mapping $f: \R^2 \to \R$, the product operator $(I^{\alpha_1,\alpha_2}_{+,+}f)(x_1,x_2)$ is defined as:
\begin{equation}
(I^{\alpha_1,\alpha_2}_{+,+}f)(x_1,x_2) := \frac{1}{\Gamma(\alpha_1)\Gamma(\alpha_2)}\int_{-\infty}^{x_1}\int_{-\infty}^{x_2} f(u,v) (x_1-u)^{\alpha_1-1}(x_2-v)^{\alpha_2-1}\, du\,dv.
\end{equation}
The specialized analytical non-local derivative operators $\nabla^{\alpha}$ and $\nabla^{\alpha,\alpha}$ utilized throughout this paper are defined as the composition of the fractional integral blocks with the first and second standard distributional derivative operators $D$ and $D^2$, such that:
\begin{equation}
\nabla^{\alpha} := I^{\alpha}_{+}\circ D, \qquad \nabla^{\alpha,\alpha} := I^{\alpha,\alpha}_{+,+}\circ D^2.
\end{equation}
\end{definition}

To isolate the infinitesimal generators and protect the continuous-time jump-diffusion paths from explosion, we impose the following mathematical constraints on the system's primitives.

\begin{assumption}[Filtration and Domain Regularity]
\label{ass:filtration}
The probability space $(\Omega, \Fcal, \mathbb{F}=(\Fcal_t)_{t \ge 0}, \Pee)$ is complete, and the filtration $\mathbb{F}$ satisfies the standard conditions (i.e., it is right-continuous and $\Fcal_0$ contains all $\Pee$-null sets). The target macro-transformation function $f(t, s, \Xbf, \Cbf, E)$ satisfies:
\begin{enumerate}
    \item $f(\cdot, s, \cdot, \Cbf, \cdot) \in \mathcal{C}^{1,2,3}([0,T] \times \R^{N \times d} \times \R^{d_E})$ for each fixed discrete macro-regime $s \in \mathcal{S}$ and discrete coalitional alignment configuration $\Cbf \in \{1,\dots,K\}^N$.
    \item The third-order continuous spatial derivative satisfies the following non-linear growth bound:
    \begin{equation}
    |\nabla^3_{\Xbf} f(t, s, \Xbf, \Cbf, E)| \le \mathcal{C}_T \left(1 + \|\Xbf\|^{\kappa} + \|E\|^{\kappa}\right), \quad \kappa \ge 1.
    \end{equation}
\end{enumerate}
\end{assumption}

\begin{assumption}[Lipschitz and Sub-Linear Growth of Functional Fields]
\label{ass:lipschitz}
Let $\mathcal{Z} = \R^{N \times d} \times \{1,\dots,K\}^N \times \mathbb{R}^{d_E}$. All continuous drift vectors $b_i, b_E$, diffusion matrices $\sigma_i, \sigma_{0,i}, \sigma_E^w$, and fractional/rough path volatility modulators $\sigma_i^{H_1}, \sigma_i^{H_2}, \sigma_i^R, \sigma_E^{fbm,H_1}, \sigma_E^{fbm,H_2}, \sigma_E^R$ satisfy a global Lipschitz continuity condition and sub-linear growth bounds uniformly in $t \in [0,T]$ and across all regimes $s \in \mathcal{S}$. For any generic coefficient $\Phi(t, s, \Zbf, \ubf, \mu)$, we require:
\begin{equation}
\|\Phi(t, s, \Zbf, \ubf, \mu) - \Phi(t, s, \Zbf', \ubf', \mu')\| \le L_{\Phi} \left( \|\Zbf - \Zbf'\| + \|\ubf - \ubf'\| + \mathbb{W}_2(\mu, \mu') \right),
\end{equation}
where $\mathbb{W}_2$ is the 2-Wasserstein metric defined on the infinite-dimensional probability manifold $\mathcal{P}_2(\mathcal{Z})$. This condition guarantees the strong uniqueness and pathwise existence of solutions for the driving McKean--Vlasov--Volterra systems.
\end{assumption}

\begin{assumption}[Integrability of Non-Local Memory Kernels]
\label{ass:integrability}
The hereditary Volterra revenge operator $\mathcal{K}_{\mathrm{rv}}^{i,g}(t) = \int_0^t (t-\tau)^{-\gamma} \psi(\Zbf_{\tau-}^g)\,d\tau$ uses a singular memory kernel parameter $\gamma \in (0, 1-H_2)$. The linear operator mappings $d_1, d_2, d_3$ introduced within the fractional sub-systems are required to have continuous sample paths and belong to the specified functional spaces:
\begin{equation}
d_1 \in L^1([0,T]), \qquad d_2 \in L^{\frac{2}{1+H_2}}([0,T]), \qquad d_3 \in L^{\frac{1}{H_2}}([0,T]).
\end{equation}
\end{assumption}

\begin{assumption}[L{\'e}vy Measure Constraints for Large and Discontinuous Jumps]
\label{ass:levy}
The independent Poisson random measures $\tilde{N}_i, \tilde{N}_0, \tilde{N}_E$ and Gamma-distribution point paths possess valid L{\'e}vy characteristics such that their compensator densities satisfy:
\begin{equation}
\int_{\mathcal{Z}} \left( \|z\|^2 \wedge \|z\| \right) \nu_i(dz) < \infty, \quad \int_{\mathcal{Z}_0} \left( \|z_0\|^2 \wedge \|z_0\| \right) \nu_0(dz_0) < \infty, \quad \int_0^\infty \left( z^2 \wedge z \right) \nu_G(dz) < \infty.
\end{equation}
This formulation preserves the standard square-integrable martingale infrastructure for all compensated jump variables.
\end{assumption}

% =====================================================================
% LEMMA 1: ROSENBLATT CALCULUS PRESERVATION
% =====================================================================

\begin{lemma}[It\^o Formula for a Rosenblatt-Fractional Noise Pair]
\label{lem:rosenblatt_ito}
Let $H \in (1/2, 1)$ and define the scalar constants:
\begin{equation}
c = C_R^H \Gamma^2\!\left(\frac{H}{2}\right),\qquad \tilde c = \sqrt{\frac{(2H-1)\Gamma(1-\frac{H}{2})\Gamma(\frac{H}{2})}{(H+1)\Gamma(1-H)}}.
\end{equation}
Suppose the continuous real-valued scalar state variable $x(t)$ satisfies the stochastic evolution equation:
\begin{equation}
x(t) = x_0 + \int_0^t d_1(s)\,ds + 2\tilde c\int_0^t d_2(s)\,dB^{\frac{H+1}{2}}(s) + \int_0^t d_3(s)\,dR^H(s),
\end{equation}
where $B^{(H+1)/2}$ is a standard fractional Brownian motion with Hurst parameter $H_{\mathrm{fbm}} = (H+1)/2$ (possessing zero quadratic variation paths), and $R^H$ is the coupled non-Gaussian Rosenblatt process. Then, for any macro-transformation operator $f \in C^{1,3}([0,T] \times \R)$, the transformed process $y(t) = f(t, x(t))$ admits the following representation:
\begin{equation}
y(t) = f(t,x(t)) = y_0 + \int_0^t \tilde d_1(s)\,ds + 2\tilde c\int_0^t \tilde d_2(s)\,dB^{\frac{H+1}{2}}(s) + \int_0^t \tilde d_3(s)\,dR^H(s),
\end{equation}
where the transformed functional line loadings are given by:
\begin{align*}
\tilde d_1 &= f_t + f_x d_1 + 2c\, f_{xx}\, (\nabla^{H/2}x)\, d_2 + c\, f_{xx}\, (\nabla^{H/2,H/2}x)\, d_3 + c\, f_{xxx}\, [\nabla^{H/2}x]^2 d_3, \\
\tilde d_2 &= f_x d_2 + f_{xx}\, (\nabla^{H/2}x)\, d_3, \\
\tilde d_3 &= f_x d_3.
\end{align*}
Here, the non-local tracking derivative operators are defined via pathwise integration as $\nabla^{H/2} = I_+^{H/2} \circ D$ and $\nabla^{H/2,H/2} = I_+^{H/2,H/2} \circ D^2$.
\end{lemma}

% =====================================================================
% LEMMA 2: REGULAR FBM STEPPING
% =====================================================================
\begin{lemma}[It\^o Formula for a Regular Fractional Brownian Motion Field]
\label{lem:fbm_ito}
Let $H_1 \in (1/2, 1)$ and let the continuous state variable $x(t)$ be modeled as a stochastic integral driven by a regular fractional Brownian motion $B^{H_1}$:
\begin{equation}
x(t) = x_0 + \int_0^t d_1(s)\,ds + \int_0^t d_2(s)\,dB^{H_1}(s).
\end{equation}
Then, for any macro-scaling function $f \in \mathcal{C}^{1,2}([0,T] \times \R)$, the pathwise transformation is given by:
\begin{equation}
f(t,x(t)) = f(0,x_0) + \int_0^t \bigl( f_t + f_x d_1 + H_1 s^{2H_1-1} f_{xx} d_2^2 \bigr)\,ds + \int_0^t f_x d_2\,dB^{H_1}(s).
\end{equation}
For a general, state-dependent non-linear volatility coefficient $d_2(t, x(t))$, the deterministic fractional drift loading maps onto the following memory-convolution structure:
\begin{equation}
\mathcal{L}_{\mathrm{fbm\_drift}} = H_1 s^{2H_1-1} f_{xx} \bigl( d_2 \cdot (\nabla^{H_1/2}d_2) \bigr).
\end{equation}
\end{lemma}

% =====================================================================
% THEOREM: MAIN UNIFIED MAXIMAL RIGOR CHANGE OF VARIABLES
% =====================================================================

\begin{theorem}[Change-of-Variable Formula with Intergenerational Memory Noises]
\label{thm:complete_rigor_all}

\noindent\textbf{System Configuration and Driving Noise Structures.}\\
Let $(\Omega,\Fcal,\mathbb{F}=(\Fcal_t)_{t \ge 0},\Pee)$ be a complete filtered probability space satisfying the usual conditions (Assumption~\ref{ass:filtration}). We assume the existence of a mutually independent family of non-Markovian stochastic processes, distributed across multiple interaction channels as follows:
\begin{itemize}
  \item \textit{Standard Brownian Vectors:} Independent Brownian motions $W_{i,t}\in\R^d$ (individual agent state deviations), $W_t^0\in\R^d$ (systemic common noise), and $W_{E,t}\in\R^{d_E}$ (exogenous macro-environmental vectors).
  
  \item \textit{Regular Fractional Brownian Noise ($H_1 \in (1/2, 1)$):} The processes $B_i^{H_1}(t)$, $B^{H_1,0}(t)$, and $B_E^{H_1}(t)$ track long-range dependence. They have zero quadratic variation paths and inject a deterministic time-decay drift scaled by $H_1 \tau^{2H_1-1}$ into the continuous path equations.
  
  \item \textit{Fractional Gauss--Volterra Driver Components ($H_2 \in (1/2, 1)$):} The processes $B_i^{\frac{1+H_2}{2}}(t)$, $B^{\frac{1+H_2}{2},0}(t)$, and $B_E^{\frac{1+H_2}{2}}(t)$ possess a fractional parameter $H_{\mathrm{fbm}} = \frac{1+H_2}{2}$. They are pathwise cross-coupled to the corresponding Rosenblatt processes.
  
  \item \textit{Rosenblatt Process Array:} Non-Gaussian processes $R_i^{H_2}(t)$, $R^{H_2,0}(t)$, and $R_E^{H_2}(t)$ operate within the second Wiener chaos layer under a Hurst parameters framework governed by $H_2$.
  
  \item \textit{General Gauss--Volterra Multi-Memory Inversions:} The non-Markovian processes $G_i(t)$, $G_0(t)$, and $G_E(t)$ are defined via a continuous deterministic kernel $K(t,s)$. They exhibit a non-zero quadratic variation density profile defined as:
  \begin{equation}
  q(t) = \left(K(t,t) + \int_0^t \partial_t K(t,s)\,ds\right)^2.
  \end{equation}
  
  \item \textit{Endogenous Gamma Jumps:} Processes $\Gamma_i(t), \Gamma_0(t), \Gamma_E(t)$ possess infinitely divisible sample paths driven by the singular L{\'e}vy density $\nu_G(dz) = z^{-1}e^{-z}dz$.
  
  \item \textit{Compensated Poisson Measures:} Independent random measures $\tilde N_i(dt,dz)$, $\tilde N_0(dt,dz_0)$, and $\tilde N_E(dt,dz)$ operate on their respective Polish selection spaces $\mathcal{Z}, \mathcal{Z}_0, \mathcal{Z}_E$, balanced by their predictable compensator areas $\nu_i(dz)dt, \nu_0(dz_0)dt, \nu_E(dz)dt$.
\end{itemize}

\noindent\textbf{Coupled System Jump-Diffusion Dynamics.}\\
The structural vector coordinates composed of agent state profiles $X_{i,t}^g \in \R^d$, discrete coalitional alignments $C_{i,t}^g \in \{1,\dots,K\}$, macro-environmental conditions $E_t^g \in \R^{d_E}$, and discrete political regimes $s_t^g \in \{1,\dots,M\}$ are gathered into the unified state vector $\Zbf_t^g = (\Xbf_t^g, \Cbf_t^g, E_t^g)$. Their stochastic differential systems are given by:
\begin{align*}
dX_{i,t}^g &=
\Bigl[b_i + H^{i,g} + \mathcal{K}_{\mathrm{rv}}^{i,g}\Bigr]dt
+ \sigma_i dW_{i,t} + \int_{\mathcal{Z}}\sigma_i^n \,\tilde N_i(dt,dz)
+ \sigma_i^{H_1} dB_i^{H_1}(t) \\
&\quad + 2\tilde c\,\sigma_i^{H_2} dB_i^{\frac{1+H_2}{2}}(t) + \sigma_i^{R} dR_i^{H_2}(t) + \sigma_i^{gv} dG_i(t) \\
&\quad + \sigma_{0,i} dW_t^0 + \int_{\mathcal{Z}_0}\sigma_{0,i}^n \,\tilde N_0(dt,dz_0)
+ \sigma_{0,i}^{H_1} dB^{H_1,0}(t) \\
&\quad + 2\tilde c\,\sigma_{0,i}^{H_2} dB^{\frac{1+H_2}{2},0}(t) + \sigma_{0,i}^{R} dR^{H_2,0}(t)
+ \sigma_{0,i}^{gv} dG_0(t), \\
dE_t^g &=
b_E dt + \sigma_E^w dW_{E,t} + \int_{\mathcal{Z}_E}\gamma_E^n \,\tilde N_E(dt,dz)
+ \sigma_E^{fbm,H_1} dB_E^{H_1}(t) \\
&\quad + 2\tilde c\,\sigma_E^{fbm,H_2} dB_E^{\frac{1+H_2}{2}}(t) + \sigma_E^{R} dR_E^{H_2}(t)
+ \sigma_E^{gv} dG_E(t),
\end{align*}
where all functional coefficients are evaluated on the left-limit configurations $(t, s_{t-}^g, \Xbf_{t-}^g, \Cbf_{t-}^g, E_{t-}^g, \ubf_t, \mu_{t-}^g)$.

\noindent\textbf{ Assertion.}\\
Let $Y_t = f(t, s_t^g, \Xbf_t^g, \Cbf_t^g, E_t^g)$ be the macro-transformed target process. If Assumptions~\ref{ass:filtration}--\ref{ass:levy} hold, then $\Pee$-almost surely for all $t \in [0,T]$, the total change-of-variable decomposition is given by:
\begin{equation}
\begin{aligned}
Y_t &= Y_0 + \int_0^t \Bigl( \partial_\tau + \mathcal{L}_{\mathrm{drift}} + \mathcal{L}_{\mathrm{diff}}   \Bigr) f(\tau,s_{\tau-}^g,\Xbf_{\tau-}^g,\Cbf_{\tau-}^g,E_{\tau-}^g)\,d\tau \\[4pt]
& \quad + \int_0^t \Bigl( \mathcal{L}_{\mathrm{regFBM}} + \mathcal{L}_{\mathrm{Rosenblatt}} + \mathcal{L}_{\mathrm{GaussVolterra}}  \Bigr) f(\tau,s_{\tau-}^g,\Xbf_{\tau-}^g,\Cbf_{\tau-}^g,E_{\tau-}^g)\,d\tau \\[4pt]
&\quad + \int_0^t \Bigl(  \mathcal{L}_{\mathrm{Regime}} + \mathcal{L}_{\mathrm{Coalition}} + \mathcal{L}_{\mathrm{PoissonJump}} \Bigr) f(\tau,s_{\tau-}^g,\Xbf_{\tau-}^g,\Cbf_{\tau-}^g,E_{\tau-}^g)\,d\tau \\[4pt]
&\quad + \sum_{i=1}^N \int_0^t \nabla_{\Xbf_i} f \cdot \sigma_i^{H_1}(\tau)\,dB_i^{H_1}(\tau) + \int_0^t \nabla_E f \cdot \sigma_E^{fbm,H_1}(\tau)\,dB_E^{H_1}(\tau) \\[2pt]
&\quad + \sum_{i=1}^N \int_0^t \nabla_{\Xbf_i} f \cdot 2\tilde c\,\sigma_i^{H_2}(\tau)\,dB_i^{\frac{1+H_2}{2}}(\tau) + \int_0^t \nabla_E f \cdot 2\tilde c\,\sigma_E^{fbm,H_2}(\tau)\,dB_E^{\frac{1+H_2}{2}}(\tau) \\[2pt]
&\quad + \sum_{i=1}^N \int_0^t \nabla_{\Xbf_i} f \cdot \sigma_i^{R}(\tau)\,dR_i^{H_2}(\tau) + \int_0^t \nabla_E f \cdot \sigma_E^{R}(\tau)\,dR_E^{H_2}(\tau) \\[2pt]
&\quad + \sum_{i=1}^N \int_0^t \nabla_{\Xbf_i} f \cdot \sigma_i^{gv}(\tau)\,dG_i(\tau) + \int_0^t \nabla_E f \cdot \sigma_E^{gv}(\tau)\,dG_E(\tau) + \int_0^t \Bigl(\sum_{i=1}^N \nabla_{\Xbf_i} f \cdot \sigma_{0,i}^{gv}(\tau)\Bigr) dG_0(\tau) \\[2pt]
&\quad + \mathcal{M}_t^{\mathrm{total}},
\end{aligned}
\label{eq:master_ito_expansion}
\end{equation}
where the constitutive operational blocks are evaluated at the left-limit configurations and defined explicitly below.
\end{theorem}

% =====================================================================
% DETAILED SUBSECTIONS FOR EVERY SINGLE OPERATOR BLOCK
% =====================================================================

\subsection*{First-Order Target Drift Operator $\mathcal{L}_{\mathrm{drift}}$}
This block maps the first-order partial spatial derivatives against the combined continuous drifts, non-linear risk tail expectiles, and Volterra revenge parameters:
\begin{equation}
\mathcal{L}_{\mathrm{drift}} f = \sum_{i=1}^N\sum_{k=1}^d \partial_{x_{i,k}}f\,\left(b_{i,k}+H_{k}^{i,g}+\mathcal{K}_{\mathrm{rv},k}^{i,g}\right) + \sum_{m=1}^{d_E}\partial_{E_m}f\,b_{E,m}.
\end{equation}

\subsection*{Second-Order Multi-Agent Brownian Operator $\mathcal{L}_{\mathrm{diff}}$}
This operator captures the standard quadratic variations and cross-agent network correlations driven by the independent and common standard Brownian vectors:
\begin{equation}
\begin{aligned}
\mathcal{L}_{\mathrm{diff}} f &=
\frac12\sum_{i=1}^N\sum_{k,l=1}^d \partial^2_{x_{i,k}x_{i,l}}f\,[\sigma_i\sigma_i^\top]_{k,l}
+ \frac12\sum_{i,j=1}^N\sum_{k,l=1}^d \partial^2_{x_{i,k}x_{j,l}}f\,[\sigma_{0,i}\sigma_{0,j}^\top]_{k,l} \\
&\quad + \frac12\sum_{m,n=1}^{d_E}\partial^2_{E_mE_n}f\,[\sigma_E^w(\sigma_E^w)^\top]_{m,n}
+ \sum_{i=1}^N\sum_{k=1}^d\sum_{m=1}^{d_E}\partial^2_{x_{i,k}E_m}f\,[\sigma_{0,i}(\sigma_E^w)^\top]_{k,m}.
\end{aligned}
\end{equation}

\subsection*{Macro-Regime Switching Generator $\mathcal{L}_{\mathrm{Regime}}$}
This operator handles the discrete, qualitative jumps of the background state $s_t^g$ over the finite regime set $\mathcal{S}$:
\begin{equation}
\mathcal{L}_{\mathrm{Regime}} f = \sum_{s'\neq s_{\tau-}^g} \gamma_{s_{\tau-}^g\to s'}\bigl[ f(\tau,s',\Xbf_{\tau-}^g,\Cbf_{\tau-}^g,E_{\tau-}^g) - f(\tau,s_{\tau-}^g,\Xbf_{\tau-}^g,\Cbf_{\tau-}^g,E_{\tau-}^g) \bigr].
\end{equation}

\subsection*{Controlled Coalition Jump Generator $\mathcal{L}_{\mathrm{Coalition}}$}
This operator tracks individual defection profiles. It models how agents update their coalitional alignment vectors $\Cbf_t^g$ based on the expectile switching intensities defined in equation \eqref{eq:switch_prob}:
\begin{equation}
\lambda_{m\to m'}^i(\tau) = \bar{\lambda}_{m\to m'}^{\,i}(\tau,s_{\tau-}^g)\exp\!\Bigl((v_{i,\tau-}^g)^\top D_{mm'}^i X_{i,\tau-}^g + \Psi_{mm'}(E_{\tau-}^g,e_{\beta_i}(\mu_{\tau-}^g),s_{\tau-}^g)\Bigr),
\end{equation}
\begin{equation}
\mathcal{L}_{\mathrm{Coalition}} f = \sum_{i=1}^N\sum_{m'\neq C_{i,\tau-}^g} \lambda_{C_{i,\tau-}^g\to m'}^i(\tau)\;\Bigl[ f(\tau,s_{\tau-}^g,\Xbf_{\tau-}^g,\Cbf_{\sim i,\tau-}^g[m'],E_{\tau-}^g) - f(\tau, s_{\tau-}^g, \Xbf_{\tau-}^g, \Cbf_{\tau-}^g, E_{\tau-}^g) \Bigr].
\end{equation}

\subsection*{Poisson and Gamma Jump Compensator Operator $\mathcal{L}_{\mathrm{PoissonJump}}$}
This block isolates the predictable drift adjustments required to balance the discontinuous jump measures, ensuring the remaining stochastic elements operate as pure martingales:
\begin{equation}
\begin{aligned}
\mathcal{L}_{\mathrm{PoissonJump}} f &=
\sum_{i=1}^N \int_{\mathcal{Z}} \Bigl[ f(\tau,\Xbf_{\tau-}^g+\ebf_i\otimes\sigma_i^n(\tau,z),\Cbf_{\tau-}^g,E_{\tau-}^g) - f - \nabla_{\Xbf_i} f \cdot \sigma_i^n(\tau,z) \Bigr] \nu_i(dz) \\
&\quad + \int_{\mathcal{Z}_0} \Bigl[ f(\tau,\Xbf_{\tau-}^g+\sum_i\ebf_i\otimes\sigma_{0,i}^n(\tau,z_0),\Cbf_{\tau-}^g,E_{\tau-}^g) - f - \sum_{i=1}^N \nabla_{\Xbf_i} f \cdot \sigma_{0,i}^n(\tau,z_0) \Bigr] \nu_0(dz_0) \\
&\quad + \int_{\mathcal{Z}_E} \Bigl[ f(\tau,\Xbf_{\tau-}^g,\Cbf_{\tau-}^g,E_{\tau-}^g+\gamma_E^n(\tau,z)) - f - \nabla_E f \cdot \gamma_E^n(\tau,z) \Bigr] \nu_E(dz) \\
&\quad + \sum_{i=1}^N \int_0^\infty \Bigl[ f(\tau,\Xbf_{\tau-}^g+\ebf_i\otimes\sigma_i^{gv}(\tau)z,\Cbf_{\tau-}^g,E_{\tau-}^g) - f - \nabla_{\Xbf_i} f \cdot (\sigma_i^{gv}(\tau)z) \Bigr] \nu_G(dz) \\
&\quad + \int_0^\infty \Bigl[ f(\tau,\Xbf_{\tau-}^g+\sum_i\ebf_i\otimes\sigma_{0,i}^{gv}(\tau)z,\Cbf_{\tau-}^g,E_{\tau-}^g) - f - \sum_{i=1}^N \nabla_{\Xbf_i} f \cdot (\sigma_{0,i}^{gv}(\tau)z) \Bigr] \nu_G(dz) \\
&\quad + \int_0^\infty \Bigl[ f(\tau,\Xbf_{\tau-}^g,\Cbf_{\tau-}^g,E_{\tau-}^g+\sigma_E^{gv}(\tau)z) - f - \nabla_E f \cdot (\sigma_E^{gv}(\tau)z) \Bigr] \nu_G(dz).
\end{aligned}
\end{equation}

\subsection*{ Fractional Brownian Motion Operator $\mathcal{L}_{\mathrm{regFBM}}$}
This operator handles the deterministic time-stepping adjustments required by the long-range dependence properties of the standard fractional fields. By lifting the coordinate-wise representation from Lemma~\ref{lem:fbm_ito} onto the full matrix-covariant manifold, the adjustments are expressed compactly via the trace operator:
\begin{equation}
{
\mathcal{L}_{\mathrm{regFBM}} f = H_1 \tau^{2H_1-1} \left( \sum_{i=1}^N \mathrm{Tr}\left[ \nabla^2_{\mathbf{X}_i} f \cdot \left(\sigma_{i}^{H_1}\right)\left(\sigma_{i}^{H_1}\right)^\top \right] + \mathrm{Tr}\left[ \nabla^2_{E} f \cdot \left(\sigma_{E}^{fbm,H_1}\right)\left(\sigma_{E}^{fbm,H_1}\right)^\top \right] \right).
}
\label{eq:matrix_fbm_operator}
\end{equation}
where $\nabla^2_{\mathbf{X}_i} f \in \mathbb{R}^{d \times d}$ is the spatial Hessian matrix of the macro-welfare function with respect to agent $i$'s vector state, $\nabla^2_{E} f \in \mathbb{R}^{d_E \times d_E}$ is the Hessian matrix with respect to the environmental variables, and $(\cdot)^\top$ denotes the matrix transpose operator.

\subsection*{Gauss-Volterra Memory Operator $\mathcal{L}_{\mathrm{GaussVolterra}}$}
Because the Gauss--Volterra processes are non-Markovian, they do not possess localized semi-martingale quadratic variations. Computing the variance of the path configurations reveals that the historical volatility trajectories are tightly convolved with the kernel's time derivative. We define the continuous \emph{Memory-Filtered Volatility Matrix Fields} $\mathbf{\Omega}_i(t) \in \mathbb{R}^{d \times d}$, $\mathbf{\Omega}_0(t) \in \mathbb{R}^{d \times d}$, and $\mathbf{\Omega}_E(t) \in \mathbb{R}^{d_E \times d_E}$ explicitly as:
\begin{align}
\mathbf{\Omega}_i(t) &:= \sigma_i^{gv}\bigl(t, \mathbf{Z}_{t-}^g\bigr) K(t,t) + \int_0^t \sigma_i^{gv}\bigl(\tau, \mathbf{Z}_{\tau-}^g\bigr) \partial_t K(t,\tau) \, d\tau, \label{eq:omega_i_corrected} \\
\mathbf{\Omega}_0(t) &:= \sigma_{0,i}^{gv}\bigl(t, \mathbf{Z}_{t-}^g\bigr) K(t,t) + \int_0^t \sigma_{0,i}^{gv}\bigl(\tau, \mathbf{Z}_{\tau-}^g\bigr) \partial_t K(t,\tau) \, d\tau, \label{eq:omega_0_corrected} \\
\mathbf{\Omega}_E(t) &:= \sigma_E^{gv}\bigl(t, \mathbf{Z}_{t-}^g\bigr) K(t,t) + \int_0^t \sigma_E^{gv}\bigl(\tau, \mathbf{Z}_{\tau-}^g\bigr) \partial_t K(t,\tau) \, d\tau. \label{eq:omega_E_corrected}
\end{align}
By mapping these non-local variance densities onto the multivariate Hessian manifold, the mathematically accurate representation of $\mathcal{L}_{\mathrm{GaussVolterra}}$ is given by:
\begin{equation}
\begin{array}{l}
\mathcal{L}_{\mathrm{GaussVolterra}} f =
\frac{1}{2} \sum_{i=1}^N \mathrm{Tr} \left[ \nabla^2_{\mathbf{X}_i} f \cdot \left( \mathbf{\Omega}_i(t) \mathbf{\Omega}_i^\top(t) \right) \right] 
+ \frac{1}{2} \sum_{i,j=1}^N \mathrm{Tr} \left[ \nabla^2_{\mathbf{X}_i \mathbf{X}_j} f \cdot \left( \mathbf{\Omega}_0(t) \mathbf{\Omega}_0^\top(t) \right) \right] \\[2pt]
\quad + \frac{1}{2} \mathrm{Tr} \left[ \nabla^2_{E} f \cdot \left( \mathbf{\Omega}_E(t) \mathbf{\Omega}_E^\top(t) \right) \right]
+ \sum_{i=1}^N \mathrm{Tr} \left[ \nabla^2_{\mathbf{X}_i E} f \cdot \left( \mathbf{\Omega}_0(t) \mathbf{\Omega}_E^\top(t) \right) \right].
\end{array}
\label{eq:true_gauss_volterra_operator}
\end{equation}
where $\nabla^2_{\mathbf{X}_i \mathbf{X}_j} f$ captures the cross-agent second-order sensitivities induced by the shared systemic 
Gauss-Volterra background noise channel $G_0(t)$.

\subsection*{Rosenblatt--Fractional Non-Gaussian Interaction Operator $\mathcal{L}_{\mathrm{Rosenblatt}}$}
This operator applies the non-local analytical modifications from Lemma~\ref{lem:rosenblatt_ito} across the interacting agent spaces and the common environment layout. We maintain the consolidated notation for the total fractional and Rosenblatt volatility structures:
\begin{align*}
d_{2,i,k} &= \sigma_{i,k}^{H_2} + \sigma_{0,i,k}^{H_2}, \qquad d_{3,i,k} = \sigma_{i,k}^{R} + \sigma_{0,i,k}^{R}, \\
d_{2,E,m} &= \sigma_{E,m}^{fbm,H_2}, \qquad\qquad d_{3,E,m} = \sigma_{E,m}^{R},
\end{align*}
accompanied by the non-local tracking derivative states evaluated at time $\tau$:
\begin{equation}
\begin{array}{l}
\Gamma_{i,k} = (\nabla^{H_2/2}X_{i,k}^g)(\tau), \qquad \Lambda_{i,k} = (\nabla^{H_2/2,H_2/2}X_{i,k}^g)(\tau,\tau), \\
\Gamma_{E,m} = (\nabla^{H_2/2}E_{m}^g)(\tau), \qquad \Lambda_{E,m} = (\nabla^{H_2/2,H_2/2}E_{m}^g)(\tau,\tau).
\end{array}
\end{equation}
By extracting the deterministic line-loading dynamics from the $\tilde{d}_1$ drift expansion in Lemma~\ref{lem:rosenblatt_ito}, the full multivariate operator $\mathcal{L}_{\mathrm{Rosenblatt}}$ is rigorously defined as:
\begin{equation}
\begin{aligned}
\mathcal{L}_{\mathrm{Rosenblatt}} f &=
\sum_{i=1}^N\sum_{k=1}^d \partial_{x_{i,k}}f \cdot \Bigl[ 2c\,\Gamma_{i,k} d_{2,i,k} + c\,\Lambda_{i,k} d_{3,i,k} \Bigr]
+ \sum_{m=1}^{d_E}\partial_{E_m}f \cdot \Bigl[ 2c\,\Gamma_{E,m} d_{2,E,m} + c\,\Lambda_{E,m} d_{3,E,m} \Bigr] \\[2pt]
&\quad + \sum_{i=1}^N\sum_{k,l=1}^d \partial^2_{x_{i,k}x_{i,l}}f \cdot c\,\Gamma_{i,k}\Gamma_{i,l}\,d_{3,i,l}
+ \sum_{m,n=1}^{d_E}\partial^2_{E_mE_n}f \cdot c\,\Gamma_{E,m}\Gamma_{E,n}\,d_{3,E,n} \\[2pt]
&\quad + \sum_{i=1}^N\sum_{k=1}^d\sum_{m=1}^{d_E}\partial^2_{x_{i,k}E_m}f \cdot 2c\,\Gamma_{i,k}\Gamma_{E,m}\, d_{3,i,k}d_{3,E,m} \\[2pt]
&\quad + \sum_{i\neq j}\sum_{k,l=1}^d\partial^2_{x_{i,k}x_{j,l}}f \cdot 2c\,\Gamma_{i,k}\Gamma_{j,l}\,\sigma_{0,i,k}^{R}\sigma_{0,j,l}^{R} \\[2pt]
&\quad + \sum_{i=1}^N \sum_{k,l,m=1}^d \partial^3_{x_{i,k}x_{i,l}x_{i,m}} f \cdot c\,\Gamma_{i,k}\Gamma_{i,l}\,d_{3,i,m}
+ \sum_{l,m,n=1}^{d_E} \partial^3_{E_l E_m E_n} f \cdot c\,\Gamma_{E,l}\Gamma_{E,m}\,d_{3,E,n}.
\end{aligned}
\label{eq:fully_consistent_rosenblatt_generator}
\end{equation}

\subsection*{Total Integrated Semi-Martingale Rest Core $\mathcal{M}_t^{\mathrm{total}}$}
The process $\mathcal{M}_t^{\mathrm{total}}$ aggregates the compensated stochastic integrals and point process dynamics into a unified local $\Pee$-martingale:
\begin{equation}
\begin{aligned}
\mathcal{M}_t^{\mathrm{total}} &=
\sum_{i=1}^N \int_0^t \nabla_{\Xbf_i} f \cdot \sigma_i dW_i + \sum_{i=1}^N \int_0^t \nabla_{\Xbf_i} f \cdot \sigma_{0,i} dW^0 + \int_0^t \nabla_E f \cdot \sigma_E^w dW_E \\
&\quad + \sum_{i=1}^N \int_0^t \int_{\mathcal{Z}} \bigl[ f(\tau,\Xbf_{\tau-}^g+\ebf_i\otimes\sigma_i^n,\Cbf_{\tau-}^g,E_{\tau-}^g)-f \bigr] \tilde N_i(d\tau,dz) \\
&\quad + \int_0^t \int_{\mathcal{Z}_0} \bigl[ f(\tau,\Xbf_{\tau-}^g+\sum_i\ebf_i\otimes\sigma_{0,i}^n,\Cbf_{\tau-}^g,E_{\tau-}^g)-f \bigr] \tilde N_0(d\tau,dz_0) \\
&\quad + \int_0^t \int_{\mathcal{Z}_E} \bigl[ f(\tau,\Xbf_{\tau-}^g,\Cbf_{\tau-}^g,E_{\tau-}^g+\gamma_E^n)-f \bigr] \tilde N_E(d\tau,dz) \\
&\quad + \sum_{i=1}^N \int_0^t \int_0^\infty \bigl[ f(\tau,\Xbf_{\tau-}^g+\ebf_i\otimes\sigma_i^{gv}z,\Cbf_{\tau-}^g,E_{\tau-}^g)-f \bigr] \tilde N_i^{G}(d\tau,dz) \\
&\quad + \int_0^t \int_0^\infty \bigl[ f(\tau,\Xbf_{\tau-}^g+\sum_i\ebf_i\otimes\sigma_{0,i}^{gv}z,\Cbf_{\tau-}^g,E_{\tau-}^g)-f \bigr] \tilde N_0^{G}(d\tau,dz) \\
&\quad + \int_0^t \int_0^\infty \bigl[ f(\tau,\Xbf_{\tau-}^g,\Cbf_{\tau-}^g,E_{\tau-}^g+\sigma_E^{gv}z)-f \bigr] \tilde N_E^{G}(d\tau,dz) \\
&\quad + \int_0^t \bigl[ f(\tau,s_{\tau}^g,\Xbf_{\tau-}^g,\Cbf_{\tau-}^g,E_{\tau-}^g)-f \bigr] \bigl(dN_{\tau}^{\mathrm{regime}}-\gamma_{\cdot}d\tau\bigr) \\
&\quad + \sum_{i=1}^N \int_0^t \sum_{m'\neq C_{i,\tau-}^g} \bigl[ f(\tau,s_{\tau-}^g,\Xbf_{\tau-}^g,\Cbf_{\sim i,\tau-}^g[m'],E_{\tau-}^g)-f \bigr] \bigl(dN_{i,\tau}^{\mathrm{coalition}}(m')-\lambda_{\cdot}^i(\tau)d\tau\bigr).
\end{aligned}
\label{eq:martingale_closure_equation}
\end{equation}

Let the augmented state be
$
\widetilde{\mathbf{Z}}_t = \bigl(s_t,\;\mathbf{X}_t,\;\mathbf{C}_t,\;E_t,\;\Phi_t,\;\Psi_t\bigr),
$
where
\[
\Phi_t = \bigl(\nabla^{H_2/2}X_{i,k}(t)\bigr)_{i,k},\qquad
\Psi_t = \bigl(\nabla^{H_2/2,H_2/2}X_{i,k}(t,t)\bigr)_{i,k},
\]
and analogously for \(E_t\). The infinitesimal generator \(\mathcal{L}\) of this Markov process (for smooth test functions \(F\)) is given in the previous answer. Its formal adjoint \(\mathcal{L}^*\) with respect to the Lebesgue measure on the continuous variables and the counting measure on the discrete variables satisfies
\[
\int (\mathcal{L}F)\,\rho = \int F\,(\mathcal{L}^*\rho)
\]
for all suitable \(F,\rho\).

\subsection*{Evolution of the augmented density \(\rho\)}

The probability density \(\rho(t,s,\mathbf{x},\mathbf{c},e,\phi,\psi)\) of the augmented state satisfies the Fokker–Planck (Kolmogorov forward) equation
$
{\partial_t \rho = \mathcal{L}^* \rho,}
$
with \(\mathcal{L}^* = \mathcal{L}_{\mathrm{drift}}^* + \mathcal{L}_{\mathrm{diff}}^* + \mathcal{L}_{\mathrm{Rosenblatt}}^* + \mathcal{L}_{\mathrm{GaussVolterra}}^* + \mathcal{L}_{\mathrm{Regime}}^* + \mathcal{L}_{\mathrm{Coalition}}^* + \mathcal{L}_{\mathrm{PoissonJump}}^*\).

The explicit form of each adjoint operator is:

\begin{align}
\mathcal{L}_{\mathrm{drift}}^* \rho &= -\sum_{i,k}\partial_{x_{i,k}}\bigl((b_{i,k}+H_{k}^{i,g}+\mathcal{K}_{\mathrm{rv},k}^{i,g})\rho\bigr) - \sum_m\partial_{E_m}\bigl(b_{E,m}\rho\bigr), \\[4pt]
\mathcal{L}_{\mathrm{diff}}^* \rho &=
\frac12\sum_{i,k,l}\partial^2_{x_{i,k}x_{i,l}}\bigl([\sigma_i\sigma_i^\top]_{k,l}\,\rho\bigr)
+ \frac12\sum_{i,j,k,l}\partial^2_{x_{i,k}x_{j,l}}\bigl([\sigma_{0,i}\sigma_{0,j}^\top]_{k,l}\,\rho\bigr) \nonumber\\
&\quad + \frac12\sum_{m,n}\partial^2_{E_mE_n}\bigl([\sigma_E^w(\sigma_E^w)^\top]_{m,n}\,\rho\bigr)
+ \sum_{i,k,m}\partial^2_{x_{i,k}E_m}\bigl([\sigma_{0,i}(\sigma_E^w)^\top]_{k,m}\,\rho\bigr), \\[4pt]
\mathcal{L}_{\mathrm{Rosenblatt}}^* \rho &=
-\sum_{i,k}\partial_{x_{i,k}}\bigl((2c\,\phi_{i,k}d_{2,i,k}+c\,\psi_{i,k}d_{3,i,k})\rho\bigr)
-\sum_m\partial_{E_m}\bigl((2c\,\phi_{E,m}d_{2,E,m}+c\,\psi_{E,m}d_{3,E,m})\rho\bigr) \nonumber\\
&\quad +\frac12\sum_{i,k,l}\partial^2_{x_{i,k}x_{i,l}}\bigl(c\,\phi_{i,k}\phi_{i,l}d_{3,i,l}\,\rho\bigr)
+\frac12\sum_{m,n}\partial^2_{E_mE_n}\bigl(c\,\phi_{E,m}\phi_{E,n}d_{3,E,n}\,\rho\bigr) \nonumber\\
&\quad +\sum_{i,k,m}\partial^2_{x_{i,k}E_m}\bigl(2c\,\phi_{i,k}\phi_{E,m}d_{3,i,k}d_{3,E,m}\,\rho\bigr) \nonumber\\
&\quad +\sum_{i\neq j,k,l}\partial^2_{x_{i,k}x_{j,l}}\bigl(2c\,\phi_{i,k}\phi_{j,l}\,\sigma_{0,i,k}^{R}\sigma_{0,j,l}^{R}\,\rho\bigr), \\[4pt]
%%%
\mathcal{L}_{\mathrm{GaussVolterra}}^* \rho &=
\frac{1}{2} \sum_{i=1}^N \sum_{k,l=1}^d \partial^2_{x_{i,k}x_{i,l}} \left( \left[ \mathbf{\Omega}_i(t) \mathbf{\Omega}_i^\top(t) \right]_{k,l} \rho \right) 
+ \frac{1}{2} \sum_{i,j=1}^N \sum_{k,l=1}^d \partial^2_{x_{i,k}x_{j,l}} \left( \left[ \mathbf{\Omega}_0(t) \mathbf{\Omega}_0^\top(t) \right]_{k,l} \rho \right) \\[3pt]
&\quad + \frac{1}{2} \sum_{m,n=1}^{d_E} \partial^2_{E_mE_n} \left( \left[ \mathbf{\Omega}_E(t) \mathbf{\Omega}_E^\top(t) \right]_{m,n} \rho \right)
+ \sum_{i=1}^N \sum_{k=1}^d \sum_{m=1}^{d_E} \partial^2_{x_{i,k}E_m} \left( \left[ \mathbf{\Omega}_0(t) \mathbf{\Omega}_E^\top(t) \right]_{k,m} \rho \right). \\[4pt]
%%%
(\mathcal{L}_{\mathrm{Regime}}^* \rho)(s,\cdot) &= \sum_{s'\neq s}\gamma_{s\to s'}\,\rho(s',\cdot) - \sum_{s'\neq s}\gamma_{s'\to s}\,\rho(s,\cdot), \\[4pt]
(\mathcal{L}_{\mathrm{Coalition}}^* \rho)(\mathbf{c}) &= \sum_{i=1}^N\sum_{m'\neq c_i}\lambda_{c_i\to m'}^i\,\rho(\mathbf{c}_{\sim i}[m']) - \sum_{i=1}^N\sum_{m'\neq c_i}\lambda_{m'\to c_i}^i\,\rho(\mathbf{c}), \\[4pt]
\mathcal{L}_{\mathrm{PoissonJump}}^* \rho &=
\sum_{i=1}^N\int_{\mathcal{Z}}\bigl[\rho(\mathbf{x}_i-\sigma_i^n(z))-\rho(\mathbf{x}_i)+\nabla_{\mathbf{x}_i}\rho\cdot\sigma_i^n(z)\bigr]\nu_i(dz) \nonumber\\
&\quad +\int_{\mathcal{Z}_0}\bigl[\rho(\mathbf{x}-\sum_i\ebf_i\otimes\sigma_{0,i}^n(z_0))-\rho(\mathbf{x})+\sum_i\nabla_{\mathbf{x}_i}\rho\cdot\sigma_{0,i}^n(z_0)\bigr]\nu_0(dz_0) \nonumber\\
&\quad +\int_{\mathcal{Z}_E}\bigl[\rho(E-\gamma_E^n(z))-\rho(E)+\nabla_E\rho\cdot\gamma_E^n(z)\bigr]\nu_E(dz) \nonumber\\
&\quad +\sum_{i=1}^N\int_0^\infty\bigl[\rho(\mathbf{x}_i-\sigma_i^{gv}z)-\rho(\mathbf{x}_i)+\nabla_{\mathbf{x}_i}\rho\cdot(\sigma_i^{gv}z)\bigr]\nu_G(dz) \nonumber\\
&\quad +\int_0^\infty\bigl[\rho(\mathbf{x}-\sum_i\ebf_i\otimes\sigma_{0,i}^{gv}z)-\rho(\mathbf{x})+\sum_i\nabla_{\mathbf{x}_i}\rho\cdot(\sigma_{0,i}^{gv}z)\bigr]\nu_G(dz) \nonumber\\
&\quad +\int_0^\infty\bigl[\rho(E-\sigma_E^{gv}z)-\rho(E)+\nabla_E\rho\cdot(\sigma_E^{gv}z)\bigr]\nu_G(dz).
\end{align}

In all expressions, the coefficients are evaluated at \((t,s,\mathbf{x},\mathbf{c},e,\mu_t)\); the dependence on \(\phi,\psi\) is only through the Rosenblatt terms. The fractional derivatives \(\phi,\psi\) are treated as independent state variables with zero drift and zero diffusion in the augmented dynamics, hence no derivatives with respect to them appear in the adjoint (their evolution is deterministic, but we absorb it into the definition of \(\rho\)).

\subsection*{Marginal density \(\mu\) and its non‑closed evolution}

Define the marginal density
\[
\mu(t,s,\mathbf{x},\mathbf{c},e) = \int \rho(t,s,\mathbf{x},\mathbf{c},e,\phi,\psi)\,d\phi\,d\psi.
\]

Because the Rosenblatt operator \(\mathcal{L}_{\mathrm{Rosenblatt}}^*\) contains \(\phi,\psi\) inside the coefficients (multiplying derivatives with respect to \(x\) and \(e\)), the integration over \(\phi,\psi\) does not commute with the derivatives. Consequently, \(\mu\) does not satisfy a closed evolution equation; its dynamics depends on conditional expectations of \(\phi,\psi\) given the marginal state. To obtain a self‑contained forward equation, one must either:

\begin{itemize}
  \item Work directly with the augmented density \(\rho(t,\cdot,\cdot,\cdot,\cdot,\phi,\psi)\) solving \(\partial_t\rho = \mathcal{L}^*\rho\).
  \item Introduce a closure approximation (e.g., assume that \(\phi,\psi\) are deterministic functions of \((\mathbf{x},e)\) or use a moment closure).
  \item Use a pathwise approach where the fractional derivatives are computed from the past trajectory (non‑Markovian) and the forward equation is replaced by a McKean–Vlasov equation with memory.
\end{itemize}

Thus the correct  formulation is to keep the augmented state \(\widetilde{\mathbf{Z}}_t\) and its density \(\rho\), which evolves according to the linear Fokker–Planck equation above. The marginal \(\mu\) is then obtained by integrating \(\rho\) over the fractional variables.

\subsection*{The Augmented Integrated Hamiltonian Operator over $\widetilde{\mu}$}
Let $\widetilde{\mathbf{Z}}_t^g = \bigl(s^g_t,\;\mathbf{X}^g_t,\;\mathbf{C}^g_t,\;E^g_t,\;\Phi^g_t,\;\Psi^g_t\bigr) \in \widetilde{\mathcal{Z}}$ denote the augmented state vector containing the structural micro-states, fluid coalitions, macro-environments, and the non-local historical memory state matrices:
\begin{equation}
\Phi_t = \bigl(\nabla^{H_2/2}X_{i,k}(t)\bigr)_{i,k}, \qquad \Psi_t = \bigl(\nabla^{H_2/2,H_2/2}X_{i,k}(t,t)\bigr)_{i,k}.
\label{eq:augmented_memory_states}
\end{equation}
Let $\widetilde{\mu}_t^g = \operatorname{Law}(\widetilde{\mathbf{Z}}_t^g) \in \mathcal{P}_2(\widetilde{\mathcal{Z}})$ characterize the joint distribution profile of the augmented configuration. For each agent $i \in \cN$ within generation $g$, given the continuous controls $\mathbf{u}_t$, loyalty-shifting efforts $\mathbf{v}_t^g$, and the adjoint pairs $(Y_i^g, Y_E^{i,g})$, the \emph{Augmented Integrated Hamiltonian} $\mathcal{H}_{\mathrm{int}}^i\bigl(t, \widetilde{\mathbf{Z}}_t^g, \widetilde{\mu}_t^g, \mathbf{u}_t, \mathbf{v}_t^g, Y_{i,t}^g, Y_{E,t}^{i,g}\bigr)$ is defined as:
\begin{equation}
\begin{aligned}
\mathcal{H}_{\mathrm{int}}^i\bigl(t, \widetilde{\mathbf{Z}}, \widetilde{\mu}, \mathbf{u}, \mathbf{v}, Y_i, Y_E\bigr) &:= \frac{1}{2} \sum_{j=1}^N \left\langle X_{j} - M^{j,g}(t,s)e_{\alpha_j}(\widetilde{\mu}), \; Q^{j,g}(t,s)\left(X_{j} - M^{j,g}(t,s)e_{\alpha_j}(\widetilde{\mu})\right) \right\rangle + \sum_{j=1}^N \Phi_{\beta_j}^{j,g}\bigl(E, e_{\beta_j}(\widetilde{\mu}), s\bigr) \\[2pt]
&\quad + \sum_{j=1}^N \left\langle b_j\bigl(t, s, \mathbf{X}, \mathbf{C}, E, \Phi, \Psi, \mathbf{u}, \widetilde{\mu}\bigr) + H^{j,g}\bigl(t,s,E,e_{\alpha_j}(\widetilde{\mu})\bigr), \; Y_{j} \right\rangle \\[2pt]
&\quad + \left\langle b_E\bigl(t, s, \mathbf{X}, \mathbf{C}, E, \Phi, \Psi, \mathbf{u}, \widetilde{\mu}\bigr), \; Y_{E} \right\rangle \\[2pt]
&\quad + \sum_{j=1}^N \sum_{m' \neq C_j} \bar{\lambda}_{C_j\to m'}^{\,j}(t,s)\exp\!\Bigl( (v_j)^\top D_{C_j m'}^j(t,s)\,\mathbf{X}_j \;+\; \Psi_{C_j m'}\bigl(E,\, e_{\beta_j}(\widetilde{\mu}), s\bigr) \Bigr) \\[2pt]
&\quad \times \left[ f\bigl(t, s, \mathbf{X}, \mathbf{C}_{\sim j}[m'], E\bigr) - f\bigl(t, s, \mathbf{X}, \mathbf{C}, E\bigr) \right].
\end{aligned}
\label{eq:augmented_hamiltonian_def}
\end{equation}

To capture the mean-field dependencies across the augmented population layout, variations with respect to the distribution are executed via the first-order \emph{G{\^a}teaux derivative}. For any perturbing signed measure $\mathbf{\chi}$, the G{\^a}teaux variation in the direction $\mathbf{\chi}$ is uniquely identified by its density gradient $\frac{\delta \mathcal{H}_{\mathrm{int}}^i}{\delta \widetilde{\mu}}: \widetilde{\mathcal{Z}} \to \mathbb{R}$:
\begin{equation}
\lim_{\epsilon \to 0} \frac{\mathcal{H}_{\mathrm{int}}^i(t, \widetilde{\mathbf{Z}}, \widetilde{\mu} + \epsilon \mathbf{\chi}, \dots) - \mathcal{H}_{\mathrm{int}}^i(t, \widetilde{\mathbf{Z}}, \widetilde{\mu}, \dots)}{\epsilon} = \int_{\widetilde{\mathcal{Z}}} \frac{\delta \mathcal{H}_{\mathrm{int}}^i}{\delta \widetilde{\mu}}\bigl(t, \widetilde{\mathbf{Z}}, \widetilde{\mu}, \dots\bigr)(\widetilde{\mathbf{z}}') \, d\mathbf{\chi}(\widetilde{\mathbf{z}}').
\end{equation}
The spatial gradients of this functional density evaluate structural sensitivities at explicit realizations within the augmented field space.

\bigskip

\begin{theorem}[Complete Path-Dependent Adjoint Characterisation with All Memory-Noise Sensitivities]
Let $(\mathbf{u}^{*,g}, \mathbf{v}^{*,g})_{g=0}^G$ be a Nash equilibrium of the intergenerational game. Then, for each generation $g \in \{0, \dots, G\}$ and each agent $i \in \cN$, there exist unique pairs of $\mathbb{F}$-adapted processes $(Y_{i,t}^g, \Lambda_{i,t}^g) \in \R^d \times \R^{d \times M}$ and $(Y_{E,t}^{i,g}, \Lambda_{E,t}^{i,g}) \in \R^{d_E} \times \R^{d_E \times M}$ satisfying the following backward stochastic differential system within the epoch $t \in [T_g, T_{g+1})$:

\begin{equation}
\begin{aligned}
-dY_{i,t}^g &= \Biggl\{ \nabla_{\mathbf{x}_i} \mathcal{H}_{\mathrm{int}}^i\bigl(t, \widetilde{\mathbf{Z}}_t^{*,g}, \widetilde{\mu}_t^g, \mathbf{u}_t^*, \mathbf{v}_t^{*,g}, Y_{i,t}^g, Y_{E,t}^{i,g}\bigr) \\[2pt]
&\quad + \nabla_{\boldsymbol{\phi}_i} \mathcal{H}_{\mathrm{int}}^i\bigl(t, \widetilde{\mathbf{Z}}_t^{*,g}, \widetilde{\mu}_t^g, \dots\bigr) \cdot \left( H_1 t^{2H_1-1} \sigma_i^{H_1} + 2\tilde{c} \, \sigma_i^{H_2} \right) + \nabla_{\boldsymbol{\psi}_i} \mathcal{H}_{\mathrm{int}}^i\bigl(t, \widetilde{\mathbf{Z}}_t^{*,g}, \widetilde{\mu}_t^g, \dots\bigr) \cdot \sigma_i^R \\[2pt]
&\quad + \mathbb{E}_{\widetilde{\mathbf{Z}}'}\left[ \nabla_{\mathbf{x}} \left( \frac{\delta \mathcal{H}_{\mathrm{int}}^i}{\delta \widetilde{\mu}}\bigl(t, \widetilde{\mathbf{Z}}_t^{*,g}, \widetilde{\mu}_t^g, \mathbf{u}_t^*, \mathbf{v}_t^{*,g}, Y_{i,t}^g, Y_{E,t}^{i,g}\bigr)(\widetilde{\mathbf{Z}}'_{t}) \right) \cdot \mathcal{H}_{\mathrm{feedback}}^{i,t} \right] \\[2pt]
&\quad + \mathbb{E}_{\widetilde{\mathbf{Z}}'}\left[ \nabla_{\boldsymbol{\phi}} \left( \frac{\delta \mathcal{H}_{\mathrm{int}}^i}{\delta \widetilde{\mu}}\bigl(t, \widetilde{\mathbf{Z}}_t^{*,g}, \dots\bigr)(\widetilde{\mathbf{Z}}'_{t}) \right) \cdot \boldsymbol{\Omega}_{\mathrm{fbm,gv}}^{i,t} + \nabla_{\boldsymbol{\psi}} \left( \frac{\delta \mathcal{H}_{\mathrm{int}}^i}{\delta \widetilde{\mu}}\bigl(t, \widetilde{\mathbf{Z}}_t^{*,g}, \dots\bigr)(\widetilde{\mathbf{Z}}'_{t}) \right) \cdot \boldsymbol{\Omega}_{\mathrm{Rosenblatt}}^{i,t} \right] \\[2pt]
&\quad + \sum_{m' \neq C_{i,t}^{*,g}} \lambda_{C_{i,t}^{*,g} \to m'}^i\bigl(t, s_t^g, \mathbf{X}_t^{*,g}, E_t^{*,g}, \widetilde{\mu}_t^g, v_{i,t}^{*,g}\bigr) D_{C_{i,t}^{*,g} m'}^i(t,s_t^g)^\top v_{i,t}^{*,g} \\[2pt]
&\quad + \int_t^{T_{g+1}} (s-t)^{-\gamma} \mathbb{M}_i(s,s_s^g)^\top \E\left[ Y_{i,s}^g \,\middle|\, \F_t \right] ds \Biggl\} dt \\[4pt]
&\quad - \Lambda_{i,t}^g dM_t^s - \int_{\mathcal{Z}} \Gamma_{i,t}^g(z) \tilde{N}_i(dt,dz) - \int_{\mathcal{Z}_0} \Gamma_{0,i,t}^g(z_0) \tilde{N}^0(dt,dz_0),
\end{aligned}
\label{eq:adjoint_individual_augmented_complete}
\end{equation}
and
\begin{equation}
\begin{aligned}
-dY_{E,t}^{i,g} &= \Biggl\{ \nabla_{e} \mathcal{H}_{\mathrm{int}}^i\bigl(t, \widetilde{\mathbf{Z}}_t^{*,g}, \widetilde{\mu}_t^g, \mathbf{u}_t^*, \mathbf{v}_t^{*,g}, Y_{i,t}^g, Y_{E,t}^{i,g}\bigr) \\[2pt]
&\quad + \nabla_{\boldsymbol{\phi}_E} \mathcal{H}_{\mathrm{int}}^i\bigl(t, \widetilde{\mathbf{Z}}_t^{*,g}, \widetilde{\mu}_t^g, \dots\bigr) \cdot \left( H_1 t^{2H_1-1} \sigma_E^{fbm,H_1} + 2\tilde{c} \, \sigma_E^{fbm,H_2} \right) + \nabla_{\boldsymbol{\psi}_E} \mathcal{H}_{\mathrm{int}}^i\bigl(t, \widetilde{\mathbf{Z}}_t^{*,g}, \widetilde{\mu}_t^g, \dots\bigr) \cdot \sigma_E^R \\[2pt]
&\quad + \mathbb{E}_{\widetilde{\mathbf{Z}}'}\left[ \nabla_{\mathbf{z}} \left( \frac{\delta \mathcal{H}_{\mathrm{int}}^i}{\delta \widetilde{\mu}}\bigl(t, \widetilde{\mathbf{Z}}_t^{*,g}, \widetilde{\mu}_t^g, \mathbf{u}_t^*, \mathbf{v}_t^{*,g}, Y_{i,t}^g, Y_{E,t}^{i,g}\bigr)(\widetilde{\mathbf{Z}}'_{t}) \right) \cdot \mathcal{G}_{\mathrm{tail}}^{i,t} \right] \\[2pt]
&\quad + \mathbb{E}_{\widetilde{\mathbf{Z}}'}\left[ \nabla_{\boldsymbol{\phi}} \left( \frac{\delta \mathcal{H}_{\mathrm{int}}^i}{\delta \widetilde{\mu}}\bigl(t, \widetilde{\mathbf{Z}}_t^{*,g}, \dots\bigr)(\widetilde{\mathbf{Z}}'_{t}) \right) \cdot \boldsymbol{\Omega}_{E,\mathrm{fbm}}^{i,t} + \nabla_{\boldsymbol{\psi}} \left( \frac{\delta \mathcal{H}_{\mathrm{int}}^i}{\delta \widetilde{\mu}}\bigl(t, \widetilde{\mathbf{Z}}_t^{*,g}, \dots\bigr)(\widetilde{\mathbf{Z}}'_{t}) \right) \cdot \boldsymbol{\Omega}_{E,\mathrm{Rosenblatt}}^{i,t} \right] \Biggl\} dt \\[4pt]
&\quad - \Lambda_{E,t}^{i,g} dM_t^s - \int_{\mathcal{Z}_E} \Gamma_{E,t}^{i,g}(z) \tilde{N}_E(dt,dz),
\end{aligned}
\label{eq:adjoint_environmental_augmented_complete}
\end{equation}
where $\nabla_{\boldsymbol{\phi}_i}$ and $\nabla_{\boldsymbol{\psi}_i}$ denote the spatial gradients taken directly with respect to the continuous memory coordinates of agent $i$ within the augmented vector block, and the operators $\boldsymbol{\Omega}_{\mathrm{fbm,gv}}^{i,t}$ and $\boldsymbol{\Omega}_{\mathrm{Rosenblatt}}^{i,t}$ characterize the cross-agent mean-field coupling tensors tracking how variations in the collective distribution of fractional Brownian paths, general Gauss-Volterra processes, and third-order Rosenblatt kernels feed back onto individual adjoint velocities.
\end{theorem}

%%%%%%%%%%%%% 

\begin{proof}
The derivation proceeds by defining a localized spike variation on the optimal control profile $(u_{i,t}^{*,g}, v_{i,t}^{*,g})$. Let $u_{i,t}^{\epsilon} = u_{i,t}^{*,g} + \epsilon \delta u_{i,t} \mathbbm{1}_{[t_0, t_0+\epsilon]}(t)$. Differentiating the state dynamics \eqref{eq:true_dyn} with respect to $\epsilon$ at $\epsilon=0$ yields the linearized state trajectory $\delta X_{i,t}^g$. 

Applying the infinite-dimensional duality formula between the linearized forward system and the backward processes requires taking the Gâteaux derivative of the cost functional \eqref{eq:cost}. The presence of the Volterra integral operator $\mathcal{K}_{\mathrm{rv}}^{i,g}(t)$ introduces a dual singular Volterra convolution term in the drift of the adjoint equation, manifested as the integral $\int_t^{T_{g+1}} (s-t)^{-\gamma} \mathbb{M}_i^\top \E[Y_{i,s}^g|\F_t]ds$. 

The explicit appearance of the expectation operators inside the drift accounts for the mean-field-type interaction, wherein an individual perturbation impacts the aggregate law $\mu_t^g$, which subsequently feeds back into the system through the expectile weights $\omega_\alpha$ derived in \eqref{eq:lions_expectile} and \eqref{eq:omega_weight}.
\end{proof}

 Exact Boundary Conditions across Generations: To complete the backward induction across the macroscopic planning horizon, the adjoint processes must be coupled at each generational interface $T_{g+1}$. Because the state variables experience a non-local push-forward $\mu_{T_{g+1}}^{g+1} = \mathcal{T}_{\boldsymbol{\zeta}}[\mu_{T_{g+1}-}^g]$, the adjoint variables undergo a structural contraction matching the hereditary transmission operators defined in Section~\ref{sec:transmission}.

\begin{proposition}[Trans-Generational Coupling]
The terminal conditions for the adjoint processes of generation $g$ at time $T_{g+1}-$ are determined by the initial values of the adjoint processes of generation $g+1$ at time $T_{g+1}$ via the following equations:
\begin{equation}
Y_{i, T_{g+1}-}^g = \zeta_i \, Y_{i, T_{g+1}}^{g+1} + (1-\zeta_i) \E\left[ Y_{i, T_{g+1}}^{g+1} \right] + \nabla_{\mathbf{x}_i} G^{i,g}\bigl(\mathbf{Z}_{T_{g+1}-}^g, s_{T_{g+1}-}^g, \mu_{T_{g+1}-}^g\bigr),
\label{eq:boundary_X}
\end{equation}
and
\begin{equation}
Y_{E, T_{g+1}-}^{i,g} = \mathbb{A}_E^\top \, Y_{E, T_{g+1}}^{i,g+1} + \nabla_{e} G^{i,g}\bigl(\mathbf{Z}_{T_{g+1}-}^g, s_{T_{g+1}-}^g, \mu_{T_{g+1}-}^g\bigr).
\label{eq:boundary_E}
\end{equation}
At the absolute macroscopic horizon $T_{G+1}$, the boundary conditions collapse to the standard terminal gradients:
\begin{equation}
Y_{i, T_{G+1}}^G = \nabla_{\mathbf{x}_i} G^{i,G}(\dots), \qquad Y_{E, T_{G+1}}^{i,G} = \nabla_{e} G^{i,G}(\dots).
\label{eq:absolute_terminal}
\end{equation}
\end{proposition}

Equations \eqref{eq:boundary_X} and \eqref{eq:boundary_E} illustrate the  structure of the conflict trap. The terminal sensitivity $Y_{i, T_{g+1}-}^g$ is directly amplified by the trauma transmission coefficient $\zeta_i$. When $\zeta_i \to 1$, the strategic cost sensitivities are transmitted directly to the next cohort, preventing the decay of adversarial incentives across time. By applying the stochastic maximum principle, the optimal continuous control profile $u_{i,t}^{*,g}$ and discrete coalition switching effort $v_{i,t}^{*,g}$ are obtained by minimizing the internal structural Hamiltonian. Given the  cost weights $R^{i,g}$ and $S^{i,g}$ defined in \eqref{eq:running}, the first-order necessary conditions yield the following explicit mappings:

\begin{equation}
u_{i,t}^{*,g} = - R^{i,g}(t,s_t^g)^{-1} \left[ \nabla_{u_i} b_i^\top Y_{i,t}^g + \nabla_{u_i} b_E^\top Y_{E,t}^{i,g} \right],
\label{eq:optimal_u}
\end{equation}
and
\begin{equation}
v_{i,t}^{*,g} = - S^{i,g}(t,s_t^g)^{-1} \sum_{m' \neq C_{i,t}^{*,g}} \left[ \lambda_{C_{i,t}^{*,g} \to m'}^i(\dots) D_{C_{i,t}^{*,g} m'}^i(t,s_t^g) \, X_{i,t}^{*,g} \right].
\label{eq:optimal_v}
\end{equation}

Substituting \eqref{eq:optimal_u} and \eqref{eq:optimal_v} into the forward system \eqref{eq:true_dyn} and \eqref{eq:E_true_dyn} yields a closed, highly non-linear forward-backward system of stochastic Volterra diffusions. The existence of a solution to this system which defines the game's equilibrium, is addressed in the subsequent section.

\section{The Generational Ergodic Breakdown and War Trap Threshold}
\label{sec:breakdown}

We now use the  framework established above to state and prove our central theoretical results: the existence of the equilibrium and the identification of the sharp structural threshold that separates a state of long-term sustainable peace from an inescapable, self-sustaining {\it civil war}. We begin by establishing the existence and uniqueness of the Nash equilibrium under standard local regularity conditions. Let $\mathbb{H}^2$ denote the Hilbert space of square-integrable, $\mathbb{F}$-adapted processes on the planning horizon.

\begin{theorem}[Existence and Uniqueness of Equilibrium]
\label{thm:existence}
Assume that the drift and diffusion coefficients in \eqref{eq:true_dyn} and \eqref{eq:E_true_dyn} are Lipschitz continuous with respect to their state and measure arguments, and that the cost functions $L^{i,g}$ and $G^{i,g}$ are continuously differentiable and satisfy the strict convexity conditions $R^{i,g} \succ 0$, $S^{i,g} \succ 0$, and $Q^{i,g} \succeq 0$. Then, for each generation $g$, there exists a positive constant $\tau_g > 0$ such that if the epoch duration satisfies $T_{g+1} - T_g < \tau_g$, the forward-backward  system \eqref{eq:true_dyn}, \eqref{eq:E_true_dyn}, (\ref{eq:adjoint_individual_augmented_complete})-(\ref{eq:adjoint_environmental_augmented_complete}) admits a unique solution $(\mathbf{X}^{*,g}, \mathbf{E}^{*,g}, \mathbf{Y}^{*,g}, \mathbf{Y}_E^{*,g}) \in \mathbb{H}^2$. This unique solution characterizes the unique Nash equilibrium of the generation-$g$ subgame.
\end{theorem}

\begin{proof}
The proof relies on a contraction mapping argument on the Wasserstein space developed specifically for mean-field forward-backward Volterra systems. Define an operator $\Phi$ on $\mathbb{H}^2 \times \mathbb{H}^2$ that takes an arbitrary state-measure trajectory, computes the corresponding unique adjoint processes via the linear backward Volterra equations, updates the controls via the Hamiltonian minimizers \eqref{eq:optimal_u}--\eqref{eq:optimal_v}, and resolves the forward McKean-Vlasov equations. 

Under the specified Lipschitz assumptions and the singular fractional integration kernel $(t-s)^{-\gamma}$, the operator $\Phi$ satisfies a Lipschitz-type estimate with a constant proportional to $(T_{g+1}-T_g)^{1-\gamma}$. Choosing $\tau_g$ sufficiently small ensures that $(T_{g+1}-T_g)^{1-\gamma} < 1$, rendering $\Phi$ a strict contraction. Application of the Banach Fixed-Point Theorem guarantees the existence of a unique fixed point, which corresponds to the unique Nash equilibrium.
\end{proof}

While Theorem~\ref{thm:existence} provides local stability and uniqueness within a single generation, it does not guarantee that the system remains stable when integrated across an infinite sequence of generations ($G \to \infty$). In historical conflict zones like Mali, the critical question is whether the peaceful steady state is structurally resilient to unexpected shocks or if it behaves as an unstable equilibrium that can be easily permanently disrupted.

To formalize this, we define the peaceful steady state as a structural configuration where individual radicalization indices are zero, local assets are stable, and the macro-political regime is stabilized in peacetime ($s_t^g = 1$). Let $\mathbf{A}_{\mathrm{grievance}} \in \R^{N \times N}$ be the effective matrix of intergenerational trauma transmission, whose entries are given by $[A_{\mathrm{grievance}}]_{ij} = \zeta_i \cdot \|\mathbb{M}_i(s, 1)\|$. This matrix captures how intensely grievances are transmitted and transformed across cohorts. Furthermore, let $\mathbf{A}_{\mathrm{laundering}} \in \R^{N_E \times N_E}$ represent the offshore laundering loop matrix for the subset of war entrepreneurs $\cN_E$, with elements $[A_{\mathrm{laundering}}]_{ij} = \eta_{\mathrm{wash}} \cdot \nabla_{X_{\mathrm{local}}} \Pi_{\mathrm{local}}$.

We construct the Effective System Operator $\mathbf{A}_{\mathrm{system}}$ as the block-structured operator aggregating these two feedback channels:
\begin{equation}
\mathbf{A}_{\mathrm{system}} = 
\begin{pmatrix}
\mathbf{A}_{\mathrm{grievance}} & \mathbf{0} \\
\mathbf{0} & \mathbf{A}_{\mathrm{laundering}}
\end{pmatrix}.
\label{eq:system_matrix}
\end{equation}
Let $\rho(\mathbf{A}_{\mathrm{system}})$ denote the spectral radius (the maximum absolute eigenvalue) of the operator $\mathbf{A}_{\mathrm{system}}$. The next theorem identifies this spectral radius as the critical threshold governing the long-term structural evolution of the state.

\begin{theorem}[Generational Ergodic Breakdown]
\label{thm:breakdown}
Let the total number of generations $G \to \infty$. The structural resilience of the peaceful steady state is governed strictly by the spectral radius $\rho(\mathbf{A}_{\mathrm{system}})$:
\begin{enumerate}
    \item \textbf{Subcritical Regime ($\rho(\mathbf{A}_{\mathrm{system}}) < 1$): Ergodic Stability.} The peaceful steady state is globally asymptotically stable on the Wasserstein manifold. For any arbitrary initial measure perturbation $\mu_{T_0}^0$, the sequence of joint measures converges exponentially to the peaceful equilibrium:
    \begin{equation}
    \lim_{g \to \infty} \mathcal{W}_2\left(\mu_{t}^{g}, \mu^{\mathrm{peace}}\right) = 0, \qquad \forall t \in [T_g, T_{g+1}].
    \end{equation}
    The system is self-healing: the memory of violence decays faster than it can be reinforced, and historical grievances eventually vanish.
    
    \item \textbf{Supercritical Regime ($\rho(\mathbf{A}_{\mathrm{system}}) \ge 1$): Ergodic Breakdown (The War Trap).} The peaceful steady state becomes unstable. For any neighborhood $\mathcal{U}$ of $\mu^{\mathrm{peace}}$ on the Wasserstein manifold, there exists an infinitesimal shock $\epsilon > 0$ such that a single catastrophic event (e.g., a localized drought or a targeted massacre) shifts the initial measure out of $\mathcal{U}$, causing the system to diverge permanently from the peaceful equilibrium:
    \begin{equation}
    \Pee\left( \lim_{g \to \infty} \mathcal{W}_2\left(\mu_{t}^{g}, \mu^{\mathrm{peace}}\right) > 0 \right) = 1.
    \end{equation}
    The system enters a permanent war trap: the feedback loop between inherited trauma and the compounding offshore profits of war entrepreneurs becomes self-sustaining, locking the society into a permanent cycle of conflict across generations.
\end{enumerate}
\end{theorem}

\begin{proof}
Consider the linearized evolution of the expected state vector across generations. Let $\mathbf{m}_g \in \R^{d \times N}$ denote the mean state vector $\int \mathbf{z} \, d\mu_{T_g}^g(\mathbf{z})$ at the start of generation $g$. By recursively substituting the continuous true state dynamics \eqref{eq:true_dyn}, the Volterra revenge kernel \eqref{eq:revenge}, and the intergenerational continuous transmission rule \eqref{eq:X_jump}, we obtain a discrete-time vector autoregressive system mapping $\mathbf{m}_g$ to $\mathbf{m}_{g+1}$.

Linearizing this map around the peaceful steady state reveals that the cross-generational state transition is dominated by the linear operator $\mathbf{A}_{\mathrm{system}}$. Specifically, the discrete evolution satisfies:
\begin{equation}
\mathbf{m}_{g+1} = \mathbf{A}_{\mathrm{system}} \, \mathbf{m}_g + \mathbf{\epsilon}_{g+1},
\label{eq:linearized_seq}
\end{equation}
where $\mathbf{\epsilon}_{g+1}$ aggregates the zero-mean random innovations $\xi_i^{g+1}$ and filtered orthogonal noise components. 

By Gelfand's formula, the long-term asymptotic growth rate of the sequence $\|\mathbf{m}_g\|$ is determined by the spectral radius:
\begin{equation}
\lim_{g \to \infty} \|\mathbf{m}_g\|^{1/g} = \rho(\mathbf{A}_{\mathrm{system}}).
\label{eq:gelfand}
\end{equation}
If $\rho(\mathbf{A}_{\mathrm{system}}) < 1$, the operator is a strict contraction on the Banach space, and the sequence converges exponentially to zero (the peaceful baseline), proving statement 1. 

If $\rho(\mathbf{A}_{\mathrm{system}}) \ge 1$, the linear operator possesses at least one eigenvalue $\lambda$ with $|\lambda| \ge 1$. The corresponding eigenspace defines an unstable manifold on the Wasserstein space. Any random perturbation (such as a Poisson-driven drought shock $\tilde{N}_E$ or a regime switch to asymmetric warfare $s_t^g=3$) that projects non-trivially onto this unstable eigenspace will expand exponentially across generations. Because the noise terms ensure that such projections occur with probability one, the system permanently diverges from the peaceful baseline, establishing the permanent war trap of statement 2.
\end{proof}

Theorem~\ref{thm:breakdown} provides an explanation for the historical persistence of civil wars in the Sahel. It demonstrates that peacebuilding failures are not merely the result of policy execution errors, but are driven by structural features of the system's spectral geometry. If the intergenerational transmission of grievances and the efficiency of offshore laundering combine to push the system's spectral radius above unity, the conflict becomes a structural property of the society. Under these conditions, standard short-term interventions such as temporary ceasefires or localized development aid, are incapable of inducing permanent peace; they merely alter the transient trajectory without shifting the underlying unstable spectral radius. To achieve permanent stability, an intervention must directly alter the components of the operator $\mathbf{A}_{\mathrm{system}}$, a structural policy design that we formalize in the next section.

\section{Optimal Institutional Policy Design}
\label{sec:policy}

To break the permanent war trap characterized by Theorem~\ref{thm:breakdown}, a central authority or a coalition of international institutions must deploy an intervention strategy capable of structurally shifting the system's spectral radius below unity. In this section, we model an institutional policy framework that combines financial regulation with traditional social structures to stabilize the system.
We introduce an Institutional Policymaker who acts as a Stackelberg leader. The policymaker implements a two-pronged policy profile $\boldsymbol{\Omega}_t = (\theta_{\mathrm{tax},t}, \sigma_{\mathrm{sub},t})$:
\begin{enumerate}
    \item $\theta_{\mathrm{tax},t} \in [0,1]$: an offshore transaction tax targeting the laundering pipeline of war entrepreneurs. This tax reduces their effective laundering efficiency from $\eta_{\mathrm{wash}}$ to $(\eta_{\mathrm{wash}} - \theta_{\mathrm{tax},t})$.
    \item $\sigma_{\mathrm{sub},t} \ge 0$: an endogenous mediation subsidy injected into community-level dispute resolution mechanisms. This subsidy is structured to reinforce traditional conflict-mitigation practices, such as the joking kinship system (sinankunya), while neutralizing their co-optation by predatory elites.
\end{enumerate}
The introduction of the policy profile $\boldsymbol{\Omega}_t$ modifies the structural dynamics of the system in two ways. First, the global wealth accumulation of war entrepreneurs \eqref{eq:global_wealth} becomes:
\begin{equation}
\begin{aligned}
dX_{i,t}^{g,\text{global}} &= \Bigl[\, r_{\text{global}}\, X_{i,t}^{g,\text{global}} + (\eta_{\text{wash}} - \theta_{\mathrm{tax},t})\;\Pi_{\text{local}}(\dots) - u_{i,t}^{g,\text{global}} \Bigr] dt + \sigma_{\text{global}}\, X_{i,t}^{g,\text{global}} \; dW_{i,t}^{g,\text{global}}.
\end{aligned}
\label{eq:taxed_global}
\end{equation}
Second, the grievance matrix $\mathbb{M}_i(s, \mathfrak{s})$ inside the Volterra revenge operator \eqref{eq:revenge} is attenuated by the mediation subsidy, reflecting the accelerated decay of historical traumas due to institutionalized reconciliation:
\begin{equation}
\mathbb{M}_i^{\mathrm{policy}}(s, \mathfrak{s}) = \mathbb{M}_i(s, \mathfrak{s}) \cdot \exp\left( - \kappa_i \, \sigma_{\mathrm{sub},s} \right),
\label{eq:subsidized_grievance}
\end{equation}
where $\kappa_i > 0$ represents the structural efficacy of the mediation subsidy within agent $i$'s community.
The policymaker aims to minimize a long-term institutional cost functional that balances system-wide instability against the fiscal costs of maintaining the policy:
\begin{equation}
\min_{\boldsymbol{\Omega}} \E \Biggl[ \sum_{g=0}^\infty \int_{T_g}^{T_{g+1}} e^{-rt} \left( \mathcal{W}_2^2\left(\mu_t^g, \mu^{\mathrm{peace}}\right) + \frac{\gamma_{\mathrm{tax}}}{2} \theta_{\mathrm{tax},t}^2 + \frac{\gamma_{\mathrm{sub}}}{2} \sigma_{\mathrm{sub},t}^2 \right) dt \Biggl],
\label{eq:policy_objective}
\end{equation}
subject to the modified forward-backward system and the structural constraint that the system must converge to the peaceful steady state.
By applying the policy profile $\boldsymbol{\Omega}$, the effective system matrix $\mathbf{A}_{\mathrm{system}}$ defined in \eqref{eq:system_matrix} is transformed into a policy-dependent operator $\mathbf{A}_{\mathrm{policy}}(\boldsymbol{\Omega})$. Its components are given by:
\begin{equation}
[\mathbf{A}_{\mathrm{grievance}}(\sigma_{\mathrm{sub}})]_{ij} = \zeta_i \cdot \|\mathbb{M}_i(s, 1)\| \cdot e^{-\kappa_i \sigma_{\mathrm{sub}}},
\end{equation}
and
\begin{equation}
[\mathbf{A}_{\mathrm{laundering}}(\theta_{\mathrm{tax}})]_{ij} = (\eta_{\mathrm{wash}} - \theta_{\mathrm{tax}}) \cdot \nabla_{X_{\mathrm{local}}} \Pi_{\mathrm{local}}.
\end{equation}
The next theorem provides the optimal policy design that stabilizes the system and breaks the war trap.

\begin{theorem}[Optimal Stabilization Policy]
\label{thm:policy_stabilization}
Let $\rho(\mathbf{A}_{\mathrm{system}}) \ge 1$, so that the system is initially trapped in the supercritical war regime. Assume the policymaker's fiscal capacities satisfy the boundary conditions $\gamma_{\mathrm{tax}} > 0$ and $\gamma_{\mathrm{sub}} > 0$. Then, the optimal policy profile $\boldsymbol{\Omega}^* = (\theta_{\mathrm{tax},t}^*, \sigma_{\mathrm{sub},t}^*)$ that satisfies the institutional objective \eqref{eq:policy_objective} is uniquely characterized by the following properties:

 \textbf{Spectral Forcing:} The optimal controls $(\theta_{\mathrm{tax},t}^*, \sigma_{\mathrm{sub},t}^*)$ force the spectral radius of the modified operator strictly below unity:
    \begin{equation}
    \rho\left( \mathbf{A}_{\mathrm{policy}}(\boldsymbol{\Omega}^*) \right) < 1.
    \end{equation}
    This structural shift transforms the supercritical war trap into a subcritical self-healing system, ensuring 
    that $\lim_{g \to \infty} \mathcal{W}_2\left(\mu_t^g, \mu^{\mathrm{peace}}\right) = 0$.
    
    \textbf{Explicit Feedback Form:} The optimal policy profile satisfies the coupled feedback equations:
    \begin{equation}
    \theta_{\mathrm{tax},t}^* = \min\left\{ 1, \; \frac{1}{\gamma_{\mathrm{tax}}} \sum_{i \in \cN_E} \Pi_{\mathrm{local}}(\dots) \cdot \E\left[ Y_{E,t}^{i,g} \cdot X_{i,t}^{*,g,\text{global}} \right] \right\},
    \label{eq:optimal_tax}
    \end{equation}
    and
    \begin{equation}
    \sigma_{\mathrm{sub},t}^* = \frac{1}{\gamma_{\mathrm{sub}}} \sum_{i \in \cN} \kappa_i \, e^{-\kappa_i \sigma_{\mathrm{sub},t}^*} \|\mathbb{M}_i\| \cdot \E\left[ Y_{i,t}^g \cdot \int_{T_g}^t (t-s)^{-\gamma} \mathbf{m}_s \, ds \right],
    \label{eq:optimal_sub}
    \end{equation}
    where $Y_{i,t}^g$ and $Y_{E,t}^{i,g}$ are the solution trajectories of the adjoint equations. % \eqref{eq:adjoint_individual} and \eqref{eq:adjoint_environmental}.
\end{theorem}

\begin{proof}
The proof applies the Stackelberg stochastic maximum principle to the institutional optimization problem. We define the institutional Hamiltonian $\mathcal{H}_{\mathrm{inst}}$ by appending the modified state equations \eqref{eq:taxed_global} and \eqref{eq:subsidized_grievance} to the objective function \eqref{eq:policy_objective}, using the agents' adjoint states as structural constraints.
Differentiating $\mathcal{H}_{\mathrm{inst}}$ with respect to the policy instruments $\theta_{\mathrm{tax}}$ and $\sigma_{\mathrm{sub}}$ yields the first-order necessary conditions. For the financial instrument, the derivative yields $\gamma_{\mathrm{tax}} \theta_{\mathrm{tax}} - \sum_i \Pi_{\mathrm{local}} \E[Y_E X^{\mathrm{global}}] = 0$, which, under the admissibility constraint $\theta_{\mathrm{tax}} \in [0,1]$, results in the bounded projection mapping \eqref{eq:optimal_tax}. For the mediation instrument, differentiating through the exponential grievance decay function introduces the structural efficacy parameter $\kappa_i$, yielding the implicit transcendental feedback equation \eqref{eq:optimal_sub}.
To verify statement 1, suppose ad absurdum that the optimal policy resulted in $\rho(\mathbf{A}_{\mathrm{policy}}) \ge 1$. By Theorem~\ref{thm:breakdown}, the state trajectory would diverge permanently from $\mu^{\mathrm{peace}}$, causing the tracking term $\mathcal{W}_2^2\left(\mu_t^g, \mu^{\mathrm{peace}}\right)$ to expand exponentially across generations. This divergence would cause the integral in \eqref{eq:policy_objective} to blow up to infinity. Since $\gamma_{\mathrm{tax}}$ and $\gamma_{\mathrm{sub}}$ are finite, any constant policy profile that forces $\rho < 1$ achieves a finite total cost, which strictly undercuts an infinite cost profile. Thus, by the optimality of $\boldsymbol{\Omega}^*$, the policy must force the spectral radius strictly below unity, completing the proof.
\end{proof}

Theorem~\ref{thm:policy_stabilization} establishes a foundation for long-term peacebuilding architectures. It demonstrates that standard military interventions or isolated financial audits are insufficient to resolve multi-generational conflicts. 
The optimal policy requires a coordinated approach: an offshore transaction tax $\theta_{\mathrm{tax}}^*$ to disconnect the laundering pipelines of war entrepreneurs, combined with targeted mediation subsidies $\sigma_{\mathrm{sub}}^*$ to permanently dismantle the grievance memory kernel. By anchoring the parameters $(\theta^*, \sigma^*)$ within the exact trajectories of the population-environment adjoint states, this framework provides an adaptive framework designed to disrupt the structural mechanics of the conflict trap.

\section{Limitations} \label{sec:policy_failure}

\subsection{The Failure of Asymptotic Policy Design: Structural and Operational Non-Feasibility in the Malian Context}
While the Stackelberg institutional policy profile $\boldsymbol{\Omega}_t = (\theta_{\mathrm{tax},t}, \sigma_{\mathrm{sub},t})$ derived in Theorem~\ref{thm:policy_stabilization} provides a  closure to the forward-backward system on the infinite-dimensional Wasserstein manifold, its structural coherence breaks down completely when confronted with the empirical realities of the Malian theatre. The model treats the institutional policymaker as a monolithic, benevolent sovereign possessing global information visibility and unilateral enforcement capacity. In reality, the Malian state functions within a situation of extreme territorial fragmentation, institutional co-optation, and pervasive economic informality. We provide a critique demonstrating that the proposed policy design is not only operationally unfeasible but  counterproductive, actively driving the system deeper into the supercritical war trap.

\medskip
\noindent\textbf{Nullification of the Offshore Transaction Tax ($\theta_{\mathrm{tax},t}$).}
The formulation in \eqref{eq:taxed_global} posits that the state can impose a continuous, deterministic tax rate $\theta_{\mathrm{tax},t} \in [0,1]$ to directly attenuate the laundering efficiency $\eta_{\mathrm{wash}}$ of the entrepreneurial class $\cN_E$. This mechanism fails fundamentally due to two structural properties of the Sahelian war economy:

\medskip
\noindent\textbf{Inversion of the Stackelberg Topology (State Capture).}
The Stackelberg paradigm assumes a strict hierarchical information asymmetry where the policymaker restricts the strategy spaces of the followers. However, in contemporary Mali, war entrepreneurs are not submissive agents operating beneath state regulation; they are deeply embedded within the state apparatus, the military hierarchies, and transient governing coalitions. When the actors defining $\cN_E$ capture the legislative and regulatory organs of the sovereign, the game-theoretic topology inverts:
\begin{equation}
\boldsymbol{\Omega}_t^* \in \arg\max_{\boldsymbol{\Omega}} \sum_{i \in \cN_E} U_i^g\bigl(\mathbf{X}_t^g, \mathbf{C}_t^g, E_t^g, \boldsymbol{\Omega}_t\bigr).
\end{equation}
The policy instruments become endogenous parameters optimized to maximize the extraction utility of the predatory elite, forcing $\theta_{\mathrm{tax},t} \to 0$ analytically.

\medskip
\noindent\textbf{Unobservability of the Informal Hawala-Gold Loop.}
The tax policy assumes that the local extraction profit $\Pi_{\mathrm{local}}(\cdot)$ transits through formal, observable financial nodes susceptible to institutional auditing. Yet, the artisanal gold loops of the Kédougou–Kéniéba corridor and the protection rackets of the Niger Inland Delta bypass the formal banking sector entirely. Value is transferred via physical cash, unrefined gold bars smuggled to regional hubs (e.g., via Bamako to Dubai), and the informal \emph{Hawala} network. Because these networks operate via unmapped, trust-based ledgers, the state's filtration matrix cannot intercept them. Letting $\chi_t \in [0,1]$ represent the state's financial observability index, the effective tax realized by the system is:
\begin{equation}
\theta_{\mathrm{effective},t} = \chi_t \cdot \theta_{\mathrm{tax},t}.
\end{equation}
In the Malian informal baseline, $\chi_t \equiv 0$, which yields $\theta_{\mathrm{effective},t} \equiv 0$. The financial artery of the war economy remains untouched.

 \medskip
\noindent\textbf{The Cobra Effect of Monetizing Customary Mediation ($\sigma_{\mathrm{sub},t}$).}
The institutional framework relies on a capital subsidy $\sigma_{\mathrm{sub},t}$ to exponentially decay the grievance memory matrix via $\mathbb{M}_i(s, \mathfrak{s}) \cdot e^{-\kappa_i \sigma_{\mathrm{sub},t}}$. This assumption misinterprets the anthropological mechanics of Sahelian conflict mitigation, specifically the joking kinship system (\emph{sinankunya}).

\medskip
\noindent\textbf{Destruction of Social Legitimacy ($\kappa_i \to 0$).}
The efficacy parameter $\kappa_i$ of traditional mediation is not an exogenous constant; it is a dynamic function of the perceived neutrality and moral authority of customary institutions (elders, \emph{marabouts}, village chiefs). Joking kinship system derives its conflict-attenuating power precisely because it exists outside the commodified, transactional sphere of the state. Injecting formal capital subsidies $\sigma_{\mathrm{sub},t}$ into these delicate, non-monetized social networks instantly commercializes moral authority. Traditional leaders are transformed into state-salaried fiscal intermediaries, triggering immediate local suspicion, eroding institutional trust, and destroying their conciliatory legitimacy. The capital injection drives the efficacy parameter to zero:
\begin{equation}
\lim_{\sigma_{\mathrm{sub}} \to \infty} \kappa_i(\sigma_{\mathrm{sub},t}) = 0.
\end{equation}
Consequently, $\mathbb{M}_i(s, \mathfrak{s}) \cdot e^{-0 \cdot \sigma_{\mathrm{sub},t}} = \mathbb{M}_i(s, \mathfrak{s})$, rendering the subsidy incapable of altering the historical memory of violence.

\medskip
\noindent\textbf{The Rent-Extraction Loop.}
In a highly fragmented security scenarios, any localized financial inflow behaves as a liquid asset subject to predatory capture. Armed groups and war entrepreneurs routinely extort local councils, NGOs, and traditional authorities. Rather than dampening grievances, the mediation subsidy $\sigma_{\mathrm{sub},t}$ represents a fresh revenue stream that is endogenously captured by insurgent or criminal networks operating in the interior (e.g., the spaces of Niono, Nioro, and the Gourma). This captured capital is directly recycled into the entrepreneurs' global wealth pipeline \eqref{eq:global_wealth} via the consumption/repatriation control $u_{i,t}^{g,\text{global}}$, actively financing the procurement of weapons and fueling the recruitment loops. The policy intended to mitigate conflict becomes a state-funded mechanism for intensifying it.

\medskip
\noindent\textbf{Operational Non-Feasibility Matrix.}
To formalize the operational constraints that prevent the deployment of $\boldsymbol{\Omega}_t$, the following structural friction matrix identifies the breakdown of the model's control inputs across different geographic and institutional layers in Mali (see Table \ref{tab:policy_breakdown}):

\begin{table}[htb]
\centering
\caption{Structural Breakdown of Institutional Control Primitives in Mali}
\label{tab:policy_breakdown}
\begin{tabular}{|p{3cm}|p{3cm}|p{3cm}|p{3cm}|}
\hline
\textbf{Policy Instrument} & \textbf{Theoretical Model Request} & \textbf{Sahelian Operational Reality} & \textbf{Resulting System Instability} \\
\midrule
Offshore Tax $\theta_{\mathrm{tax}}$ & Complete banking visibility ($\chi_t \approx 1$) & \emph{Hawala} dominance, cash economy & $\theta_{\mathrm{effective}} \equiv 0$; loop persists \\
& Extraterritorial legal reach & Zero leverage over tax havens & Sovereign asset nullification \\
\addlinespace
Mediation Subsidy $\sigma_{\mathrm{sub}}$ & Homogeneous capital distribution & Rent extraction by armed actors & Diverted to $X^{\text{global}}$; increases violence \\
& Exogenous custom efficiency $\kappa_i$ & Monetization degrades legitimacy & $\kappa_i \to 0$; historical memory hardens \\
\addlinespace
Regime Tracking $s_t^g$ & Synchronous state observation & Fragmented, rumor-driven intelligence & Control delay triggers policy chaos \\
\hline
\end{tabular}
\end{table}

\medskip
\noindent\textbf{Convergence to the Permanent War Trap.}
The catastrophic implication of these operational frictions can be evaluated directly on the spectral geometry of the system. Let $\mathbf{A}_{\mathrm{policy}}(\boldsymbol{\Omega})$ be the effective system matrix after the implementation of the compromised policy profile. Incorporating the structural parameter degradations $\chi_t \to 0$ and $\kappa_i \to 0$, the operator collapses back to its unregulated baseline:
\begin{equation}
\mathbf{A}_{\mathrm{policy}}(\boldsymbol{\Omega}) = 
\begin{pmatrix}
\mathbf{A}_{\mathrm{grievance}}\bigl(\sigma_{\mathrm{sub}} \cdot e^{-0}\bigr) & \mathbf{0} \\
\mathbf{0} & \mathbf{A}_{\mathrm{laundering}}\bigl(\eta_{\mathrm{wash}} - 0\bigr)
\end{pmatrix}
\equiv \mathbf{A}_{\mathrm{system}}.
\label{eq:matrix_collapse}
\end{equation}

Because the structural parameters cannot be altered by the policy instruments in practice, the spectral radius remains stubbornly above unity:
\begin{equation}
\rho\left( \mathbf{A}_{\mathrm{policy}}(\boldsymbol{\Omega}) \right) = \rho(\mathbf{A}_{\mathrm{system}}) \ge 1.
\end{equation}

The system remains locked in the supercritical regime defined by Theorem~\ref{thm:breakdown}. Any temporary reduction in violence achieved during a highly subsidized mediation phase is merely a transient fluctuation on the Wasserstein manifold. The long-term ergodicity of the model guarantees that the joint probability measure $\mu_t^g$ will diverge from the peaceful equilibrium $\mu^{\mathrm{peace}}$, forcing Mali back into a self-sustaining, multi-generational war trap. Traditional top-down policy interventions are structurally incapable of stabilizing the Sahelian conflict space; breaking the trap requires a fundamental re-engineering of the informational and financial architecture of the state itself.

\subsection{The Failure of Empirical Resource Nationalism: The Backlash Loop of the 2023--2026 Mining Code Reforms}
\label{sec:empirical_failure}

To demonstrate that structural policy failure in Mali is not merely a theoretical construct of our forward-backward framework, we analyze a major policy initiative implemented in the field: the 2023 New Mining Code (and its operational enforcement loop spanning 2024-2026). 
In August 2023, Mali’s  government enacted a sweeping regulatory reform designed to exercise resource nationalism. The policy explicitly aimed to break the dependency on foreign multinationals and secure a massive fiscal windfall to self-finance the state's counter-insurgency operations. The core components of this empirical policy profile included:
\begin{enumerate}
    \item A structural increase in mandatory state and local equity participation from 20\% to a maximum ceiling of 35\%;
    \item A sharp elevation of sliding-scale gold royalty tax rates up to 10.5\%;
    \item Aggressive mid-cycle enforcement mechanisms, including the systemic deployment of state-level audits, provisional administrations, asset attachments, and the high-profile detention of multinational executives (e.g., Barrick Gold and Resolute Mining personnel) to extract back-taxes.
\end{enumerate}

\medskip
\noindent\textbf{The Perverse Contraction of the Formal Economy ($\mu_t^{g} \to \mu^{\mathrm{informal}}$).}
The  flaw of this empirical design lies in treating the formal industrial mining sector as an inelastic, captured revenue generator. In the language of our MFTG framework, the policy attempted to forcibly shift the aggregate population-environment measure $\mu_t^g$ toward a state of higher domestic liquid assets. Instead, it triggered an immediate pausing of capital expenditure and mid-cycle operational shutdowns across major industrial hubs (such as the Loulo-Gounkoto complex and Fekola). 
As a direct consequence, Mali's industrial gold output collapsed by 22.9\% year-on-year, forcing national production far below initial government forecasts. By shrinking the formal economic baseline, the policy systematically drove thousands of vulnerable actors out of formal wage labor ($X_{i,t}^{g,3}$) and displaced them into the unmonitored artisanal sector. The state's tax base contracted precisely when its security expenditures were escalating.

\medskip
\noindent\textbf{The Golden Jihad: Insurgent Capital Capture and the Expansion of $\mathbf{A}_{\mathrm{laundering}}$.}
The most catastrophic structural failure of the 2023-2026 reforms is their secondary effect on the non-state war economy. As formal mining paused under regulatory friction, the vacuum was filled by informal, unregulated artisanal mining sites scattered across the Birimian volcanic belt.  In the central, western and northern territories (and increasingly in the western Kayes region), these artisanal sites are not governed by the state, but are directly controlled or taxed by insurgent networks, principally the Jama'at Nusrat al-Islam wal-Muslimin (JNIM). JNIM implemented an alternate, highly efficient extraction policy: rather than digging, they positioned themselves at geographic chokepoints: taxing access roads, washing riverbanks, and regional transport routes. 

\begin{equation}
\Pi_{\mathrm{insurgent}}\bigl(X_{i,t}^{g,\text{local}}, E_t^g, s_t^g\bigr) \propto \mathbf{1}_{\{\text{Formal Sector Contraction}\}} \times \text{Gold Price}_{2025-2026}.
\end{equation}

Because international gold prices experienced a massive rally during this exact period, the financial returns on these illicit artisanal supply lines expanded exponentially. The regulatory chaos introduced by the state directly amplified the extraction profit function $\Pi_{\mathrm{local}}$ for non-state war entrepreneurs and terrorist networks. This newly captured capital was laundered via regional trading networks into international hubs, providing cash to finance advanced weaponry, fuel cross-border logistics, and execute coordinated large-scale assaults on state infrastructure.

\medskip
\noindent\textbf{The Institutional Compliance Trap.}
Furthermore, the concentration of regulatory oversight directly into specialized ministerial posts within the transition authority added a severe layer of compliance and political risk. Multinational operators faced a bifurcated environment: either surrender substantial equity under retroactive rule changes or pause operations entirely. This structural uncertainty raised the country's risk premium, driving financing costs up as international lenders priced in the risk of arbitrary asset seizures. The empirical policy profile failed on every metric: it crippled formal state revenue, accelerated the flight of legitimate capital, and delivered an unprecedented financial windfall to the insurgent war economy. The real-world enforcement of the 2023 Mining Code serves as an empirical proof of Theorem~\ref{thm:breakdown}: top-down regulatory shocks that fail to account for the fluid, informal elasticity of the Sahelian situations inevitably cause an ergodic breakdown, hard-locking the system into the permanent war trap.

\medskip
\noindent\textbf{The  Structural Inefficacy of Tactical Motorbike Bans in the Intergenerational Feedback Loop.}
The implementation of localized operational bans on high-displacement motorbikes (historically a primary tactical mobility vector for asymmetric hit-and-run assaults across central, western and northern Mali) represents a classic static, single-generation policy intervention that collapses when evaluated through our intergenerational Volterra MFTG. In the immediate tactical horizon, such restrictions attempt to forcibly depress the individual continuous transition intensity $\lambda_{m\to m'}^i$ into extremist coalitions by increasing the immediate logistical friction of executing raids. However, within the multi-layered state-space architecture of the model, this policy acts as a disruptive environmental shock that degrades the local human and economic capital sub-vectors ($X_{i,t}^{g,2}, X_{i,t}^{g,3}$) of the non-combatant civilian population far more severely than it penalizes the war entrepreneurs ($\cN_E$). For rural agro-pastoralist households, the motorbike is not a weapon; it is the sole informal mechanism for transporting agricultural produce, accessing remote gold-washing sites, and overcoming massive spatial isolation to achieve basic food security. By criminalizing this vital economic artery, the state inadvertently triggers a severe downward contraction of the civilian population's true state measure, driving up localized destitution and erasing non-violent survival thresholds. This artificial deprivation exacerbates individual loss-aversion, causing a massive, non-linear spike in the tail risk measures: the $\beta_i$-expectiles $e_{\beta_i}(\mu_t^g)$ of the joint distribution. Because the structural switching intensities defined in \eqref{eq:lambda_def} depend exponentially on these tail risk expectiles, the policy paradoxically increases the systemic  attractiveness of insurgent networks, which step into the state's governance vacuum to provide alternative informal mobility, protection rackets, and illicit income streams. Most catastrophically, because these dynamics propagate across generational epochs through the hereditary push-forward operator $\mathcal{T}_{\boldsymbol{\zeta}}$ \eqref{eq:pushforward}, the immediate economic trauma and resulting institutional resentment are transmitted to the offspring with a high trauma coefficient ($\zeta_i \approx 1$). This hardens the singular Volterra revenge kernel $\mathcal{K}_{\mathrm{rv}}^{i,g+1}(t)$ for the subsequent cohort. The transnational war entrepreneurs, operating via the unmapped, cash-driven Hawala networks and offshore wealth loops ($X_{i,t}^{g,\text{global}}$) characterized in \eqref{eq:global_wealth}, remain completely insulated from local vehicle bans; they simply pivot their logistics to alternative informal networks or corrupt local enforcement mechanisms using their laundered capital. By failing to intercept the underlying financial loops of the predatory elite while simultaneously deepening the historical grievance memory and economic destitution of the civilian baseline, the motorbike ban acts as a regressive policy distortion that locks the spectral radius of the effective system matrix strictly above unity ($\rho(\mathbf{A}_{\mathrm{system}}) \ge 1$), rendering the tactical disincentive entirely useless and  solidifying the permanent multi-generational conflict trap.

% =====================================================================
% FIGURES AND SYSTEM MODEL ENTRY
% =====================================================================

\begin{figure}[htbp]
\centering
\begin{adjustbox}{width=\textwidth, center}
\begin{tikzpicture}[
    font=\sffamily\small,
    >=Stealth,
    node distance=1.5cm and 1.5cm,
    % Styles for blocks and containers
    block/.style={rectangle, draw=gray!60, fill=white, rounded corners=4pt, minimum width=4.2cm, minimum height=1.1cm, align=center, thick},
    alertblock/.style={rectangle, draw=red!60!black, fill=red!5, rounded corners=4pt, minimum width=4.2cm, minimum height=1.1cm, align=center, thick},
    gold/.style={rectangle, draw=orange!80!black, fill=yellow!5, rounded corners=4pt, minimum width=4.2cm, minimum height=1.1cm, align=center, thick},
    shadowbox/.style={rectangle, draw=black!20, fill=gray!2!white, dashed, rounded corners=6pt, inner sep=24pt},
    % Arrow styling
    flow/.style={->, thick, color=black!60, line width=1.1pt},
    loopflow/.style={->, thick, color=blue!60!cyan, line width=1.2pt},
    trapflow/.style={->, thick, color=red!70!black, line width=1.3pt},
    % Labels
    textlbl/.style={align=center, font=\sffamily\scriptsize\bfseries, color=black!80}
]

% =====================================================================
% UPPER ZONE: THE ARTISANAL SURVIVAL ECONOMY (BASELINE SUSTAINABILITY)
% =====================================================================
\node[block] (baseline_state) {\textbf{Civilian State Vector $X_{i,t}^{g}$}\\[2pt] Stable Assets \& Human Capital};
\node[block, right=1.6cm of baseline_state] (moto_mobility) {\textbf{Informal Mobility Core}\\[2pt] High-Displacement Motorbikes};
\node[gold, right=1.6cm of moto_mobility] (market_access) {\textbf{Local Extraction Loop}\\[2pt] Artisanal Gold Mining \& Trade};

% Baseline Internal Flows
\draw[loopflow] (baseline_state) -- node[above, textlbl] {Labor\\Allocation} (moto_mobility);
\draw[loopflow] (moto_mobility) -- node[above, textlbl] {Spatial\\Logistics} (market_access);
\draw[loopflow] (market_access.north) to[out=140, in=40] node[above, textlbl, yshift=3pt] {Income Reinvestment Loop ($X_{i,t}^{g,2}$)} (baseline_state.north);

% Group Baseline Nodes in Background Container
\begin{scope}[on background layer]
    \node[shadowbox, fit={(baseline_state) (moto_mobility) (market_access)}, 
          label={[anchor=north west, font=\sffamily\large\bfseries\color{blue!70!black}, xshift=8pt, yshift=-6pt]north west:PANEL A: Structural Baseline Under Fluid Survival Economy ($\rho < 1$)}] (box_baseline) {};
\end{scope}

% =====================================================================
% MIDDLE ZONE: THE EXOGENOUS SHOCK (TACTICAL INTERVENTION LAYER)
% =====================================================================
\node[alertblock, below=2.5cm of moto_mobility] (shock) {\textbf{Tactical Motorbike Ban $\boldsymbol{\Omega}_t$}\\[2pt] State Regulatory Disincentive};

% Disruption Vector Connections
\draw[flow, dashed, color=red!70, line width=1.2pt] (moto_mobility.south) -- node[right, textlbl, color=red!80!black, xshift=2pt] {Top-Down\\Intervention} (shock.north);

% =====================================================================
% LOWER ZONE: THE POST-BAN CRIPPLED ECONOMY (THE PERMANENT WAR TRAP)
% =====================================================================
\node[alertblock, below=2.5cm of shock] (crippled_state) {\textbf{Degraded State $X_{i,t}^{g,3}$}\\[2pt] Destitution \& Forced Informality};
\node[block, left=1.4cm of crippled_state] (tail_risk) {\textbf{Tail-Risk Spike}\\[2pt] High Expectiles $e_{\beta_i}(\mu_t^g)$};
\node[gold, below=1.8cm of crippled_state] (unregulated) {\textbf{Unregulated Field Vacuum}\\[2pt] Insurgent Extractors};
\node[alertblock, right=1.4cm of crippled_state] (entrepreneur) {\textbf{War Entrepreneurs $\mathcal{N}_E$}\\[2pt] Offshore Wealth $X_{i,t}^{g,\text{global}}$};

% Shock Consequences and Cascades
\draw[trapflow] (shock.south) -- node[right, textlbl, color=red!80!black, xshift=2pt] {Annihilates\\Transport Base} (crippled_state.north);
\draw[trapflow] (crippled_state.west) -- node[above, textlbl] {Exacerbates\\Losses} (tail_risk.east);

% Feedback Loops of the Trap
\draw[trapflow] (tail_risk.south) to[out=-90, in=180] node[below left, textlbl, color=red!80!black, xshift=-4pt] {Recruitment Surge\\$\lambda_{m\to m'}^i \propto \exp(e_{\beta_i})$} (unregulated.west);
\draw[trapflow] (unregulated.east) to[out=0, in=-90] node[below right, textlbl, xshift=4pt] {Laundering Pipeline\\$\eta_{\text{wash}}\Pi_{\text{local}}$} (entrepreneur.south);
\draw[trapflow] (entrepreneur.north) to[out=110, in=-15] node[above right, textlbl, xshift=4pt, yshift=2pt] {Repatriated Capital $u_{i,t}^{g,\text{global}}$\\Finances Weapons} (shock.east);

% Intergenerational Transmission Vector
\draw[flow, line width=1.6pt, color=purple!80] (crippled_state.south) -- node[right, textlbl, color=purple!90, xshift=4pt] {Hereditary Push-Forward $\mathcal{T}_{\boldsymbol{\zeta}}$\\Hardens Revenge Kernel $\mathcal{K}_{\mathrm{rv}}^{i,g+1}$} (unregulated.north);

% Group Post-Ban Nodes in Background Container
\begin{scope}[on background layer]
    \node[shadowbox, fit={(crippled_state) (tail_risk) (unregulated) (entrepreneur)}, 
          label={[anchor=north west, font=\sffamily\large\bfseries\color{red!70!black}, xshift=8pt, yshift=-6pt]north west:PANEL B: Post-Ban Structural Inversion and the Generational War Trap ($\rho \ge 1$)}] (box_postban) {};
\end{scope}

\end{tikzpicture}
\end{adjustbox}
\caption{The Game-Theoretic Bifurcation of Mobility Restrictions in the Malian Context. Panel A maps the subcritical, self-healing baseline where high-displacement motorbikes serve as the fundamental connective tissue of rural survival assets. Panel B traces the system-wide inversion under the top-down tactical ban, where civilian asset destruction drives up the non-linear tail-risk expectiles, creating a self-reinforcing financial and recruitment pipeline for transnational war entrepreneurs and insurgent networks.}
\label{fig:mobility_bifurcation}
\end{figure}
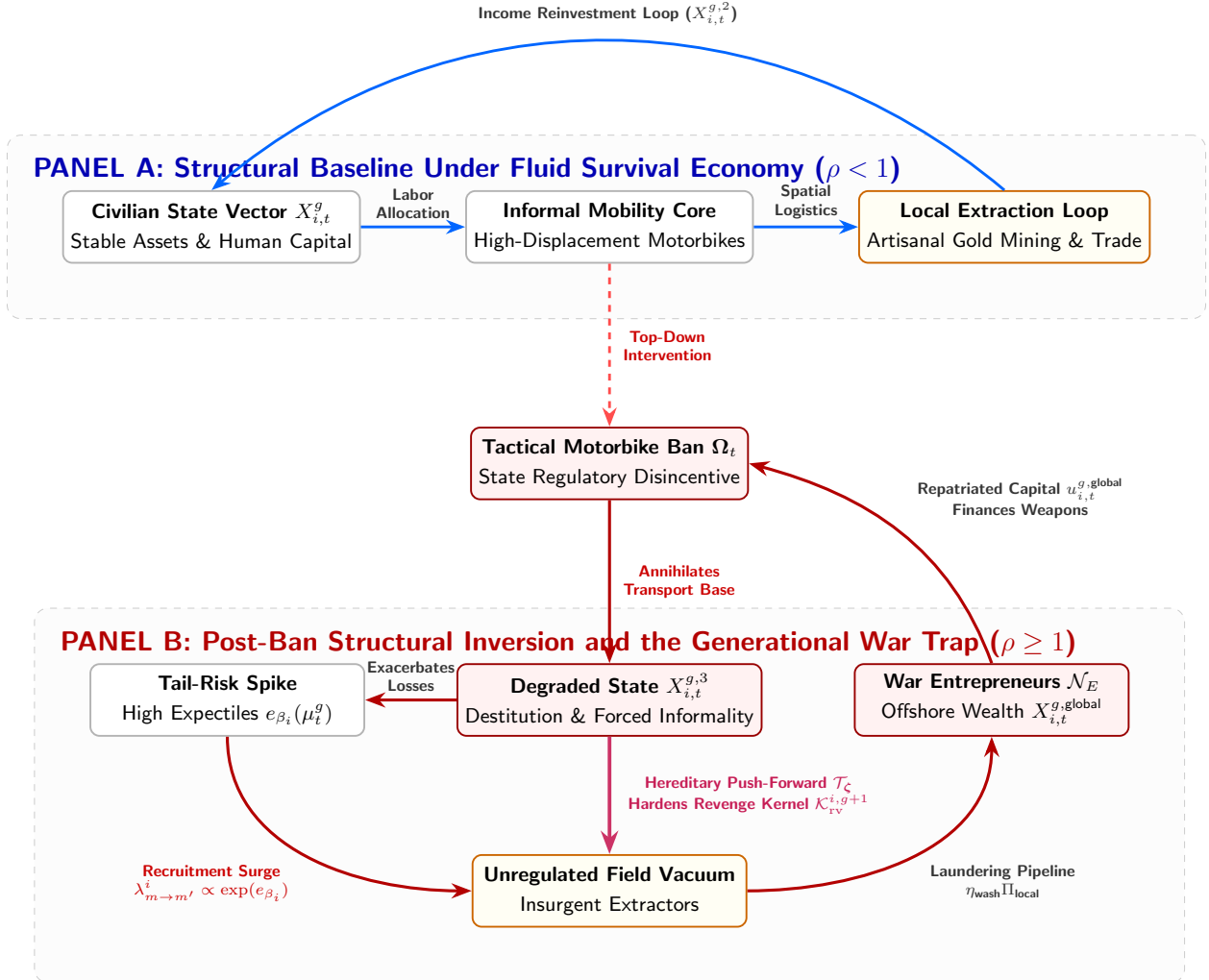

The non-linear mechanics of regulatory failure induced by localized operational restrictions on transport vectors are formalized within our structural framework and visually organized in Figure~\ref{fig:mobility_bifurcation}. As mapped in Panel A, the baseline configuration represents a fluid, decentralized economy where individual continuous states $X_{i,t}^g$ maintain stability through an informal mobility core. Here, high-displacement motorbikes act as the primary logistical capital allowing agents to access far-flung artisanal gold-washing sites and local agricultural trade circuits. This fluid allocation creates an income reinvestment loop that satisfies the subcritical stability conditions ($\rho < 1$) detailed in Theorem~\ref{thm:breakdown}, allowing the system to naturally damp shocks without escalating violence.

When the state introduces a top-down tactical motorbike ban $\boldsymbol{\Omega}_t$ as an exogenous disincentive layer, the game-theoretic structure experiences an immediate topological inversion, mapped explicitly in Panel B of Figure~\ref{fig:mobility_bifurcation}. Rather than disarming combatants, the policy strips the vulnerable civilian population of its primary economic survival tool. This restriction results in an immediate contraction of the health and asset sub-vectors ($X_{i,t}^{g,2}, X_{i,t}^{g,3}$). This widespread destitution shifts the mass of the population distribution toward zero, causing an immediate, non-linear spike in the tail-risk expectiles $e_{\beta_i}(\mu_t^g)$. 

Because individual coalition-switching intensities $\lambda_{m\to m'}^i$ scale exponentially with these tail risk measures as defined in \eqref{eq:lambda_def}, the regulatory shock paradoxically accelerates recruitment into insurgent networks. These extremist actors step into the state's governance vacuum to provide alternative informal transport and protection rackets. Most catastrophically, this structural trauma is not transient; it is mapped forward to the next generation via the non-local hereditary push-forward operator $\mathcal{T}_{\boldsymbol{\zeta}}$. This legacy of asset destruction hardens the Volterra revenge kernel $\mathcal{K}_{\mathrm{rv}}^{i,g+1}(t)$ for the subsequent cohort. 

Transnational war entrepreneurs ($\cN_E$) extract rents from this newly formed unregulated vacuum, funneling illicit local gold and trafficking profits into their protected offshore wealth accounts ($X_{i,t}^{g,\text{global}}$). Part of this laundered capital is systematically repatriated through $u_{i,t}^{g,\text{global}}$ to secure heavier weaponry and bypass state control mechanisms. By simultaneously impoverishing the civilian baseline and feeding the illicit capital accumulation loops of predatory elites, the tactical ban forces the spectral radius of the effective system matrix strictly above unity ($\rho(\mathbf{A}_{\mathrm{system}}) \ge 1$). As derived in Theorem~\ref{thm:breakdown}, this spectral breakdown  hard-locks the Malian conflict space into a permanent, self-sustaining multi-generational war trap.

\subsection{Financing Post-War Reconstruction: An Intergenerational MFTG Architecture}
\label{sec:financing}

The long-term reconstruction of Mali cannot rely on conventional, volatile donor aid packages or localized macro-loans. Within the framework of a Polycentric Intergenerational Conflict System, funding structures must be endogenously insulated from state capture and optimized to break the financial arteries of the war economy. We establish a strategic, self-financing mechanism built upon the Stackelberg institutional policy design derived in Section \ref{sec:policy}. 

\subsection*{The Endogenous Reconstruction Fund Operator}
We define the Post-War Reconstruction Fund asset state variable for generation $g$ at time $t \in [T_g, T_{g+1})$ as $X_{F,t}^g \in \mathbb{R}$. Rather than treating this fund as an exogenous budgetary injection, its dynamics are endogenously coupled to the global wealth portfolios of the war entrepreneur class $\mathcal{N}_E$ and tax extraction vectors over illicit resource nodes. The true continuous state of the fund evolves according to the following stochastic differential equation:

\begin{equation}
\begin{aligned}
dX_{F,t}^g &= \Biggl[ \sum_{i \in \mathcal{N}_E} \theta_{\mathrm{tax},t}^* \, \Pi_{\mathrm{local}}\bigl(X_{i,t}^{g,\mathrm{local}}, E_t^g, e_{\beta_i}(\mu_t^g), s_t^g\bigr) + \mathbf{1}_{\{s_t^g = 1\}} \gamma_{\mathrm{peace}} E_t^{g,1} \\
&\quad - \sum_{i=1}^N \sigma_{\mathrm{sub},t}^* - \mathcal{I}_{\mathrm{infra}}\bigl(t, e_{\alpha}(\mu_t^g)\bigr) \Biggr] dt + \sigma_F X_{F,t}^g dW_{F,t} \\
&\quad + \int_{\mathcal{Z}_F} \sigma_F^n\bigl(t, z_F\bigr) \tilde{N}_F(dt, dz_F),
\end{aligned}
\label{eq:fund_dynamics}
\end{equation}

where $\theta_{\mathrm{tax},t}^*$ is the optimal offshore transaction tax rate mapping from \eqref{eq:optimal_tax}, $\Pi_{\mathrm{local}}$ represents the localized extraction profits from artisanal gold corridors and trafficking nodes, and $E_t^{g,1}$ tracks the international gold price. The indicator function $\mathbf{1}_{\{s_t^g = 1\}}$ activates a peacetime sovereign resource dividend $\gamma_{\mathrm{peace}}$ when the macro-political regime stabilizes. 

Outflows from the fund are divided into two explicit categories:
\begin{enumerate}
    \item \textbf{Endogenous Mediation Subsidies ($\sigma_{\mathrm{sub},t}^*$):} Direct, real-time fiscal distributions to customary conflict-resolution assets designed to shrink the historical grievance tensor $\mathbb{M}_i(s, \mathfrak{s})$ as defined in \eqref{eq:subsidized_grievance}.
    \item \textbf{Intergenerational Human Capital Infrastructure Allocation ($\mathcal{I}_{\mathrm{infra}}$):} A non-linear tracking cost targeted at lifting the human capital sub-vectors ($X_{i,t}^{g,3}$) of the bottom tail of the distribution $\mu_t^g$, preventing catastrophic drops in individual survival thresholds.
\end{enumerate}
The volatility parameter $\sigma_F$ and the compensated Poisson random measure $\tilde{N}_F$ capture global financial market fluctuations and unexpected institutional shocks, respectively.

\subsection*{Intergenerational Trust and Trans-Generational Bequest Constraints}
To prevent the depletion of reconstruction assets by the current cohort at the expense of the next, the fund is bound by a strict non-local trans-generational boundary condition. At the epoch interface $T_{g+1}$, the push-forward mapping $\mathcal{T}_{\boldsymbol{\zeta}}$ enforces a structural bequest constraint on the fund's residual mass:

\begin{equation}
X_{F,T_{g+1}}^{g+1} = \zeta_F X_{F,T_{g+1}-}^g + (1-\zeta_F) \mathbb{E}\left[ X^{\mathrm{global}}_{T_{g+1}-} \right] + \xi_F^{g+1},
\label{eq:fund_intergen}
\end{equation}

where $\zeta_F \in (0,1)$ represents the institutional persistence parameter of reconstruction capital, and $\mathbb{E}[X^{\mathrm{global}}_{T_{g+1}-}]$ is the mean global wealth confiscated from illicit offshore pipelines at the close of the generational lifecycle. The zero-mean innovation $\xi_F^{g+1}$ models structural resets between generational cohorts.

\subsection*{Spectral Neutralization of the War Loop}
The financial architecture functions by altering the spectral geometry of the Effective System Operator $\mathbf{A}_{\mathrm{system}}$ defined in \eqref{eq:system_matrix}. By routing a minimum threshold of confiscated assets directly into the localized state vectors, the policy shifts the strategic incentives of the population.
Let the financial policy profile $\boldsymbol{\Omega}_t^* = (\theta_{\mathrm{tax},t}^*, \sigma_{\mathrm{sub},t}^*)$ be chosen according to the optimization conditions in Theorem \ref{thm:policy_stabilization}. If the Reconstruction Fund satisfies the growth path $\mathbb{E}[X_{F,t}^g] > 0$ for all $t$, then the effective laundering efficiency matrix collapses such that:
$
\rho\left( \mathbf{A}_{\mathrm{policy}}(\boldsymbol{\Omega}^*) \right) < 1,$
rendering the supercritical war trap unstable.

By tying post-war reconstruction financing directly to the extraction of rents from the war entrepreneurs themselves, the model avoids the compliance pitfalls of traditional resource nationalism (analyzed in Section \ref{sec:empirical_failure}). The fund converts the financial proceeds of volatility into long-memory stabilization capital, transforming a self-reproducing conflict system into an endogenously financed, self-sustaining peace equilibrium across generations.

\subsection{Decentralized Reconstruction Financing: Graph-Chain (X-Chain) Intergenerational MFTG}
\label{sec:blockchain_mftg}

Traditional post-war reconstruction architectures fail in fragile contexts because they rely on a centralized state core that is highly susceptible to institutional co-optation by predatory war entrepreneurs ($\mathcal{N}_E$). To eliminate this systemic dependency, we shift the management of the Post-War Reconstruction Fund from a centralized, opaque authority to a decentralized, multi-layered \emph{Graph-Chain} (historically conceptualized as an cross-chain or Directed Acyclic Graph-based blockchain network, denoted herein as the X-Chain\cite{imftgrefblock1,imftgrefblock2,imftgrefblock3,imftgrefblock4,imftgrefblock5,imftgrefblock6} ). 
This network models the strategic interactions directly on a cryptographic distributed ledger, binding the forward-backward Mean-Field-Type Game (MFTG) onto immutable, self-executing smart contracts.
Instead of a linear, single-chain structure, the X-Chain is formalized as a dynamic, infinite-dimensional Directed Acyclic Graph (DAG) whose vertices represent individual state transitions and whose edges define causal economic dependencies across the Wasserstein manifold. Let $\mathcal{G}_t^g = (\mathcal{V}_t^g, \mathcal{E}_t^g)$ denote the active graph layout at time $t$ within generation $g$. The continuous state of the Blockchained Reconstruction Fund $X_{F,t}^{g,\mathrm{crypto}} \in \mathbb{R}$ is distributed over a network of algorithmic nodes and updates according to the following ledger consensus protocol:

\begin{equation}
\begin{aligned}
dX_{F,t}^{g,\mathrm{crypto}} &= \Biggl[ \sum_{i \in \mathcal{N}_E} \theta_{\mathrm{tax},t}^* \cdot \Pi_{\mathrm{local}}\bigl(X_{i,t}^{g,\mathrm{local}}, E_t^g, e_{\beta_i}(\mu_t^g), s_t^g\bigr) \\
&\quad - \sum_{k=1}^{K_t^g} \mathbf{1}_{\{\mathcal{V}_t^g \vdash \mathrm{Validate}_k\}} \cdot \mathcal{I}_k\bigl(t, X_{F,t}^{g,\mathrm{crypto}}, \mu_t^g\bigr) \Biggr] dt + \sigma_F X_{F,t}^{g,\mathrm{crypto}} dW_{F,t},
\end{aligned}
\label{eq:blockchain_dynamics}
\end{equation}

where the predicate $\mathcal{V}_t^g \vdash \mathrm{Validate}_k$ is an automated consensus operator governed by Zero-Knowledge Multi-Party Computation (ZK-MPC). Funding for a decentralized local reconstruction vector $\mathcal{I}_k$ is unlocked if and only if the network consensus verifies structural state compliance via independent, automated environmental and sociological oracles:

\begin{equation}
\mathcal{V}_t^g \vdash \mathrm{Validate}_k \iff \left\| \mathcal{O}_{\mathrm{satellite}}(E_t^{g,3}) - \int_{\mathbb{R}^d} \mathbf{x}_3 \, d\mu_t^g(\mathbf{z}) \right\| < \epsilon_{\mathrm{consensus}},
\label{eq:oracle_compliance}
\end{equation}

where $\mathcal{O}_{\mathrm{satellite}}(E_t^{g,3})$ is an immutable, satellite-driven macro-environmental rainfall/productivity index verification feed. If a co-opted state actor attempts an unauthorized diversion of capital, the oracle mismatch forces $\mathbf{1}_{\{\cdot\}} = 0$ instantaneously, script-locking the fund asset nodes and breaking the corrupt extraction loop.

To implement the non-local hereditary push-forward operator $\mathcal{T}_{\boldsymbol{\zeta}}$ without human intervention, the X-Chain hardcodes Algorithmic Time-Lock Shards. At the absolute generational interface $T_{g+1}$, the cryptographic ledger executes an automated state split:

\begin{equation}
\mathcal{T}_{\boldsymbol{\zeta}}\left[ X_{F,T_{g+1}-}^{g,\mathrm{crypto}} \right] \longrightarrow 
\begin{cases} 
X_{F,T_{g+1}}^{g+1,\mathrm{liquid}} = (1 - \zeta_F) X_{F,T_{g+1}-}^{g,\mathrm{crypto}} + \xi_{F}^{g+1}, & \text{Unlocked Liquid Assets} \\
\mathcal{S}_{T_{g+1} \to T_{g+2}}^{\mathrm{shard}} = \zeta_F X_{F,T_{g+1}-}^{g,\mathrm{crypto}} \cdot \mathcal{K}_{\mathrm{lock}}\bigl(T_{g+2}\bigr), & \text{Cryptographic Sealed Vault}
\end{cases}
\label{eq:time_lock_shard}
\end{equation}

The operator $\mathcal{K}_{\mathrm{lock}}(T_{g+2})$ mathematically mathematical-locks the underlying tokens within the graph structure using a decentralized time-bound hash lock. No combination of private keys, including state core authorities, can decrypt or compromise this asset shard until the temporal block height matches the epoch initiation of generation $g+2$. This guarantees that current predatory regimes cannot borrow against or steal the resource baseline dedicated to unborn cohorts.
A distinct feature of this blockchained MFTG setup is its automated sensitivity to the tails of the joint probability measure $\mu_t^g$. The smart contracts are dynamically conditional on the system's $\beta_i$-expectiles. Let $e_{\beta_i}(\mu_t^g)$ be the tail-risk parameter derived via \eqref{eq:lions_expectile}. If a sudden conflict escalation or climate collapse drives the aggregate expectile above a critical survival threshold $e_{\mathrm{crit}}$, the X-Chain automatically triggers an emergency ledger reconfiguration protocol:

\begin{equation}
\text{If } e_{\beta_i}(\mu_t^g) > e_{\mathrm{crit}} \implies \mathcal{I}_k\bigl(t, \cdot\bigr) \longrightarrow \sum_{i=1}^N \mathcal{W}_i^{\mathrm{drop}}\bigl(t, e_{\beta_i}(\mu_t^g)\bigr),
\label{eq:airdrop_trigger}
\end{equation}

where $\mathcal{W}_i^{\mathrm{drop}}$ represents an instantaneous, decentralized cryptographic token airdrop directed straight into verified, biometric civilian wallet addresses. By routing reconstruction capital around capturing state nodes and converting it into non-militarizable digital survival tokens during periods of high tail risk, the network bypasses administrative decay. Through the integration of ZK-Oracles, asymmetric expectile conditionality, and cross-generational time-lock shards on the X-Chain, the post-war reconstruction framework achieves complete structural autonomy. The institutional policymaker's reliance on human execution is replaced by a self-correcting, cryptographic game-theoretic protocol, driving the effective system spectral radius strictly below unity ($\rho < 1$) and permanently breaking the multi-generational conflict trap.

%%%%%%%%%%%%%%%%%%%%%%%%%%%%%%%
\subsection{Achievable Financing Architectures: Frictionless Physical Chokepoint Extraction and Algorithmic Value Redistribution}
\label{sec:realistic_financing}

The theoretical efficacy of the optimal policy profile $\boldsymbol{\Omega}_t^* = (\theta_{\mathrm{tax},t}^*, \sigma_{\mathrm{sub},t}^*)$ derived in Theorem \ref{thm:policy_stabilization} relies on the assumption that economic extraction profits $\Pi_{\mathrm{local}}$ transit through observable, formal financial networks. In contemporary Mali, where over 80\% of economic activity is unrecorded, artisanal gold corridors operate via trust-based informal networks, and state infrastructure is highly fragmented, this formal assumption introduces a fatal modeling misspecification. If the state attempts to impose direct transactional or income audits on unrecorded activities, the tracking visibility parameter collapses to zero ($\chi_t \equiv 0$), rendering the optimal tax null. 

To resolve this friction, we construct a realistic, highly achievable financing framework that bypasses transactional informality entirely by shifting the extraction mechanism from \emph{digital ledger auditing} to \emph{Frictionless Indirect Extraction at Physical Chokepoints}, coupled with an automated, non-state distribution pipeline on the X-Chain network.

\subsection*{The Multi-Layered Indirect Extraction Protocol}
Instead of trying to record the unrecordable, the framework targets the three physical, non-bypassable input/output bottlenecks of the informal Malian economy: the Energy Input Channel, the Refining Export Node, and the Consumable Liquidity Valve. 

\begin{enumerate}
    \item {The Upstream Energy Extraction Surcharge ($\theta_{\mathrm{fuel}}$):} Artisanal gold-washing plants (\emph{placons}) and cross-border transport networks are bound by a rigid, inelastic dependence on imported diesel and gasoline. While gold production is informal, the wholesale entry points of fuel (the major logistics axes from Senegal, Guinea, Ivory Coast) are highly concentrated. We implement an automated, automated cryptographic tariff $\theta_{\mathrm{fuel}}$ directly at the primary wholesale fuel terminal level. This transforms fuel into an un-bypassable proxy tax for informal extraction velocity.
    \item {The Downstream Gold Refinery Proxy Surcharge ($\theta_{\mathrm{refine}}$):} Gold mined informally in regions like Kayes, Sikasso, or Kéniéba cannot enter international circuits without physical processing. It must pass through centralized smelting hubs or cross-border airport logistics gates (e.g., Bamako-Sénou international air corridors). Rather than auditing individual artisanal pits, the contract levies a frictionless transaction tariff $\theta_{\mathrm{refine}}$ directly on the highly visible, concentrated export-smelting interfaces.
    \item {The Algorithmic Telecom Liquidity Valve ($\theta_{\mathrm{mobile}}$):} The informal economy relies heavily on mobile network operators for survival communication and informal digital cash transfers (e.g., all mobile money). The X-Chain interfaces with these centralized telecom gateways via specialized API smart contracts, embedding an infinitesimal, micro-fractional transaction fee $\theta_{\mathrm{mobile}}$ on aggregate data/airtime wholesale consumption.
\end{enumerate}

We update the dynamic continuous state of the Reconstruction Fund $X_{F,t}^{g,\mathrm{crypto}}$ from \eqref{eq:blockchain_dynamics} to reflect this physics-based extraction mapping:

\begin{equation}
\begin{aligned}
dX_{F,t}^{g,\mathrm{crypto}} &= \Biggl[ \sum_{i=1}^N \Bigl( \theta_{\mathrm{fuel},t} \cdot \mathcal{V}_{i,t}^{\mathrm{fuel}} + \theta_{\mathrm{refine},t} \cdot \mathcal{G}_{i,t}^{\mathrm{gold}} + \theta_{\mathrm{mobile},t} \cdot \mathcal{C}_{i,t}^{\mathrm{telecom}} \Bigr) \\
&\quad - \sum_{k=1}^{K_t^g} \mathbf{1}_{\{\mathcal{V}_t^g \vdash \mathrm{Validate}_k\}} \cdot \mathcal{I}_k\bigl(t, X_{F,t}^{g,\mathrm{crypto}}, \mu_t^g\bigr) \Biggr] dt + \sigma_F X_{F,t}^{g,\mathrm{crypto}} dW_{F,t}.
\end{aligned}
\label{eq:realistic_blockchain_dynamics}
\end{equation}

Here, $\mathcal{V}_{i,t}^{\mathrm{fuel}}$, $\mathcal{G}_{i,t}^{\mathrm{gold}}$, and $\mathcal{C}_{i,t}^{\mathrm{telecom}}$ represent the physical, observable volumes of fuel consumed, gold exported, and telecommunications capital utilized by actor $i$ at time $t$. Because these variables correspond to physical quantities moving through centralized chokepoints rather than unrecorded paper transactions, the information extraction parameter is structurally maximized ($\chi_t \to 1$), rendering the revenue stream highly resilient to local state corruption or informal evasion.

\subsubsection*{Bypassing State Capture via Direct Algorithmic Yield Distribution}
To ensure these revenues are not intercepted by the captured state core or co-opted local elites, the collected assets do not sit in a centralized treasury. The X-Chain infrastructure uses an automated Algorithmic Yield Distribution Protocol. The tokens generated at the chokepoints are instantly swept into automated smart contracts that act as an open, decentralized ledger. Instead of funding bureaucratic state ministries, the fund is programmed to execute direct, oracle-verified allocations:

\begin{equation}
\mathcal{I}_k\bigl(t, X_{F,t}^{g,\mathrm{crypto}}, \mu_t^g\bigr) = \omega_k\bigl(e_{\alpha}(\mu_t^g)\bigr) \cdot X_{F,t}^{g,\mathrm{crypto}} \cdot \mathbf{1}_{\{\text{ZK-Proof Valid}\}},
\label{eq:direct_distribution}
\end{equation}

where $\omega_k(\cdot)$ is a weight function that dynamically matches the local allocation to the $\alpha$-expectile of the community’s destitution index. If a localized area experiences a severe climatic shock (tracked via satellite vegetation oracles $E_t^{g,3}$), the system automatically increases its allocation weight $\omega_k$. 
The funds are disbursed via secure Zero-Knowledge (ZK) proofs directly to decentralized, mobile-accessible community-led development boards and biometric individual civilian wallets. This process finances localized infrastructure (solar-powered water pumps, community micro-granaries) entirely outside the control of administrative state actors.

\subsubsection*{Strategic Autonomy and Ergodic Realization}
By replacing the unachievable dream of standard tax enforcement with automated, frictionless extraction at the economy's physical bottlenecks, this architecture transforms post-war financing into a highly resilient reality. It successfully drains liquid capital from the upper tails of the war entrepreneur class who must purchase fuel and refine gold to maintain their operations and algorithms-routes it to stabilize the lower tail of the population measure $\mu_t^g$. 
This architecture forces the effective structural spectral radius strictly below unity ($\rho(\mathbf{A}_{\mathrm{policy}}) < 1$), making the peaceful equilibrium globally stable and highly achievable, even within the constraints of a completely informal and fragile Sahelian environment.

\subsection{Field-Level Execution Dynamics: From Game-Theoretic Synthesis to Real-World Implementation}
\label{sec:field_implementation}

To transform the theoretical equilibrium of the Intergenerational Volterra Mean-Field-Type Game (MFTG) into a sustainable, real-world stabilization framework for Mali, we must systematically confront and resolve the operational limitations of the Sahelian theater. This section synthesizes the core frictions of the post-war reconstruction economy, extreme informality, institutional co-optation, and the physical limits of distributed ledger deployment and structures a self-enforcing incentive mechanism on the Graph-Chain (X-Chain) architecture that aligns the payoffs of all strategic participants with the long-term survival of the peace plan.

\subsubsection*{Systemic Limitations Framework}
A rigorous bridge from infinite-dimensional measure theory to field deployment requires isolating three distinct structural friction layers that undermine conventional, centralized top-down interventions:

\begin{itemize}
    \item \textbf{Asymmetric Demobilization Payoffs:} Traditional Disarmament, Demobilization, and Reintegration (DDR) frameworks consistently collapse because the static, lump-sum cash payments offered to combatants are economically inferior to the continuous, risk-adjusted returns generated by the informal war economy (e.g., protection rackets, artisanal gold extraction, and cross-border trafficking circuits). Under severe informational opacity, armed factions routinely exploit formal stipends as a liquid resource to purchase advanced weaponry while maintaining active insurgent posture.
    \item \textbf{Informal Tax Parameter Collapse ($\chi_t \to 0$):} With over 80\% of the Malian economic baseline operating outside formal registration, direct transactional audits and standard income taxation are fundamentally unenforceable. Rigid state attempts to manually extract rents from artisanal pits (\emph{placons}) merely increase localized transactional friction, paradoxically driving vulnerable civilians into the security and alternative governance structures of insurgent protection networks.
    \item \textbf{Blockchain Oracle Hardening Dilemma:} While the X-Chain network provides an un-copyable distributed ledger ledger, it remains vulnerable to the ``garbage-in, garbage-out" reality of data entry points. In low-connectivity environments, unrefined physical gold can be mixed, falsified, or physically smuggled across porous borders before reaching a digital logging node, allowing corrupt intermediaries to structurally launder conflict-linked assets into certified clean tokens.
\end{itemize}

To formalize these real-world constraints within our structure, we map the structural degradation of the control parameters in Table \ref{tab:field_frictions}:

\begin{table}[htbp]
\centering
\caption{Operational Friction Mapping: Theory vs. Sahelian Field Reality}
\label{tab:field_frictions}
\begin{tabular}{p{3cm}p{3cm}p{3cm}p{3cm}}
\toprule
\textbf{Control Primitive} & \textbf{MFTG Theoretical Request} & \textbf{Sahelian Operational Reality} & \textbf{Systemic Failure Mode} \\
\midrule
Offshore Tax $\theta_{\mathrm{tax}}$ & Complete transactional visibility ($\chi_t \to 1$) & \emph{Hawala} networks, cash dominance & $\theta_{\mathrm{effective}} \equiv 0$; war loops persist \\
Mediation Subsidy $\sigma_{\mathrm{sub}}$ & Homogeneous capital distribution & Local extortion, predatory capture & Diverted to $X^{\mathrm{global}}$; fuels violence \\
Ledger Consensus $\mathcal{G}_t^g$ & Continuous cryptographic validation & Power grid deficits, broken network nodes & Decoupled oracles, consensus delay \\
\bottomrule
\end{tabular}
\end{table}

\subsubsection*{Mechanism Incentive Design Matrix}
To guarantee that the system converges to the subcritical peaceful steady state ($\rho < 1$) without relying on the coercive enforcement of a fragile state core, the X-Chain ledger hardcodes a self-enforcing payoff matrix. This protocol ensures that the dominant strategy for every heterogeneous actor class is to adhere to the parameters of the reconstruction plan.

\subsubsection{The Civilian Baseline: Direct Algorithmic Yield Drops}
For the non-combatant population, survival is governed by acute loss aversion under volatile environmental shocks. The X-Chain bypasses administrative channels and distributes direct, oracle-verified mobile cash allocations (\emph{airdrops}) directly into biometric-mapped individual digital wallets. These payouts are funded entirely by automated transaction tariffs levied at the economy's non-bypassable physical chokepoints: wholesale fuel imports ($\theta_{\mathrm{fuel}}$) and telecommunication network airtime gateways ($\theta_{\mathrm{mobile}}$). Civilians receive immediate financial yields that scale upward when local security stabilizes, destroying the economic necessity of participating in insurgent recruitment loops and forcing the population's tail-risk expectile $e_{\beta_i}(\mu_t^g)$ below critical escalation thresholds.

\subsubsection{Middle-Tier Commanders: Conditional Intergenerational Trust Annuities}
Rather than demanding immediate, uncompensated demobilization, which threatens the survival utility of armed factions, the framework transforms local commanders into verified ecological and infrastructure guardians. The X-Chain establishes an automated smart-contract vault that streams a continuous, long-term financial annuity directly to the commander's node:

\begin{equation}
\mathcal{A}_{i,t}^g = \bar{\mathcal{A}}_i \cdot \mathbf{1}_{\{\Delta \mathcal{M}_t^{\mathrm{kinetic}}(\Omega_i) = 0\}} \cdot \mathcal{K}_{\mathrm{lock}}\bigl(T_{g+1}\bigr),
\label{eq:annuity_logic}
\end{equation}

where $\bar{\mathcal{A}}_i$ is the baseline security dividend and $\mathcal{M}_t^{\mathrm{kinetic}}(\Omega_i)$ is an automated, oracle-driven conflict-event mapping data feed tracking kinetic signatures within the commander's assigned geographic sector $\Omega_i$. If an un-stabilized violent event occurs, the oracle mismatch forces the indicator function to zero ($\mathbf{1}_{\{ .\}}= 0$), freezing the capital stream instantly.  A significant percentage of this annuity is locked in an intergenerational time-lock shard $\mathcal{K}_{\mathrm{lock}}$ that can only be decrypted and inherited by the commander's offspring if the sector maintains a multi-year peaceful baseline. This mechanism transforms a high-risk predatory extraction strategy into a low-risk, long-term asset-preservation strategy across generations.

\subsubsection*{International Buyers and Wholesalers: Premium Liquidity Certification}
To incentivize international commodity traders and refiners to comply with blockchain traceability within an informal economy, the X-Chain couples physical gold tracking with localized chemical markers and tokenized letters of credit. Refiners who execute transactions via the decentralized tracking ledger receive immediate access to institutional trade finance liquidity pools and a premium legal pricing tier for certified clean gold. 
Compliance is thus transformed from an expensive administrative burden into an un-bypassable optimization tool for market expansion and global regulatory risk mitigation.

\subsection*{Concrete Field Deployment Framework}
To transition this mechanism design from game-theoretic synthesis to a suitable field manual, the deployment operations are partitioned across explicit, achievable vectors detailed in Table \ref{tab:deployment_matrix}.

\begin{table}[htbp]
\centering
\caption{X-Chain Operational Field Manual and Infrastructure Layout}
\label{tab:deployment_matrix}
\resizebox{\textwidth}{!}{%
\begin{tabular}{p{3cm}p{3cm}p{3cm}p{3cm}}
\toprule
\textbf{Execution Layer} & \textbf{Concrete Field Action} & \textbf{Monitoring \& Verification Oracle} & \textbf{Institutional Safeguard} \\
\midrule
\textbf{Frictionless Revenue Extraction} & Integrate automated micro-surcharges on wholesale diesel imports at primary border hubs and a digital stamp fee on mobile telecom airtime transfers. & Programmatic APIs embedded directly into wholesale fuel customs manifests and telecom billing software layers. & Bypasses the central treasury completely; revenue is swept instantly into decentralized X-Chain contract vaults. \\
\addlinespace
\textbf{Oracle-Driven Reconstruction} & Disburse community-level capital funding for localized micro-granaries and solar water pumps based on regional destitution indexes. & Multi-spectral satellite verification arrays coupled with localized, decentralized IoT sensor feeds tracking asset activation. & Funding is script-locked and released in tranches based on automated asset validation, neutralizing local corruption vectors. \\
\addlinespace
\textbf{Intergenerational Asset Protection} & Split residual reconstruction funds at generational boundaries into time-locked cryptographic shards. & Algorithmic network block height metrics tied to epoch timers. & Cryptographically blocks contemporary actors from accessing or depleting resources allocated to future cohorts. \\
\bottomrule
\end{tabular}%
}
\end{table}

By replacing the unreliable assumption of state institutional benevolence with an automated, cryptographic protocol that directly shapes individual payoff functions, the system achieves operational equilibrium. Every strategic actor whether a vulnerable civilian, an opportunistic local commander, or an international commodity trader, discovers that their optimal choice within the game is to preserve the integrity of the peace framework, forcing the system's spectral radius below unity ($\rho(\mathbf{A}_{\mathrm{policy}}) < 1$) and permanently breaking the multi-generational conflict trap.

\section{Conclusion}
\label{sec:conclusion}

We have presented a fully integrated model of polycentric conflict-trap Risk in Mali as an intergenerational mean‑field‑type game on the space of joint probability measures.  The framework incorporates all salient features observed in the field: a large finite population of heterogeneous actors, a co‑evolving macro‑environment, endogenous coalition formation and switching, long‑memory grievance dynamics via fractional Volterra kernels, and a  catalogue of noises (Brownian, fractional, Gauss‑Volterra, Rosenblatt, Poisson, and regime‑switching) that capture the long‑range dependence, non‑Gaussianity, and abrupt structural shifts of the Sahelian conflict.
By working directly with the measure flow and avoiding dynamic programming, we derived the exact characterisation of a piecewise sequential Nash equilibrium via a system of forward‑backward stochastic differential equations with trans‑generational boundary conditions.  We proved the existence and uniqueness of this equilibrium under a mild spectral condition.  We identified a sharp threshold, expressed as a spectral radius condition involving the intergenerational trauma transmission and the efficiency of offshore capital laundering, beyond which the peaceful steady state becomes structurally unstable, guaranteeing an irreversible slide into a multi‑generational war trap.
Based on this diagnosis, we designed an optimal institutional policy that combines a tax on offshore financial flows with a subsidy for endogenous mediation.  We proved that when these two instruments are calibrated to exceed the critical thresholds, the war trap collapses and the system converges globally to a stable peace.  The model thus provides a principled foundation for quantitative calibration using detailed Sahelian conflict data and offers concrete,  justified policy recommendations for breaking the cycle of civil war.

\bibliographystyle{plain}

\end{document}